\documentclass[twocolumn,superscriptaddress,prx,amsmath,amstex,amssymb,citeautoscript,nofootinbib]{revtex4-2}
\usepackage[english]{babel}
\usepackage{letltxmacro}
\usepackage{latexsym}
\usepackage[utf8]{inputenc}
\LetLtxMacro{\ORIGselectlanguage}{\selectlanguage}
\makeatletter
\DeclareRobustCommand{\selectlanguage}[1]{%
  \@ifundefined{alias@\string#1}
    {\ORIGselectlanguage{#1}}
    {\begingroup\edef\x{\endgroup
       \noexpand\ORIGselectlanguage{\@nameuse{alias@#1}}}\x}%
}
\newcommand{\definelanguagealias}[2]{%
  \@namedef{alias@#1}{#2}%
}
\makeatother
\definelanguagealias{en}{english}
\definelanguagealias{English}{english}
\usepackage{graphicx}
\usepackage{amsmath}
\usepackage{amsfonts}
\usepackage{amsthm}   %
\usepackage{amssymb}
\usepackage{bm}
\usepackage{color}
\usepackage[percent]{overpic}
\usepackage{soul} 
\usepackage{amssymb}
\usepackage{wasysym}
\usepackage{dsfont}
\usepackage{float}
\usepackage{mathtools}
\usepackage{enumitem}
\usepackage{dsfont}
\usepackage{multirow}
\usepackage{colortbl}
\usepackage{mathrsfs}
\usepackage{makecell}
\usepackage{physics}

\usepackage{hyperref}
\hypersetup{
    bookmarks=false,         %
    unicode=false,          %
    pdftoolbar=false,        %
    pdfmenubar=true,        %
    pdffitwindow=false,     %
    pdfstartview={FitH},    %
    pdftitle={},    %
    pdfauthor={Authors},     %
    pdfsubject={},   %
    pdfcreator={},   %
    pdfproducer={}, %
    pdfkeywords={many-body localization} {matrix elements} {disordered systems}, %
    pdfnewwindow=true,      %
    colorlinks=true,       %
    linkcolor=blue,          %
    citecolor=blue,        %
    filecolor=magenta,      %
    urlcolor=black           %
}

\theoremstyle{definition}
\newtheorem{definition}{Definition}[section]
\newtheorem{theorem}{Theorem}[section]
\newtheorem{lemma}{Lemma}[section]
\newtheorem{corollary}{Corollary}[section]
\newtheorem{proposition}{Proposition}[section]

\newcommand{\be}{\begin{equation}}
\newcommand{\ee}{\end{equation}}
\newcommand{\bea}{\begin{eqnarray}}
\newcommand{\eea}{\end{eqnarray}}

\newcommand{\rAngle}{\rangle \! \rangle }
\newcommand{\lAngle}{\langle \! \langle }

\renewcommand{\bar}[1]{\overline{ #1}}

\newcommand{\ri}{\mathrm{i}}
\newcommand{\rd}{\partial}

\usepackage{graphicx}
\usepackage[colorinlistoftodos]{todonotes}
\usepackage{verbatim}

\newcommand{\bbrakket}[2]{\mbox{$ \langle\!\langle #1 | #2 \rangle\!\rangle $}}

\newcommand{\kket}[1]{\mbox{$| #1 \rangle\!\rangle$}}
\newcommand{\bbra}[1]{\mbox{$\langle\!\langle #1 |$}}

\newcommand{\bbA}{\mathbb{A}}
\newcommand{\bbC}{\mathbb{C}}
\newcommand{\bbD}{\mathbb{D}}

\newcommand{\bbH}{\mathbb{H}}

\newcommand{\bbZ}{\mathbb{Z}}
\newcommand{\bbG}{\mathbb{G}}
\newcommand{\bbT}{\mathbb{T}}

\newcommand{\calB}{\mathcal{B}}
\newcommand{\calC}{\mathcal{C}}

\newcommand{\calE}{\mathcal{E}}
\newcommand{\calF}{\mathcal{F}}

\newcommand{\calH}{\mathcal{H}}

\newcommand{\calL}{\mathcal{L}}
\newcommand{\calM}{\mathcal{M}}
\newcommand{\calN}{\mathcal{N}}
\newcommand{\calO}{\mathcal{O}}

\newcommand{\calR}{\mathcal{R}}

\newcommand{\calZ}{\mathcal{Z}}

\newcommand{\cH}{\mathcal{H}}

\newcommand{\U}{\mathrm{U}}

\newcommand{\rw}{\mathrm{w}}

\newcommand{\sfg}{\mathsf{g}}
\newcommand{\sfh}{\mathsf{h}}

\newcommand{\sfk}{\mathsf{k}}

\newcommand{\appref}[1]{Appendix~\ref{#1}}
\newcommand{\eqnref}[1]{Eq.~\eqref{#1}}

\newcommand{\figref}[1]{Fig.~\ref{#1}}

\usepackage{tikz}
\usepackage{tikz-cd}
\usetikzlibrary{arrows}
\usetikzlibrary{intersections}
\usetikzlibrary{shapes.geometric}
\usetikzlibrary{decorations.pathmorphing, patterns,shapes}
\usetikzlibrary{decorations.markings}
\usetikzlibrary{patterns}
\usetikzlibrary{fit}

\tikzset{
	mid arrow/.style={postaction={decorate,decoration={
				markings,
				mark=at position .575 with {\arrow[#1]{stealth}}
	}}},
	near arrow/.style={postaction={decorate,decoration={
				markings,
				mark=at position .275 with {\arrow[#1]{stealth}}
	}}},
	far arrow/.style={postaction={decorate,decoration={
				markings,
				mark=at position .800 with {\arrow[#1]{stealth}}
	}}},
}
\pgfmathsetmacro\MathAxis{height("$\vcenter{}$")}

\usepackage{soul}
\makeatletter
\def\l@subsection#1#2{}
\def\l@subsubsection#1#2{}
\makeatother

\newtheorem*{lemma*}{Lemma}

\begin{document}

\title{Tensor network formulation of symmetry protected topological phases in mixed states}

\author{Hanyu Xue}
\affiliation{Yuanpei College, Peking University, Beijing 100871, People’s Republic of China}

\author{Jong Yeon Lee}
\affiliation{Department of Physics, University of California, Berkeley, California 94720, USA}
\affiliation{Department of Physics, University of Illinois at Urbana-Champaign, Urbana, Illinois 61801, USA}

\author{Yimu Bao}
\affiliation{Kavli Institute for Theoretical Physics, University of California, Santa Barbara, CA 93106, USA}

\begin{abstract}
We define and classify symmetry-protected topological (SPT) phases in mixed states based on the tensor network formulation of the density matrix. In one dimension, we introduce strong injective matrix product density operators (MPDO), which describe a broad class of short-range correlated mixed states, including the locally decohered SPT states. We map strong injective MPDO to a pure state in the doubled Hilbert space and define the SPT phases according to the cohomology class of the symmetry group in the doubled state. Although the doubled state exhibits an enlarged symmetry, the possible SPT phases are also constrained by the Hermiticity and the semi-positivity of the density matrix. We here obtain a complete classification of SPT phases with a direct product of strong $G$ and weak $K$ unitary symmetry given by the cohomology group $\mathcal{H}^2(G, \text{U}(1))\oplus\mathcal{H}^1(K, \mathcal{H}^1(G, \text{U}(1)))$. The SPT phases in our definition are preserved under symmetric local circuits consisting of non-degenerate channels. This motivates an alternative definition of SPT phases according to the equivalence class of mixed states under a ``one-way" connection using symmetric non-degenerate channels. In locally purifiable MPDO with strong symmetry, we prove that this alternative definition reproduces the cohomology classification. We further extend our results to two-dimensional mixed states described by strong semi-injective tensor network density operators and classify the possible SPT phases.
\end{abstract}

\maketitle

\section{Introduction}
Over the past decade, significant theoretical progress has been made to classify and understand symmetry-protected topological (SPT) orders in short-range entangled quantum ground states~\cite{chen2011classification,schuch2011classifying,chen2011complete,chen2013symmetry, pollmann2012symmetry,gu2009tensor,chen2012symmetry,senthil2015symmetry}.
Such states feature anomalous boundary modes, degenerate entanglement spectra~\cite{pollmann2010entanglement}, and can be detected using non-decaying string order parameters in one dimension and perimeter-law decaying nonlocal order parameters in higher dimensions~\cite{den1989preroughening,kennedy1992hidden,pollmann2012detection}. 
These features are robust under constant-depth local unitary circuit evolution unless the protecting symmetry is broken.

Symmetry-protected topological phases in mixed states, on the other hand, have been less explored outside the thermal Gibbs states~\cite{roberts2017symmetry}. 
Such a theory is needed in light of recent quantum simulation experiments to characterize SPT states prepared using non-equilibrium protocols, which are generally non-thermal mixed states due to environmental decoherence~\cite{de2019observation,sompet2022realizing,bluvstein2022quantum}. 
This raises pivotal questions about the resilience of the SPT order against local decoherence and the methodologies for defining and classifying mixed-state SPT phases.

Currently, SPT phases in mixed states have been investigated and defined using several approaches.
The first approach builds on various conventional diagnostics of SPT pure states.
Specifically, Ref.~\cite{mcginley2020fragility} characterizes the SPT mixed states according to the coherence of degenerate edge states and shows that the SPT phase protected by unitary symmetry is robust under decoherence. 
Ref.~\cite{de2022symmetry} uses non-decaying string order parameters in mixed states as a defining feature. 
Ref.~\cite{lee2022symmetry} defines SPT mixed states according to the anomaly in the doubled-state formulation of density matrix using the Choi-Jamiolkowski isomorphism~\cite{CHOI1975, JAMIOLKOWSKI1972}.
The second approach defines SPT phases according to the equivalence class of mixed states under local symmetric quantum channel circuits~\cite{ma2023average}. 
Here, one requires a ``two-way" connection between states in the same phase~\cite{coser2019classification,ma2023average,sang2023mixed}.
Specifically, two mixed states $\rho_1$ and $\rho_2$ belong to the same phase if there exist symmetric local quantum channels $\calN_{1,2}$ such that $\rho_2 = \calN_1[\rho_1]$ and $\rho_1 = \calN_2[\rho_2]$.
However, it remains unclear whether the definitions based on these two distinct approaches are compatible and produce the same classification of SPT mixed states.

Notably, the SPT phases in these definitions are for short-range correlated mixed states and are different from the mixed-state phases determined by the generalized separability criteria~\cite{hastings2011topological,chen2024separability,chen2023symmetry}.
Specifically, a symmetric mixed state is separable if written as an ensemble of pure states generated by symmetric short-depth unitary circuits~\cite{chen2023symmetry}.
The key difference is that a separable mixed state can possess long-range classical correlations; the mixed states in the same phase are related by local operations and classical communication, whose nonlocal nature can lead to a distinct classification of phases.

We here focus on the non-thermal mixed states with short-range correlation.
The issue of defining SPT phases in these states is partly because the notion of an energy gap is missing.
The energy gap is central for establishing equivalence classes of quantum ground states based on whether two states can be adiabatically connected without the gap closing.
Two states in the same equivalence class are then related by finite-time unitary evolution and share the same long-range defining features.
Moreover, the gap ensures short-range correlation in ground states, effectively ruling out the possibility of critical states or unconventional states given by a superposition of different SPT states.
Hence, a key for defining mixed-state phases is to generalize the notion of gap.

The tensor network formulation of SPT pure states introduces a concept equivalent to the energy gap, which hints at a possible extension to mixed states.
In particular, one-dimensional SPT pure states are described by matrix product states (MPS) consisting of local injective tensors~\cite{schuch2011classifying,cirac2021matrix}.
The injectivity condition plays the role of the energy gap and enforces short-range correlation.
In injective MPS, the SPT phases are characterized by the projective representation of the symmetry group on the virtual bond, which governs the defining features of SPT order.
The states in the same SPT phase can be always related by finite-time adiabatic evolution along a path of symmetric injective MPS, while states in different phases cannot be symmetrically connected unless the injectivity condition is violated.

In this work, we generalize the injectivity condition to tensor network formulation of density operators and define SPT phases in mixed states.
In one dimension, we introduce a broad class of short-range correlated mixed states described by matrix product density operators (MPDO) with strong injectivity conditions, which strictly generalize pure states described by injective MPS and include locally decohered SPT pure states studied in Ref.~\cite{de2022symmetry,ma2023average,lee2022symmetry} as special examples.
To define the SPT phases, we formulate the strong injective MPDO as a pure state in the doubled Hilbert space, described by an injective MPS; the mixed-state SPT phases are then defined according to the projective representation of the symmetry group in the doubled state.

We classify the mixed-state SPT phases protected by a direct product $G \times K$ of strong $G$ and weak $K$ symmetry.
In this case, the doubled state exhibits an enlarged symmetry $\bbG = (G\times G\times K)\rtimes \bbZ_2^\bbH$ as the strong symmetry can act on both ket and bra Hilbert space, and the state is further invariant under the Hermitian conjugate of the density matrix.
While the enlarged symmetry may suggest an enriched classification, the possible phases are also highly constrained because not every doubled state corresponds to a physical density matrix.
In particular, the possible SPT phases are constrained by the Hermiticity and semi-positivity of the density matrix.
In the case of unitary physical symmetry, we obtain a classification given by the cohomology group $\calH^2(G,\U(1))\oplus\calH^1(G,\calH^1(K,\U(1)))$, which is consistent with that of the average SPT phases~\cite{ma2023average}.
For the anti-unitary symmetry, such as time reversal, we show it cannot protect nontrivial SPT phases even if the strong time-reversal symmetry is preserved, which is expected due to its fragility pointed out in Ref.~\cite{mcginley2020fragility}.

The SPT phases in our definition are robust under non-degenerate symmetric local channel circuits, which preserves the strong injectivity condition.
This motivates an alternative definition of the SPT phases according to the equivalence classes of mixed states under ``one-way" connection; two states $\rho_{1,2}$ belong to the same phase if there exist non-degenerate local channel $\calN$ such that $\rho_1 = \calN[\rho_2]$ \emph{or} $\rho_2 =\calN[\rho_1]$.
We further prove that this definition reproduces the cohomology classification $\calH^2(G,\U(1))$ for locally purifiable MPDO with only strong symmetry $G$.

We also discuss the possibility of using string order parameters to probe SPT phases in strong injective MPDO.
We prove that the string order can fully characterize the mixed-state SPT phase if the phase is protected by strong Abelian symmetry or is characterized by the mixed anomaly between the strong and the weak symmetry.

Last, we extend our results to two dimensions and introduce strong semi-injective tensor network density operators (TNDO), which are mixed-state generalizations of semi-injective projective entangled pair states (PEPS) ~\cite{Molnar_2018}.
Within this framework, we show that the mixed-state SPT phases protected by $G\times K$ symmetry are classified by $\calH^3(G,\U(1))\oplus\calH^2(G,\calH^1(K,\U(1)))\oplus\calH^1(G,\calH^2(K,\U(1)))$.

The rest of the paper is organized as follows.
Section~\ref{sec:preliminaries} briefly reviews the concept of injective MPS.
Section~\ref{sec:definition} introduces the definition of MPDO with strong injectivity conditions.
Section~\ref{sec:def_spt_phases} defines the SPT phases in strong injective MPDO.
Section~\ref{sec:classification} classifies the SPT phases in the strong injective MPDO protected by both strong and weak symmetry.
Section~\ref{sec:purification} proposes a definition of SPT phases according to a ``one-way" connection using non-degenerate channels.
Section~\ref{sec:stringorder} proves the existence of string order parameters for various mixed-state SPT phases.
Section~\ref{sec:generalization} introduces strong injective TNDO in two dimensions and classifies the possible SPT phases.
We close with discussions in Sec.~\ref{sec:discussion}.

\tableofcontents

\section{Preliminaries: injective MPS}\label{sec:preliminaries}
We here review the concept of matrix product state (MPS) and the injectivity condition.
The properties of injective MPS are essential for understanding the statements regarding the strong injective MPDO introduced in this work.

Matrix product state is a theoretical tool that can describe the ground state of one-dimensional gapped local Hamiltonian, which exhibits area-law entanglement entropy~\cite{hastings2007area} and finite correlation length~\cite{hastings2004locality,hastings2006spectral}.
Specifically, the MPS takes the form
\begin{align}
    \ket{\Psi[A]} = \sum_{\{i_x = 1\}}^d \tr \left( \prod_{x = 1}^L A^{i_x}\right) \ket{\{i_x\}},
\end{align}
where the local tensor $A^{i_x}$ is a $D\times D$ matrix at site $x$ with $D$ being the bond dimension, and $i_x$ denotes the state in the physical Hilbert space of dimension $d$.

Representing 1d quantum state using MPS has redundancies.
One can always convert a translationally invariant MPS to a canonical form, in which each local tensor is decomposed into a direct sum of normal tensors in the virtual Hilbert space, i.e. $A^i = \bigoplus_\mu A^{[\mu],i}$~\cite{schuch2011classifying,cirac2021matrix}.
A tensor $B$ is \emph{normal} if the transfer matrix $\bbT = \sum_i B^i \otimes (B^*)^i$ has a unique largest eigenvalue, and the associated left and the right eigenvector represents a (strict)positive matrix.
An MPS $|\Psi[A] \rangle$ is normal if the direct sum decomposition of $A$ contains a single normal tensor.

The normal MPS is the unique ground state of gapped local Hamiltonian~\cite{perez2006matrix,cirac2021matrix} and can describe pure states with SPT order.
The unique largest eigenvalue of the transfer matrix for a normal MPS indicates that all the connected correlation functions are exponentially decaying, i.e. $\langle \calO_x\calO_{x'}\rangle -\langle \calO_x\rangle \langle \calO_{x'}\rangle \sim \exp(-|x-x'|/\xi)$, where $\langle \calO_x \rangle = \bra{\Psi}\calO_x\ket{\Psi}$, and $\xi$ is the correlation length determined by the second largest eigenvalue of the transfer matrix $\bbT$.
On the other hand, the symmetry-breaking ground state with long-ranged connected correlations is not a normal MPS. 

The injective MPS is a concept closely related to the normal MPS and is given by the following definition.
\begin{definition} \label{def:injectivity}
(Injective MPS) An MPS $\ket{\Psi[A]}$ is injective if the local tensor $A:\mathbb{C}^{D^2} \mapsto \mathbb{C}^d$ viewed as a mapping from the virtual Hilbert space $\mathbb{C}^{D^2}$ to the physical Hilbert space $\mathbb{C}^d$ is injective.
\end{definition}
An injective MPS is always normal\footnote{One can always connect an injective MPS to a fixed point injective MPS of isometric form without closing the gap of its parent Hamiltonian~\cite{schuch2011classifying}. The transfer matrix of an injective MPS therefore always has a unique largest eigenvalue same as its corresponding fixed point MPS. It is further shown that the eigenvector corresponding to the unique largest eigenvalue of the transfer matrix is positive~\cite{cirac2021matrix}.}, and a normal MPS can be always converted to an injective MPS by grouping a finite number of sites~\cite{sanz2010quantum,michalek2019quantum,rahaman2019new} (see Ref.~\cite{cirac2021matrix} for a review).
In this work, we do not distinguish the difference between these two concepts. 
We use mathematical theorems proved for normal MPS and assume the MPS is brought to the injective form by grouping neighboring sites.

\section{Strong injective matrix product density operator}\label{sec:definition}

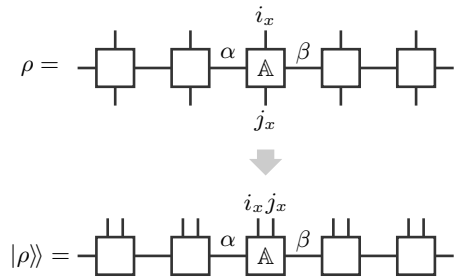
\begin{figure}[t!]
\centering
\begin{tikzpicture}
\definecolor{myred}{RGB}{240,83,90};
\definecolor{myblue}{RGB}{73,103,189};
\definecolor{myturquoise}{RGB}{83,195,189};
\draw[black!80, line width=1.0] (-0.5,0) -- (4.5,0);
\foreach \x in {0,1,...,4}{
\draw[black!80, line width=1.0] (\x,-0.5) -- (\x, 0.5);
\filldraw[fill=white, draw=black!80, line width=1.0] (\x-0.25,-0.25) rectangle (\x+0.25,0.25);
}
\node[black] at (2,0) {$\bbA$};
\node[black] at (2,0.7) {$i_x$};
\node[black] at (2,-0.7) {$j_x$};
\node[black] at (1.5,0.2) {$\alpha$};
\node[black] at (2.5,0.2) {$\beta$};
\node[black] at (-1.0,0) {$\rho = $};
\node[single arrow, fill=black!20,
      minimum width = 0.5, single arrow head extend=3,
      minimum height= 0.5,
      rotate=270] at (2,-1.2) {};

\draw[black!80, line width=1.0] (-0.5,-2.5) -- (4.5,-2.5);
\foreach \x in {0,1,...,4}{
\draw[black!80, line width=1.0] (\x+0.1,-2.5) -- (\x+0.1, -2.0);
\draw[black!80, line width=1.0] (\x-0.1,-2.5) -- (\x-0.1, -2.0);
\filldraw[fill=white, draw=black!80, line width=1.0] (\x-0.25,-2.75) rectangle (\x+0.25,-2.25);
}
\node[black] at (2,-2.5) {$\bbA$};
\node[black] at (1.85,-1.8) {$i_x$};
\node[black] at (2.15,-1.8) {$j_x$};
\node[black] at (1.5,-2.3) {$\alpha$};
\node[black] at (2.5,-2.3) {$\beta$};
\node[black] at (-1.0,-2.5) {$\kket{\rho} = $};
\end{tikzpicture}
\caption{ Matrix product density operator $\rho$ and its corresponding doubled state $\kket{\rho}$. For each local tensor $\bbA$, $i_x$ and $j_x$ denote the state in the ket and bra physical Hilbert space, respectively, and $\alpha$ and $\beta$ denote the state in the virtual Hilbert space.}
\label{fig:mpdo}
\end{figure}

In this section, we generalize the injective MPS (Def.~\ref{def:injectivity}), which describes SPT pure states, to tensor network formulation of mixed states.
To this end, we introduce matrix product density operators (MPDO) satisfying \emph{strong injectivity} conditions.
We show that the strong injectivity condition is preserved under non-degenerate local quantum channels, and the strong injective MPDO can describe a broad class of short-range correlated one-dimensional mixed states, including the locally decohered SPT states.

Consider a one-dimensional translationally invariant mixed state in the form of MPDO~\cite{cirac2017matrix}
\begin{align}
    \rho[\bbA] =  \sum_{\{i_x,j_x = 1\}}^d \tr \left(\prod_{x = 1}^L \bbA^{i_x j_x} \right) \ketbra{\{i_x\}}{\{j_x\}}.
\end{align}
where $\bbA^{i_xj_x}$ is the local tensor at site $x$, the physical indices $i_x$ and $j_x$ denote the state in the ket and bra local Hilbert space of $d$ dimension, respectively. 
Each $\bbA^{i_xj_x}$ is a $\chi \times \chi$ matrix, where $\chi$ is the bond dimension. 
In this work, we also take an equivalent formulation as a pure matrix product state (MPS) in the doubled Hilbert space when classifying mixed-state phases (see Fig.~\ref{fig:mpdo})~\cite{JAMIOLKOWSKI1972, CHOI1975}
\begin{align}
    \rho \mapsto \kket{\rho[\bbA]} :=& \sum_{\{i_x,j_x\}} \rho^{\{i_x,j_x\}} \ket{\{i_x\}}\ket{\{j_x\}}\nonumber \\
    =& \sum_{\{i_x,j_x\}} \tr\left(\prod_x \bbA^{i_xj_x}\right) \ket{\{i_x\}}\ket{\{j_x\}}.
\end{align}
Here, the $d^2$-dimensional local Hilbert space of the doubled state is denoted by the indices $i_x$ and $j_x$ together.

The structure of general MPDOs, however, remains less understood. 
Unlike MPS, we do not have a canonical form for decomposing MPDOs into ``normal'' MPDOs, each representing a legitimate physical density matrix. 
Although one could, in theory, transform the MPS representing a doubled state of the density matrix into a canonical form, it is unclear how the constraints of the physical density matrix manifest in each component of this decomposition, because not all components necessarily correspond to physical density matrices. 
Here, rather than broadly addressing MPDOs, we generalize the injectivity condition for MPS and focus on a restricted set of MPDOs with strong injectivity conditions.
\begin{definition}\label{def:stronginjective}
(Strong injective MPDO) The MPDO $\rho$ with local tensor $\bbA^{i_xj_x}$ is \emph{strong injective} if the following two conditions are satisfied:
\begin{enumerate}
    \item The mapping $\bbA:(\bbC^{\chi})^{\otimes 2} \mapsto \bbC^{d^2}$ from the virtual Hilbert space $(\bbC^{\chi})^{\otimes 2}$ to the physical Hilbert space $\bbC^{d^2}$ (ket and bra combined) is injective.
    \item The transfer matrix $\bbT=\sum_i \bbA^{ii}$ has a unique largest real eigenvalue.
\end{enumerate}
\end{definition}
The first condition is the injectivity condition for the MPDO $\rho$. 
The condition implies the connected correlation function in the doubled state $\kket{\rho}$ is exponentially decaying, i.e. $\langle\!\langle\calO_x\calO_{x'}\rangle\!\rangle - \langle\!\langle\calO_x\rangle\!\rangle\langle\!\langle\calO_{x'}\rangle\!\rangle \sim \exp(-\abs{x-x'}/\xi_2)$, where $\langle\!\langle\calO_x\rangle\!\rangle := \bbra{\rho}\calO_x\kket{\rho}/\bbrakket{\rho}{\rho}$, and $\xi_2$ is the decay length in the doubled Hilbert space.
The second condition is to ensure that the connected physical correlation function is always exponentially decaying, i.e. $\langle \calO_x \calO_{x'}\rangle - \langle \calO_x\rangle \langle \calO_{x'} \rangle \sim \exp(-|x - x'|/\xi)$, where $\langle \calO_x \rangle := \tr\rho\calO_x$.
It follows from the definition that the density matrix of the pure state described by an injective MPS is a strong injective MPDO.

The strong injectivity conditions for MPDO ensure short-range correlations similar to the injectivity condition for MPS.
In MPS, the injectivity excludes critical states, long-range correlated symmetry-breaking states (e.g. the Greenberger–Horne–Zeilinger state), as well as states given by a superposition of two states from distinct SPT phases\footnote{It is shown that such a state cannot be the ground state of a gapped local Hamiltonian~\cite{Levin_2020}.}. 
In addition, the strong injectivity conditions for MPDO also exclude a class of short-range correlated density matrices given by a classical mixture of two wave functions from different SPT phases, which exhibit characteristics of both and do not have a direct pure-state counterpart; such states do not have gapped local parent Hamtilonian in the doubled Hilbert space and are therefore not injective.

When MPDO admits a local purification, the first condition implies the second. 
The existence of local purification requires that $\rho = \tr_E \ketbra{\Psi_{QE}}$ with $\ket{\Psi_{QE}}$ being a pure state on the system $Q$ and the environment $E$, and $\ket{\Psi_{QE}}$ is described by an MPS
\begin{align}
    \ket{\Psi_{QE}} = \sum_{\{i_x = 1\}}^d \tr\left(\prod_{x = 1}^L \tilde{A}^{i_xk_x}\right) \ket{\{i_x\}}_Q \ket{\{k_x\}}_E,
\end{align}
where $k_x$ labels the physical state of the environment, and each local tensor $\tilde{A}$ is a purification of $\bbA$, i.e. $\bbA^{i_xj_x} = \sum_{k_x} \tilde{A}^{i_xk_x}\otimes (\tilde{A}^*)^{j_xk_x}$ as shown in Fig.~\ref{fig:purification}. Note that the bond dimension of $\tilde{A}$ is the square-root of the bond dimension of $\bbA$ to satisfy this equation. Under this premise, the injectivity condition of $\bbA$ implies that its local purification $\tilde{A}$ is an injective MPS. Then, the transfer matrix $\bbT$ of the MPDO is that of its purification $\ket{\Psi_{QE}}$, which has a unique largest eigenvalue because $\tilde{A}$ is injective.

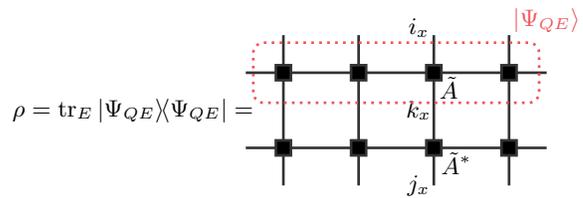
\begin{figure}
\centering
\begin{tikzpicture}
\definecolor{myred}{RGB}{240,83,90};
\definecolor{myblue}{RGB}{73,103,189};
\definecolor{myturquoise}{RGB}{83,195,189};
\draw[black!80, line width=1.0] (-0.5,0) -- (3.5,0);
\draw[black!80, line width=1.0] (-0.5,-1) -- (3.5,-1);
\foreach \x in {0,1,...,3}{
\draw[black!80, line width=1.0] (\x,-1.5) -- (\x, 0.5);
\filldraw[fill=black, draw=black!80, line width=1.0] (\x-0.1,-1.1) rectangle (\x+0.1,-0.9);
\filldraw[fill=black, draw=black!80, line width=1.0] (\x-0.1,-0.1) rectangle (\x+0.1,0.1);
}
\node[black] at (-2.0,-0.5) {$\rho = \tr_E\ketbra{\Psi_{QE}} = $};
\node[black] at (2.2,-0.2) {$\tilde{A}$};
\node[black] at (1.8,0.6) {$i_x$};
\node[black] at (1.8,-0.5) {$k_x$};
\node[black] at (2.3,-1.2) {$\tilde{A}^*$};
\node[black] at (1.8,-1.5) {$j_x$};
\draw[myred, dotted, rounded corners, line width=1.0] (-0.4,-0.4) rectangle (3.4,0.4); 
\node[myred] at (3.5, 0.7) {$\ket{\Psi_{QE}}$};
\end{tikzpicture}
\caption{ Local purification $\ket{\Psi_{QE}}$ of an MPDO. For an MPDO with local purification, its local tensor takes the form $\bbA^{i_xj_x} = \sum_{k_x} \tilde{A}^{i_xk_x} \otimes (\tilde{A}^*)^{j_xk_x}$, where $A$ is the local tensor in the purification, and $k_x$ labels the ancillary degrees of freedom in $E$.}
\label{fig:purification}
\end{figure}

Having introduced Def.~\ref{def:stronginjective}, it is now important to figure out what quantum operations preserve the strong injectivity conditions.
A generic quantum operation is a trace-preserving completely positive map, called \emph{quantum channel}. 
A quantum channel $\calN$ acting on a density matrix $\rho$ takes the form
\begin{align}
    \calN[\rho] := \sum_{i} K_{i}\rho K_{i}^\dagger,
\end{align}
where $K_i$ is the Kraus operator and satisfies $\sum_i K^\dagger_{i} K_{i} = \mathds{1}$ to preserve the trace.
A quantum channel is local if every Kraus operator acts on qubits in the vicinity of a certain site.

The strong injectivity conditions are preserved if the local quantum channel is \emph{non-degenerate}.
\begin{definition}\label{def:non-degenerate}
(Non-degenerate quantum channel) A quantum channel $\calN[\cdot]$ is non-degenerate if $\calN[\cdot]$ is an injective map when acting on the operator Hilbert space.
\end{definition}
\begin{proposition}\label{prop:local_circuit_injectivity}
Consider two MPDOs $\rho$ and $\rho' = \calN[\rho]$ related by local non-degenerate quantum channels in the form of a brickwork channel circuit $\calN$ in Fig.~\ref{fig:local_channel}(b). $\rho'$ is strong injective \emph{if and only if} $\rho$ is strong injective.
\end{proposition}
A non-degenerate quantum channel maps two distinct density matrices $\rho_1 \neq \rho_2$ to distinct resulting states, i.e. $\calN[\rho_1] \neq \calN[\rho_2]$~\footnote{We only require a non-degenerate channel to be invertible as a linear map. In most cases, it is not an invertible channel, i.e. its inverse is not a quantum channel.}.
The proposition is straightforward to prove for $\calN = \prod_x \calN_x$, where each $\calN_x$ is a non-degenerate onsite quantum channel acting only on qubit at site $x$ as in Fig.~\ref{fig:local_channel}(a). 
The resulting state $\rho' = \calN[\rho]$ is strong injective if and only if $\rho$ is strong injective because the non-degenerate onsite channel preserves the injectivity of local tensor $\bbA$ and does not change the transfer matrix $\bbT$.
In Appendix~\ref{app:injectivity_local_channel}, we provide the proof of Proposition~\ref{prop:local_circuit_injectivity} for brickwork local channel circuits.

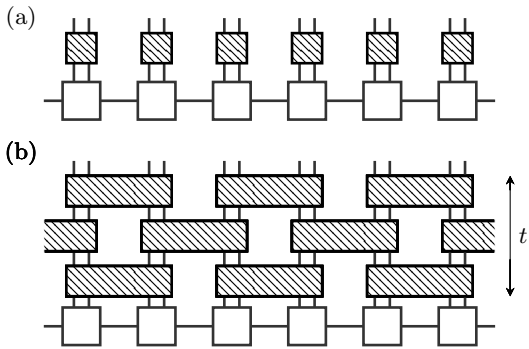
\begin{figure}
\centering
\begin{tikzpicture}
\definecolor{myred}{RGB}{240,83,90};
\definecolor{myblue}{RGB}{73,103,189};
\definecolor{myturquoise}{RGB}{83,195,189};
\draw[black!80, line width=1.0] (-0.5,0.5) -- (5.5,0.5);
\foreach \x in {0,1,...,5}{
\draw[black!80, line width=1.0] (\x+0.1,0.5) -- (\x+0.1, 1.0);
\draw[black!80, line width=1.0] (\x-0.1,0.5) -- (\x-0.1, 1.0);
\draw[black!80, line width=1.0] (\x+0.1,1.4) -- (\x+0.1, 1.6);
\draw[black!80, line width=1.0] (\x-0.1,1.4) -- (\x-0.1, 1.6);
\filldraw[fill=white, pattern=north west lines, pattern color=black, draw=black, line width=1] (\x-0.2,1.4) rectangle (\x+0.2,1.0);
\filldraw[fill=white, draw=black!80, line width=1.0] (\x-0.25,0.25) rectangle (\x+0.25,0.75);
}
\node[black] at (-0.8, 1.6) {(a)};
\draw[black!80, line width=1.0] (-0.5,-2.5) -- (5.5,-2.5);
\foreach \x in {0,1,...,5}{
\draw[black!80, line width=1.0] (\x+0.1,-2.5) -- (\x+0.1, -0.3);
\draw[black!80, line width=1.0] (\x-0.1,-2.5) -- (\x-0.1, -0.3);
\filldraw[fill=white, draw=black!80, line width=1.0] (\x-0.25,-2.75) rectangle (\x+0.25,-2.25);
\foreach \x in {0,2,4}{
\filldraw[fill=white] (\x-0.2,-1.7) rectangle (\x+1.2,-2.1);
\filldraw[pattern=north west lines, pattern color=black, draw=black, line width=1] (\x-0.2,-1.7) rectangle (\x+1.2,-2.1);
\filldraw[fill=white] (\x-0.2,-0.5) rectangle (\x+1.2,-0.9);
\filldraw[pattern=north west lines, pattern color=black, draw=black, line width=1] (\x-0.2,-0.5) rectangle (\x+1.2,-0.9);
}
\foreach \x in {-1,1,3,5}{
\filldraw[fill=white] (\x-0.2,-1.5) rectangle (\x+1.2,-1.1);
\filldraw[pattern=north west lines, pattern color=black, draw=black, line width=1] (\x-0.2,-1.5) rectangle (\x+1.2,-1.1);
}
\filldraw[fill=white, draw=white] (-1.3, -1.6) rectangle (-0.5, -0.7);
\filldraw[fill=white, draw=white] (5.5, -1.6) rectangle (6.3, -0.7);
\node[black] at (-0.8, -0.2) {(b)};
\draw[<->,>=stealth] (5.7,-2.1) -- (5.7,-1.3) node[right]{$t$} -- (5.7,-0.5);
}
\end{tikzpicture}
\caption{(a) Onsite quantum channel circuit. (b) Local quantum channel circuit of depth $t = 3$. The shaded block represents a quantum channel.}
\label{fig:local_channel}
\end{figure}

We remark that a non-degenerate local channel circuit describes the finite-time dynamics of an open system coupled to a Markovian bath.
The dynamics for a finite time $t$ governed by the Lindblad master equation~\cite{breuer2002theory}, $\rho \mapsto e^{\calL t}[\rho]$, always exists an inverse $e^{-\calL t}[\cdot]$~\cite{de2022symmetry}.
In contrast, examples of degenerate channels are the complete local dephasing, the complete depolarization channel, and tracing over a subsystem\footnote{As an example, the complete depolarization channel maps an arbitrary single-qubit density matrix $\rho_1$ to a maximally mixed state and therefore is degenerate, i.e. $\calN[\rho] = \rho_1/4 + X\rho_1 X/4 + Y\rho_1 Y/4 + Z\rho_1 Z/4 = \mathds{1}/2$, where $X$, $Y$, and $Z$ are Pauli matrices.}.
These channels, despite being local, cannot be realized by finite-time Lindblad evolution~\cite{de2022symmetry}.

Proposition~\ref{prop:local_circuit_injectivity} suggests that the strong injective MPDO can describe a broad class of short-range entangled mixed states studied in recent papers~\cite{de2022symmetry,lee2022symmetry}.
In particular, it includes the mixed states given by applying a non-degenerate local quantum channel circuit $\calN$ to a short-range entangled pure state $\ket{\Psi_0[A]}$ described by an injective MPS, i.e. $\rho = \calN[\ketbra{\Psi_0}{\Psi_0}]$.
We note that locally decohered injective MPS always has a local purification, and therefore the second condition of strong injectivity is implied from the first.
We illustrate the inclusion relation of the concepts in this section in Fig.~\ref{fig:strong_injective_MPDO}.

\begin{figure}
\centering
\begin{tikzpicture}
\draw[fill=black!30, rounded corners, line width=1.0] (-4,-1.625) rectangle (4, 3.6);
\draw[fill=black!20, rounded corners, line width=1.0] (-3.5,-1.5) rectangle (3.5, 2.5);
\draw[fill=black!10, rounded corners, line width=1.0] (-3,-1.375) rectangle (3, 1.4);
\draw[fill=white,rounded corners, line width=1.0] (-2.5,-1.25) rectangle (2.5, 0.3);
\node[black,align=center] at (0,-0.5) {Pure states $\ketbra{\Psi}$ \\ described by injective MPS $\ket{\Psi}$};
\node[black,align=center] at (0,0.8) {Locally decohered injective MPS \\ $\calN[\ketbra{\Psi}]$};
\node[black, align=center] at (0,1.9) {Locally purifiable \\ strong injective MPDO};
\node[black] at (0,3.2) {Strong injective MPDO};
\end{tikzpicture}
\caption{Strong injective MPDO describes a broad class of one-dimensional mixed states. It includes locally purifiable MPDO as a subset, which further includes the pure states, in the form of injective MPS, subject to local decoherence.}
\label{fig:strong_injective_MPDO}
\end{figure}
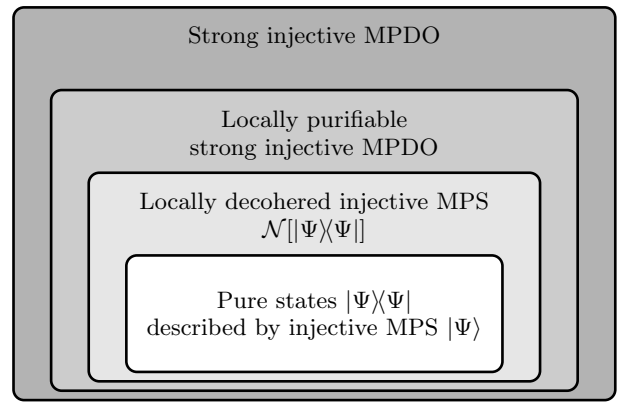

\section{SPT phases in strong injective MPDO}\label{sec:def_spt_phases}
In this section, we define the SPT phases in strong injective MPDO.
We first review the two distinct ways to impose the symmetry conditions on quantum channels and mixed states, i.e. strong and weak symmetry~\cite{buvca2012note}.
We then formulate the strong injective MPDO $\rho$ as a pure state $\kket{\rho}$ in the doubled Hilbert space, represented by an injective MPS, and define the SPT phases according to that in the corresponding pure doubled state.
The mixed-state SPT phases in our definition are preserved under the non-degenerate symmetric local quantum channels.

\subsection{Symmetry of quantum channels and mixed states}
Before discussing the SPT phases in mixed states, we specify the symmetry of quantum channels and the symmetry action on the density matrix.
We first consider a one-dimensional system with an on-site global symmetry $G$.
The symmetry acts as $U_g = \prod_x U_{g,x}$ for $g \in G$, where $U_{g,x}$ is a unitary representation of the symmetry group $G$ at site $x$.
In ground state problems, the symmetry requires the invariance of Hamiltonian under the symmetry transformation, i.e. $U_g H U_g^\dagger = H$; the symmetric state satisfies $U_g \ket{\Psi} = e^{\ri\theta} \ket{\Psi}$.

In open systems, the symmetry of the quantum channel can be implemented in two different ways~\cite{buvca2012note} (see also Ref.~\cite{albert2014symmetries,de2022symmetry}).
First, the quantum channel can exhibit strong symmetry requiring each Kraus operator to commute with the symmetry transformation (up to a phase independent of the Kraus operator)
\begin{align}
    U_g K_{i} U_g^\dagger = e^{\ri \theta_g} K_{i}.
\end{align}
In the Lindblad evolution, such a strong symmetry requires the invariance of the Hamiltonian and the Lindblad operators under the symmetry transformation $U_g$.

The definition of strong symmetry in mixed states stems from that of the quantum channel.
We call the mixed state \emph{strongly symmetric} if the density matrix $\rho$ is invariant under the symmetry action on the ket and the bra Hilbert space separately, i.e.
\begin{align}
    U_g \rho = e^{\ri \theta_g} \rho, \quad \rho U_g^\dagger = e^{-\ri\theta_g}\rho.
\end{align}
It is easy to see that for a symmetric pure state evolved under a strongly symmetric channel, the resulting mixed state has strong symmetry.

Alternatively, the quantum channel can exhibit weak symmetry, which only requires the quantum channel to commute with the symmetry action,
\begin{align}
    U_g \calN [U_g^\dagger (\,\cdot\,) U_g] U_g^\dagger = \calN[\,\cdot\,].
\end{align}
Accordingly, the weak symmetric density matrix is only invariant under the transformation on both ket and bra together, i.e. $U_g \rho U_g^\dagger = \rho$.

In this work, we consider the mixed state $\rho$ with a symmetry group $G\times K$ given by a direct product of strong symmetry $G$ and weak symmetry $K$.
The corresponding doubled state $\kket{\rho}$ exhibits an enlarged symmetry $\bbG = (G_l \times G_r \times K) \rtimes \bbZ_2^\bbH$.
Here, the left $G_l$ and the right $G_r$ symmetry stem from the $G$ symmetry that acts on the ket and bra Hilbert space, and $\bbZ_2^\bbH$ is an anti-unitary symmetry due to the Hermiticity of the density matrix. 
The hermitian conjugate $\bbH$ commutes with $K$ and swaps the two copies of $G$.

\subsection{Definition of SPT phases in strong injective MPDO}
We take the doubled state formulation of the strong injective MPDO and define its SPT phase according to the SPT phases of the injective MPS $\kket{\rho}$~\cite{schuch2011classifying}.
For an injective (or normal) MPS with an on-site symmetry, the fundamental theorem of MPS implies that each local tensor transforms as
\begin{equation} \label{eq:gauge_sym}
\begin{aligned}
    (U_g)^{ij} \bbA^{j} &= V_g^{-1} \bbA^i V_g, \text{ if $g$ is unitary,}\\
    (U_g)^{ij} \left(\bbA^{j}\right)^* &= V_g^{-1} \bbA^i V_g, \text{ if $g$ is anti-unitary,}
\end{aligned}
\end{equation}
where $V_g$ acts on the virtual Hilbert space and forms a projective representation of the symmetry group $\bbG$~\cite{perez2008string} (see also Ref.~\cite{cirac2021matrix} for a review).
Furthermore, the projective representation satisfies
\begin{equation} 
V_{g}\, {}^{g}V_{h} = \omega(g, h) V_{gh},
\label{eq:2-cocycle-proj}
\end{equation}
where the superscript $g$ acts on $V_h$, ${}^g V_h = V_h^* (V_h)$ if $g$ is anti-unitary (unitary), and $\omega(g,h)$ is a $\U(1)$ phase\footnote{For an injective MPS, the virtual Hilbert space should carry the same projective representation.}. For $g_{1,2,3} \in \bbG$, the phase factor satisfies
\begin{align} 
    \omega(g_1, g_2)\omega(g_1 g_2, g_3) = \omega(g_1, g_2 g_3){}^{g_1}\!\omega(g_2,g_3).\label{eq:2-cocycle}
\end{align}
The projective representation is only well-defined up to a $\U(1)$ phase $\varphi(g)$, i.e. $V_g \mapsto \varphi(g)V_g$.
Accordingly, $\omega(g,h)$ is defined up to $\omega(g,h) \mapsto \omega(g,h)\varphi(gh)/(\varphi(g){}^{g}\!\varphi(h))$.
The equivalence classes $[\omega]$ of $V_g$ under the above relations are given by the second cohomology group $\calH^{2}(\bbG, \U_T(1))$, which defines the SPT phase of $\kket{\rho}$.
Here, the $\bbG$-module $\U_T(1)$ is $\U(1)$ group, on which the element of group $\bbG$ acts as $g\cdot a = {}^{g}\!a$ for $g\in \bbG$ and $a \in \U_T(1)$.
The $\U(1)$ phase $\omega(g,h)$ is called 2-cocycle and is defined up to a coboundary term $\varphi(gh)/(\varphi(g){}^g\!\varphi(h))$.

The class $[\omega]$ is invariant under the smooth symmetric deformation of the injective MPS while preserving the injectivity condition.
We thus have the following statements.
\begin{proposition}\label{proposition:non-degenerate_preserve_SPT}
The strong injective MPDO $\rho$ and $\calN[\rho]$ are always in the same SPT phase if they are connected by a symmetric non-degenerate local quantum channel $\calN$.
\end{proposition}
\begin{corollary}
The strong injective MPDO $\rho$ in a non-trivial SPT phase cannot be prepared from a pure trivial product state using a symmetric non-degenerate local quantum channel.
\end{corollary}

Before proceeding, we remark that the classification given by $\calH^2(\bbG, \U_T(1))$ contains phases without a pure state correspondence. 
First, in the case with only strong unitary symmetry $G$, the doubled state exhibits the symmetry $\bbG = (G_l \times G_r)\rtimes \bbZ_2^\bbH$ leading to SPT phases classified by $\calH^2(\bbG, \U_T(1))$.
The classification is richer than $\calH^2(G, \U(1))$ for the pure-state SPT phases protected by $G$ symmetry.
Second, the weak symmetry $K$ in the symmetry group $\bbG$ of the double state naively can protect SPT phases by itself.
Such phases seem to contradict the statement in Ref.~\cite{de2022symmetry,ma2023average} that weak symmetry alone cannot protect SPT phases.
Indeed, as we show in Sec.~\ref{sec:classification}, these additional phases are unphysical as the classification $\calH^2(\bbG, \U_T(1))$ is further constrained by the Hermiticity and semi-positivity of the density matrix as well as the strong injectivity conditions of MPDO.

\section{Classification of SPT phases}\label{sec:classification}
In this section, we consider the mixed state $\rho$ with a direct product of strong symmetry $G$ and weak symmetry $K$.
We investigate the possible SPT phases in the associated doubled state $\kket{\rho}$ governed by the symmetry $\bbG = (G_l\times G_r\times K)\rtimes \bbZ_2^\bbH$.
The symmetry $\bbG$ consists of the physical symmetry $\bbG_p := G_l \times G_r \times K$ and the anti-unitary symmetry $\bbZ_2^\bbH$ due to the Hermiticity of density matrix.

The Hermiticity symmetry affects the possible SPT phases in two ways. 
It may protect the SPT phases, and it further constrains the projective representations of left and right symmetry.
In what follows, we show that the Hermiticity symmetry cannot on its own protect an SPT phase due to the strong injectivity condition.
We further show the absence of mixed anomaly between the Hermiticity and the physical symmetry.

The possible SPT phases are therefore only protected by the physical symmetry. 
We start with the case of unitary symmetry; the classification is given by $\calH^2(G_l\times G_r\times K, \U(1))$.
We consider the constraints on the possible SPT phases due to the hermiticity, the semi-positivity of the density matrix, and the strong injectivity conditions. We obtain a few key results.
First, in the case with only strong symmetry $G$, despite an enlarged symmetry $(G_l \times G_r)\rtimes \bbZ_2^\bbH$, the allowed SPT phases in the doubled state are classified by $\calH^2(G, \U(1))$.
Second, weak symmetry cannot by itself protect nontrivial SPT phases in mixed states, which is consistent with the results in Ref.~\cite{de2022symmetry,ma2023average}.
Third, the weak symmetry and strong symmetry together can protect SPT phases classified by $\calH^1(G, H^1(K,\U(1)))$~\cite{ma2023average}.
We close the section by commenting that anti-unitary symmetry, e.g. time-reversal symmetry, cannot on its own protect a non-trivial SPT phase even if the strong time-reversal symmetry is preserved.

We remark that the research problem regarding the physically allowed 1d SPT phases in the doubled state representing a mixed-density matrix first appears when studying the boundary phases of the 2d quantum double model~\cite{duivenvoorden2017entanglement}.
In the doubled state with Abelian symmetry $G \times G$, which corresponds to the 1d mixed state with strong Abelian symmetry $G$, Ref.~\cite{duivenvoorden2017entanglement} proved that the semi-positivity of the density matrix prevents the mixed anomaly between the two copies of the $G$ symmetry.
As a key step to obtain our classification, we extend this proof to the mixed state with general non-Abelian strong symmetry in one dimension.

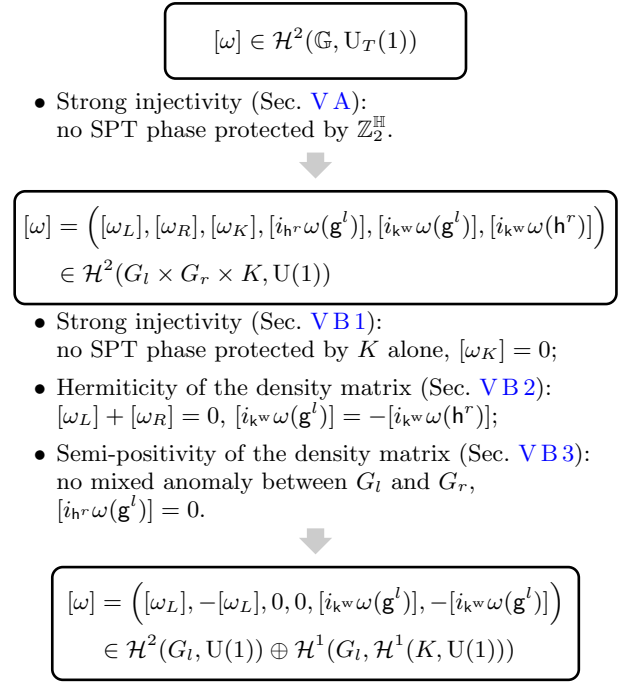
\begin{figure}
\centering
\begin{tikzpicture}
\definecolor{myred}{RGB}{240,83,90};
\definecolor{myblue}{RGB}{73,103,189};
\definecolor{myturquoise}{RGB}{83,195,189};
\draw[black,rounded corners,line width=1.0] (-2,4.5) rectangle (2,5.5) node[pos=.5]  {$[\omega] \in \calH^2(\bbG, \U_T(1))$};
\node[black,text width=\linewidth] at (0,4.0) { \begin{minipage}{\linewidth}
\begin{itemize}
\item Strong injectivity (Sec.~\ref{sec:no_hermiticity_SPT}):\\
no SPT phase protected by $\bbZ_2^\bbH$.
\end{itemize}\end{minipage}};
\node[single arrow, fill=black!20,
      minimum width = 0.5, single arrow head extend=3,
      minimum height= 0.5,
      rotate=270] at (0,3.4) {};
\draw[black,rounded corners,line width=1.0] (-4,1.5) rectangle (4,3.0) node[pos=.5,align=center]  {$\begin{aligned}[\omega] &= \left( [\omega_L], [\omega_R], [\omega_K], [i_{\sfh^r}\omega(\sfg^l)], [i_{\sfk^\rw}\omega(\sfg^l)], [i_{\sfk^\rw}\omega(\sfh^r)]\right) \\
&\in \calH^2(G_l\times G_r \times K, \U(1))\end{aligned}$};
\node[black,text width=\linewidth, text centered] at (0,-0) {\begin{minipage}{\linewidth}
\begin{itemize}
\item Strong injectivity (Sec.~\ref{sec:no_weak_SPT}): \\no SPT phase protected by $K$ alone, $[\omega_K] = 0$;
\item Hermiticity of the density matrix (Sec.~\ref{sec:hermiticity_constraint}): \\
$[\omega_L] + [\omega_R] = 0$, 
$[i_{\sfk^\rw}\omega(\sfg^l)]=-[i_{\sfk^\rw}\omega(\sfh^r)]$;
\item Semi-positivity of the density matrix (Sec.~\ref{sec:positivity_constraint}):\\
no mixed anomaly between $G_l$ and $G_r$, \\
$[i_{\sfh^r}\omega(\sfg^l)]=0$.
\end{itemize}
\end{minipage}};
\node[single arrow, fill=black!20,
      minimum width = 0.5, single arrow head extend=3,
      minimum height= 0.5,
      rotate=270] at (0,-1.6) {};
\draw[black,rounded corners,line width=1.0] (-3.5,-2) rectangle (3.5,-3.5) node[pos=.5,align=center]  {$\begin{aligned}\relax    [\omega] &= \left([\omega_L], -[\omega_L], 0, 0, [i_{\sfk^\rw}\omega(\sfg^l)], -[i_{\sfk^\rw}\omega(\sfg^l)]\right) \\
&\in \calH^2(G_l, \U(1)) \oplus \calH^1(G_l,\calH^1(K, \U(1)))\end{aligned}$};
\end{tikzpicture}
\caption{Classification of 1D mixed state SPT phases protected by the strong symmetry $G$ and the weak symmetry $K$. Starting from the SPT phases in the doubled state $\kket{\rho}$ labeled by $[\omega] \in \calH^2(\bbG, \U_T(1))$, where $\bbG = (G_l\times G_r \times K)\rtimes \bbZ_2^\bbH$, we show in this section the strong injectivity of the MPDO, the hermiticity and the semi-positivity of the density matrix constrain the possible SPT phases.}
\label{fig:summary_classification}
\end{figure}

\subsection{Absence of SPT phase protected by the Hermiticity symmetry}\label{sec:no_hermiticity_SPT}
The Hermiticity symmetry in the doubled state $\kket{\rho}$ is $\bbZ_2$ anti-unitary symmetry, which can in principle protect an SPT phase in pure state~\cite{pollmann2012symmetry}.
However, we here show that such an SPT phase cannot exist in the strong injective MPDO $\rho$ because the transfer matrix $\bbT$ has a unique largest eigenvalue.

The hermiticity symmetry of the strong injective doubled state requires
\begin{align}
    \bbA^{ij} \xrightarrow{\bbH} \left(\bbA^{ji}\right)^* = V_\bbH^{-1} \bbA^{ij} \,V_\bbH,
\end{align}
where the projective representation $V_\bbH$ satisfies $V_\bbH V^*_\bbH = V^*_\bbH V_\bbH = \pm \mathds{1}$~\cite{chen2011classification}.
The $+$ and $-$ sign corresponds to the trivial and the SPT phase, respectively.
Hence, the transfer matrix $\bbT = \sum_i \bbA^{ii}$ of the MPDO satisfies $\bbT^* = V_\bbH^{-1}\bbT  V_\bbH$.

The fact that $\bbT$ has a unique largest (real) eigenvalue prevents the SPT phase in the doubled state protected by the Hermiticity symmetry.
Let $\lambda_1$ be the largest eigenvalue of $\bbT$ and $|{\phi}_r)$ being the corresponding right eigenvector, i.e. $\bbT|{\phi}_r) = \lambda_1 |{\phi}_r)$. 
We then have $\bbT^*|{\phi_r^*}) = V_\bbH^{-1}\bbT V_\bbH|{\phi_r^*}) = \lambda_1 |{\phi_r^*})$, indicating $V_\bbH|{\phi_r^*})$ is also an eigenvector of $\bbT$ with eigenvalue $\lambda_1$. 
The uniqueness of the eigenvalue $\lambda_1$ requires $V_\bbH|{\phi_r^*}) = c|{\phi_r})$ with $c$ being a constant.
Thus, $|{\phi_r^*})$ is an eigenvector of $V_\bbH^* V_\bbH$ with a positive eigenvalue, i.e. $V_\bbH^* V_\bbH|{\phi_r^*}) = \abs{c}^2|{\phi_r^*})$, which implies $V_\bbH^* V_\bbH = V_\bbH V_\bbH^* = \mathds{1}$.
Physically, this result indicates that anti-unitary symmetry in strong injective MPDO cannot fractionalize and give rise to a Kramers doublet at the boundary.

As shown in Appendix~\ref{app:no_H_mixed_anomaly_1d}, the Hermiticity symmetry cannot exhibit a mixed anomaly with the physical symmetry. 
One can partially fix the gauge of the projective representation such that $\omega(\sfg, \bbH)\,{=}\,\omega(\bbH, \sfg)\,{=}\,\omega(\bbH,\bbH)\,{=}\,1$ for $\sfg\,{\in}\,\bbG_p$.
In what follows, we work with this gauge choice, and the Hermicity symmetry further imposes a constraint on the 2-cocycle
\begin{align}
    \omega(\sfg_1,\sfg_2) = \omega^*(\bar{\sfg}_1,\bar{\sfg}_2),\label{eq:2-cocycle-hermiticity_constraint}
\end{align}
where $\sfg_{1,2}\in\bbG_p$, and $\bar{\sfg}_{1,2} = \bbH \sfg_{1,2}\bbH$.

\subsection{SPT phase protected by physical symmetry}\label{sec:SPT_physical}
We now focus on the SPT phase protected by the physical symmetry $\bbG_p = G_l \times G_r \times K$.
In this subsection, we consider the case of unitary symmetry, and the possible SPT phases are given by the cohomology group $\calH^2(G_l \times G_r \times K, \U(1))$.
According to the K\"unneth formula~\cite{chen2014symmetry}, we have 
\begin{align}
&\calH^2(\bbG_p, \U(1)) \nonumber \\
=&\calH^2(G_l, \U(1))\oplus \calH^1(G_l, \calH^1(G_r, \U(1))) \oplus \calH^2(G_r,\U(1)) \nonumber \\
&\oplus \calH^1(G_l, \calH^1(K,\U(1))) \oplus \calH^1(G_r, \calH^1(K,\U(1)))\nonumber \\
&\oplus \calH^2(K, \U(1)).\label{eq:kunneth}
\end{align}
The first three terms in the decomposition characterize the SPT phases protected solely by the strong symmetry.
The next two terms describe the SPT phases featuring the mixed anomaly between the strong and the weak symmetry.
The last term is the SPT phase protected only by the weak symmetry.

The direct sum decomposition can be constructed from the decomposition of $2$-cochain $\omega: \bbG_p\times \bbG_p \mapsto \U(1)$ as shown in Appendix~\ref{app:direct-sum}.
In particular, the $2$-cochain acting on the subgroup $G_l$, $G_r$, and $K$ defines 
\begin{align}
    \omega_L(g_1, g_2) :=&\; \omega(\sfg_1^l, \sfg_2^l), \\
    \omega_R(h_1, h_2) :=&\; \omega(\sfh_1^r, \sfh_2^r), \\
    \omega_K(k_1, k_2) :=&\; \omega(\sfk_1^{\rw}, \sfk_2^{\rw}).
\end{align}
where we introduce the short-hand notations $\sfg^l := (g, e, e)$, $\sfh^r := (e, h, e)$, and $\sfk^{\mathrm{w}} := (e,e,k)$ with $e$ being the identity element.
Their cohomology classes $[\omega_L]$, $[\omega_R]$, and $[\omega_K]$ fully specify the elements in $\calH^2(G_l, \U(1))$, $\calH^2(G_r, \U(1))$, and $\calH^2(K,\U(1))$, respectively.
The three cross terms in the K\"unneth decomposition \eqnref{eq:kunneth} describe the mixed anomalies between symmetry groups and are characterized by the slant product
\begin{align}
    i_{\sfh^r}\omega(\sfg^l) :=& \frac{\omega(\sfg^l,\sfh^r)}{\omega(\sfh^r, \sfg^l)}, \\
    i_{\sfk^\rw}\omega(\sfg^l) :=& \frac{\omega(\sfg^l,\sfk^{\rw})}{\omega(\sfk^{\rw}, \sfg^l)},\\
    i_{\sfk^\rw}\omega(\sfh^r) :=& \frac{\omega(\sfh^r,\sfk^{\rw})}{\omega(\sfk^{\rw}, \sfh^r)},
\end{align}
where $ g \in G_l, h \in G_r, k \in K$ (see Appendix~\ref{app:direct-sum}).
The slant products above are 1-cocycle, and their equivalence classes belong to the cohomology group, e.g. $[i_{\sfh^r}\omega(\sfg^l)] \in \calH^1(G_l,\calH^1(G_r,\U(1)))$.
Thus, one can fully specify the element in $\calH^2(\bbG_p, U(1))$ using
\begin{align}
    [\omega] = \left( [\omega_L], [\omega_R], [\omega_K], [i_{\sfh^r}\omega(\sfg^l)], [i_{\sfk^\rw}\omega(\sfg^l)], [i_{\sfk^\rw}\omega(\sfh^r)]\right).\label{eq:6-tuple}
\end{align}
In what follows, we first show the absence of SPT phase protected by weak symmetry, i.e. $[\omega_K] = 0$, and then discuss how hermiticity and semi-positivity of the density matrix further constrain the class $[\omega]$.

\subsubsection{Absence of SPT phase protected by weak symmetry}\label{sec:no_weak_SPT}
Under the weak symmetry transformation $k \in K$, the local tensor of the MPDO transforms as
\begin{align}
    \sum_{i'j'} U_{\sfk^{\rw}}^{ij;i'j'}\bbA^{i'j'} = V_{\sfk^{\rw}}^\dagger \bbA^{ij} V_{\sfk^{\rw}}^{\vphantom{\dagger}},
\end{align}
where $U_{\sfk^{\rw}}^{ij;i'j'} = U_{k}^{ii'} (U_k^*)^{jj'}$.
Hence, the transfer matrix satisfies $\bbT = V_{\sfk^{\rw}}^\dagger \bbT V_{\sfk^{\rw}}^{\vphantom{\dagger}}$. 
Let $|\phi_r)$ be the right eigenvector of $\bbT$ associated with the unique largest eigenvalue. 
This indicates that $V_{\sfk^{\rw}}|\phi_r)$ is also the eigenvector with the largest eigenvalue, i.e. $V_{\sfk^{\rw}}|\phi_r) = \varphi(k)|\phi_r), \forall k \in K$.
Thus, the 2-cocycle takes the form $\omega_K(k,k') = \varphi(k)\varphi(k')/\varphi(kk')$, which is a 2-coboundary.
We then conclude that $[\omega_K] = 0$, and the weak symmetry cannot protect nontrivial SPT phases.

\subsubsection{Constraints from Hermiticity}\label{sec:hermiticity_constraint}
The Hermiticity symmetry in the doubled state constrains the 2-cocycle according to Eq.~\eqref{eq:2-cocycle-hermiticity_constraint} and therefore restricts the possible SPT phases protected by the physical symmetry.
First, taking $\sfg_1 = \sfg_1^l$ and $\sfg_2 = \sfg_2^l$ in Eq.~\eqref{eq:2-cocycle-hermiticity_constraint}, we have $\omega_L(g_1,g_2)\omega_R(g_2,g_1) = 1$ and therefore 
\begin{align}
    [\omega_L]+[\omega_R] =0.
\end{align}
Second, taking $\sfg_1 = \sfg^l$ and $\sfg_2 = \sfk^{\rw}$ in Eq.~\eqref{eq:2-cocycle-hermiticity_constraint}, we have $\omega(\sfg^l,\sfk^{\rw})\omega(\sfg^r,\sfk^{\rw}) = 1$. Similarly, $\omega(\sfk^{\rw},\sfg^l)\omega(\sfk^{\rw},\sfg^r) = 1$.
Hence, 
\begin{align}
    i_{\sfk^\rw}\omega(\sfg^l) = i_{\sfk^\rw}\omega(\sfg^r)^*,
\end{align}
which indicates $[i_{\sfk^\rw}\omega(\sfg^l)] = -[i_{\sfk^\rw}\omega(\sfg^r)]$.

\subsubsection{Constraints from semi-positivity}\label{sec:positivity_constraint}
The semi-positivity of the density matrix further constrains the SPT phases in the doubled state.
It excludes the SPT phases with the mixed anomaly between the left and the right $G$ symmetry.
Based on the decomposition of cocycle in Appendix~\ref{app:direct-sum}, the phase classified by $\calH^1(G_l,\calH^1(G_r,\U(1)))$ has a faithful representation in terms of the slant product.
In this case, one can show that the string order parameter in the doubled state fully characterizes such phases.
According to this understanding, we extend the proof in Ref.~\cite{duivenvoorden2017entanglement} to general non-Abelian symmetry group $G$ and show that the projective representations of the left and the right $G$ symmetry to commute with each other, i.e. no mixed anomaly.

The SPT phases in $\calH^1(G_l, \calH^1(G_r, \U(1)))$ are the condensate of decorated domain walls~\cite{chen2014symmetry}.
Under a right symmetry transformation of $\sfh^r$ on a subregion, the SPT state exhibits a non-trivial response at the boundary of the subregion.
The response transforms under the left symmetry $\sfg^l$ according to the slant product $i_{\sfh^r}\omega(\sfg^l)$.
Since the slant product is a one-dimensional representation of the left symmetry group, the response carries $G_l$-charge\footnote{This result follows from a simple modification of the derivation in Ref.~\cite{pollmann2012detection,duivenvoorden2017entanglement}.}.
This implies a string order in the doubled state.

In general, the charged operator decorated on two ends of the string operator supports on both the ket and the bra Hilbert space.
However, since only the ket Hilbert space carries the representation of the left symmetry, one can always find a string operator decorated by the charged operator only in the ket Hilbert space, which generically exhibits a finite overlap with the condensed object\footnote{In the space of operators charged under $G_l$, the chance for the charged operator $R_\alpha$ supported entirely on the right Hilbert space being zero has a zero measure.}. 
This string operator is of the form $\calO^{(2)}_{(\alpha,0), (e, g)}:= \calO_{\alpha,e} \otimes \calO^*_{0,g}$ and acquires a non-decaying expectation value,
\begin{align}
     \langle\!\langle \calO_{\alpha,e}\otimes \calO^*_{0, g}\rangle\!\rangle &:= \frac{\bbra{\rho}\calO_{\alpha,e}(x_1,x_2)\otimes \calO^*_{0,g}(x_1, x_2)\kket{\rho}}{\bbrakket{\rho}{\rho}} \nonumber \\
     &\xrightarrow{|x_1-x_2|\to\infty} \text{const.}
\end{align}
where $\calO_{\alpha, e}$ and $\calO_{0, g}$ are string operators acting on the ket and the bra Hilbert space, respectively.
The operator $\calO_{\alpha, g}$ in a single-copy Hilbert space takes the form
\begin{align}
    \calO_{\alpha,g}(x_1,x_2) = R_{\alpha^*}(x_1)\left(\prod_{x = x_1+1}^{x_2-1} U_{g}(x) \right) R_{\alpha}(x_2). \label{eq:string_order_parameter}
\end{align}
Here, $R_{\alpha}(x)$ carries charge $\alpha(\cdot)$ and transforms under the symmetry action as $U_{g'} R_{\alpha}(x) U_{g'}^\dagger = e^{\ri \alpha(g')} R_{\alpha}(x)$. 

The positivity of the density matrix requires other non-vanishing correlation functions in the doubled state. 
First, one can rewrite the string order parameter as 
\begin{align}
    \langle\!\langle \calO_{\alpha,e}\otimes \calO^*_{0, g}\rangle\!\rangle &= \frac{1}{\tr\rho^2} \tr \calO_{\alpha,e} \rho \calO^\dagger_{0,g}\rho \nonumber \\
    &= \frac{1}{\tr\rho^2} \tr \sqrt{\rho} \calO_{\alpha,e} \sqrt{\rho} \sqrt{\rho} \calO^\dagger_{0,g}\sqrt{\rho}
\end{align}
where $\sqrt{\rho}$ is well-defined because $\rho$ is positive.
Then, using the Cauchy-Schwarz inequality, we have
\begin{align} \label{eq:inequality}
    0 &< \left| \langle\!\langle\calO_{\alpha,e}\otimes \calO^*_{0,g}\rangle\!\rangle\right|^2 \nonumber \\
    &\leq \langle\!\langle\calO_{\alpha,e}\otimes \calO_{\alpha,e}^* \rangle\!\rangle \langle\!\langle\calO_{0,g}\otimes \calO_{0,g}^*\rangle\!\rangle.
\end{align}
This indicates the correlation function of charged objects is non-decaying\footnote{We note that this result only requires the semi-positivity of the density matrix and holds even if the density operator is not strongly injective.}, i.e. 
\begin{align}
\langle\!\langle\calO_{\alpha,e}\otimes \calO_{\alpha,e}^*\rangle\!\rangle &= \langle\!\langle R_{\alpha^*}(x_1)\!\otimes\! R_{\alpha^*}^*(x_1) \, R_{\alpha}(x_2)\!\otimes\! R_{\alpha}^*(x_2) \rangle\!\rangle \nonumber \\
&\xrightarrow{|x_1-x_2|\to\infty} \text{const.} >0.
\end{align} 
Since $\langle\!\langle R_\alpha(x)\otimes R^*_{\alpha}(x)\rangle\!\rangle = 0$, the connected correlation of $R_\alpha\otimes R_\alpha^*$ is long-range, which contradicts the injectivity condition of MPDO.
We then have $e^{\ri\alpha(\cdot)} = i_{\sfh^r}\omega(\cdot) = 1$, i.e. the slant product $i_{\sfh^r} \omega(\cdot)$ is a trivial representation, for all $h \in G_r$.
Thus, the projective representation of the left and the right symmetry commute.

The results thus far indicate that the possible SPT phases $[\omega] \in \calH^2(\bbG_p,\U(1))$ are given by
\begin{align}
    [\omega] = \left([\omega_L], -[\omega_L], 0, 0, [i_{\sfk^\rw}\omega(\sfg^l)], -[i_{\sfk^\rw}\omega(\sfg^l)]\right).
\end{align}
Hence, they are fully characterized by $[\omega_L] \in \calH^2(G, \U(1))$ and $[i_{\sfk^\rw}\omega(\sfg^l)] \in \calH^1(G, \calH^1(K,\U(1)))$.
In the system with only the strong symmetry, despite an enlarged symmetry $(G \times G) \rtimes \bbZ_2^\bbH$, the possible phases are only given by $\calH^2(G, \U(1))$.
In the system with also a weak symmetry $K$, although $K$ cannot on its own protect SPT phases, it can protect SPT phases together with the strong symmetry $G$, which are classified by $\calH^1(G, \calH^1(K,\U(1)))$.

\subsubsection{Time reversal symmetry}
So far, we have focused on the mixed-state SPT phases protected by unitary symmetry. A question arises regarding the possible SPT phase protected by anti-unitary symmetry, e.g. time-reversal symmetry.
In a seminal work, Ref.~\cite{mcginley2020fragility}  showed that the SPT phase with anti-unitary protecting symmetry is fragile against coupling to the environment even if the strong symmetry is preserved.
Specifically, the edge mode of a one-dimensional time-reversal SPT state rapidly decoheres and loses its encoded information when coupled to the environment.

The absence of SPT phases protected by an anti-unitary symmetry $\mathcal{T}$ in mixed states can be also understood in the doubled-state formulation.
The anti-unitary symmetry is ill-defined when acting on a subsystem.
Thus, despite being strongly symmetric, the doubled state does not possess individual left and right anti-unitary symmetry and only exhibits a global symmetry $\bbZ_2^\mathcal{T}\times \bbZ_2^\mathbb{H}$.
Since the $\bbZ_2$ anti-unitary symmetry in the doubled space can only take a trivial projective representation as shown in Sec.~\ref{sec:no_hermiticity_SPT}, the time-reversal symmetry cannot by itself protect non-trivial SPT phases in mixed states.

\section{One-way connection}\label{sec:purification}

The non-degenerate local quantum channel preserves the strong injectivity of MPDO and further preserves the SPT phases in our definition (Proposition~\ref{proposition:non-degenerate_preserve_SPT}) if the symmetry is imposed.
Such channel defines equivalence classes of mixed states; $\rho$ and $\rho'$ belong to the same class if there exists non-degenerate symmetric local channel $\calN$ such that $\rho = \calN[\rho']$ \emph{or} $\rho'=\calN[\rho]$.
Here, the equivalence relation only requires a ``one-way" connection, which is different from the ``two-way" connection based on generic symmetric local channels proposed in Ref.~\cite{ma2023average}.
It is natural to ask whether the equivalence class of mixed states based on this ``one-way" connection gives a compatible definition of SPT phases in mixed states.

In this section, we prove that in the case of strong unitary symmetry $G$, the equivalence class of mixed states under ``one-way" connection reproduces the cohomology classification if the strong injective MPDO has a local purification\footnote{We note that a local purification of MPDO with a finite bond dimension may not exist~\cite{de2013purifications}, even if the MPDO is translationally invariant~\cite{de2016fundamental}.}.
Proposition~\ref{proposition:non-degenerate_preserve_SPT} suggests that mixed states in different cohomology classes belong to different equivalence classes.
Here, we show that the mixed states in the same cohomology class are in a single equivalence class.
Specifically, we prove a theorem to show that an SPT mixed state in the cohomology class $[\omega_L]$, if admits a local purification, can always be realized by applying a non-degenerate strong symmetric quantum channel to an SPT pure state in the same class $[\omega_L]$ (illustrated in Fig.~\ref{fig:one-way_connection}).

\begin{theorem}\label{thm:purification}
Consider a one-dimensional quantum state with strong symmetry $G$. For a mixed state $\rho$ described by a strong injective MPDO, the following statements are equivalent:
\begin{enumerate}
    \item The state $\rho$ belongs to the mixed state SPT phase in class $[\omega]$ and has a local purification;
    \item The local purification $\ket{\Psi_{QE}}$ of $\rho$ is in the pure-state SPT phase in class $[\omega]$. Here, we require that $G$ acts trivially on the ancilla;
    \item There exists a strongly symmetric non-degenerate onsite quantum channel that maps an SPT pure state in class $[\omega]$ to $\rho$.
\end{enumerate}
\end{theorem}

\begin{proof}
To begin, we prove the second statement from the first. 
The strong symmetry condition of $\rho$ requires $U_g\rho = e^{\ri\theta}\rho$. Thus, $\bra{\Psi_{QE}}U_g \ket{\Psi_{QE}} = e^{\ri\theta}$; the purification $\ket{\Psi_{QE}}$ is symmetric under the symmetry $G$, i.e. $U_g \ket{\Psi_{QE}} = e^{\ri\theta} \ket{\Psi_{QE}}$.
Meanwhile, the on-site purification $\ket{\Psi_{QE}}$ is an injective MPS because each on-site tensor of $\rho$ obtained after tracing over the ancilla is still injective.
Hence, $\ket{\Psi_{QE}}$ is an SPT state.
The projective representation of the symmetry $G$ in $\ket{\Psi_{QE}}$ follows from that of the left symmetry in $\rho$.
Therefore, $\ket{\Psi_{QE}}$ is a pure state SPT in class $[\omega]$.

Next, we prove the third statement from the second.
Here, we show that given a local purification $\ket{\Psi_{QE}}$, we can explicitly construct a single-qubit quantum channel that maps an SPT state in class $[\omega]$ to the SPT mixed state $\rho$.
Because $A$ in $\ket{\Psi_{QE}}$ is injective tensor, we find a polar decomposition
\begin{align}
    A^{ik}_{\alpha\beta} = W^{ik}_{\gamma} B_{\gamma;\alpha\beta}, \;\; 
\begin{tikzpicture}[line cap=round,line join=round,x=1cm,y=1cm,baseline={(0, -\MathAxis pt)}]
\draw[black!80, line width=1.0] (-.5,0) -- (.5,0);
\draw[black!80, line width=1.0] (0.1,0) -- (0.1, 0.5);
\draw[black!80, line width=1.0] (-0.1,0) -- (-0.1,0.5);
\filldraw[fill=white, draw=black!80, line width=1.0] (-0.25,-0.25) rectangle (+0.25,0.25);
\node[black] at (0,0) {$A$};
\node[black] at (-0.7,0.) {$\alpha$};
\node[black] at (0.7,0.) {$\beta$};
\node[black] at (-0.2,0.65) {$i$};
\node[black] at (0.2,0.65) {$k$};
\end{tikzpicture}
=
\begin{tikzpicture}[line cap=round,line join=round,x=1cm,y=1cm,baseline={(0, -\MathAxis pt)}]
\draw[black!80, line width=1.0] (-.5,0) -- (.5,0);
\draw[black!80, line width=1.0] (0.1,0.6) -- (0.1, 1.25);
\draw[black!80, line width=1.0] (-0.1,0.6) -- (-0.1,1.25);
\draw[black!80, line width=1.0] (0,0) -- (0,0.6);
\filldraw[fill=white, draw=black!80, line width=1.0] (-0.25,-0.25) rectangle (+0.25,0.25);
\filldraw[fill=white, draw=black!80, line width=1.0] (-0.25,0.6) rectangle (+0.25,1.1);
\node[black] at (0,0.0) {$B$};
\node[black] at (-0.7,0.) {$\alpha$};
\node[black] at (0.7,0.) {$\beta$};
\node[black] at (0,0.85) {$W$};
\node[black] at (0.3,0.4) {$\gamma$};
\node[black] at (-0.3,1.35) {$i$};
\node[black] at (0.3,1.35) {$k$};
\end{tikzpicture}
\end{align}
where $B>0$ is a positive square matrix when viewed as a map from the $\chi^2$-dimensional virtual space to the $\chi^2$ dimensional physical space labeled by $\gamma$, and $W$ is an isometric embedding 
satisfying $W^\dagger W = \mathds{1}$.

The positive matrix $B$ is an injective tensor with its physical leg carrying a linear representation of the symmetry group and therefore represents an SPT state in class $[\omega]$. 
This is because the symmetry transformation preserves the invariant subspace $W W^\dagger$.
The isometric embedding together with the trace over ancillary degrees of freedom labeled by $k_x$ defines an onsite quantum channel $\calN_W:\bbC^{\chi^2} \mapsto \bbC^{d^2}$ with Kraus operators
\begin{align}
    (K_{k})_{\gamma}^{i} &= W_{\gamma}^{ik}
\end{align}
Since each Kraus operator is invariant under symmetry transformation, the channel is strongly symmetric~\cite{de2022symmetry}.
The channel acting on the pure SPT state with local tensor $B$ gives rise to the MPDO $\rho$ with local injective tensor
\begin{align}
    \mathbb{B}\xrightarrow{\calN_W}& \sum_{k, \gamma, \gamma'} W^{ik}_\gamma B^{\gamma}_{\alpha \beta} B^{\gamma'}_{\alpha' \beta'}  (W^{jk}_{\gamma'})^* \nonumber \\
    =& \sum_{k} A^{ik}_{\alpha\beta} (A^*)^{jk}_{\alpha'\beta'} = 
    \bbA^{ij}_{\alpha\beta;\alpha'\beta'}.
\end{align}
Since $B$ is bijective and $\bbA$ is injective, the channel $\cal N_W$ must be an injective map, i.e., a non-degenerate channel.
The ensuing MPDO has the virtual legs carrying the same projective representations as that in the local purification and therefore is in the class $[\omega]$, proving the statement.

Lastly, it is straightforward to prove the first statement from the third. The SPT pure state is described by a strong injective MPDO. The SPT phase of strong injective MPDO is preserved under the non-degenerate symmetric on-site channel. We have thus proved the theorem.
\end{proof}

The theorem has a few implications. First, in the case that a strong injective MPDO exists multiple local purifications, the theorem suggests that different purifications are in the same pure-state SPT phase.
Second, the fact that locally purifiable SPT mixed states can be realized by applying on-site non-degenerate channels to a pure SPT state suggests that the same string order parameter of the SPT pure state is also nondecaying in the SPT mixed states.
We emphasize that although our analysis is formulated in the doubled space, the SPT phases in our definition are characterized by the string order parameter in a single copy density matrix, i.e. $\tr \calO \rho$.

We remark that the local purifiablility of MPDO $\rho$ is preserved under local quantum channels in Fig~\ref{fig:local_channel}(b), namely, if $\rho$ has a local purification $\ket{\Psi_{QE}}$, $\calN[\rho]$ related to $\rho$ by a local channel $\calN$ also has a local purification.
Specifically, one can purify the local channel $\calN$ by a constant depth local unitary $U_{QE'}$ acting on the system $Q$ and ancilla $E'$, indicating that $\calN[\rho]$ has a local purification $\ket{\Psi_{QEE'}} = U_{QE'} \ket{\Psi_{QE}}\otimes\ket{0}_{E'}$.
This result suggests that all the locally decohered pure states are locally purifiable MPDO (as in Fig.~\ref{fig:strong_injective_MPDO}).

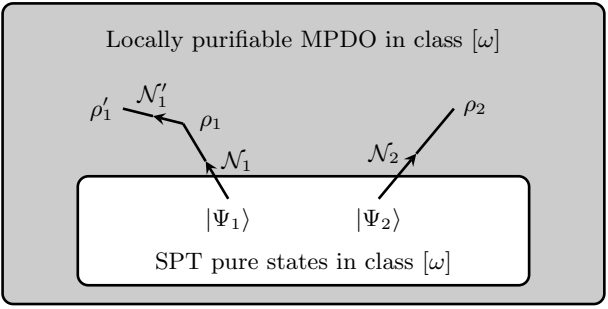
\begin{figure}[t!]
\centering
\begin{tikzpicture}
\draw[fill=black!20, rounded corners, line width=1.0] (-4,-1.5) rectangle (4, 2.5);
\draw[fill=white,rounded corners, line width=1.0] (-3,-1.25) rectangle (3, 0.2);
\node[black] at (0,-.95) {SPT pure states in class $[\omega]$};
\node[black] at (0,2.0) {Locally purifiable MPDO in class $[\omega]$};
\draw[->,>=stealth, line width=1.0] (-1.0, -0.1) node[below]{$\ket{\Psi_1}$}-- (-1.3,0.4) node[right] {\;$\calN_1$};
\draw[line width=1.0] (-1.3,0.4) -- (-1.6,0.9) node[right] {$\;\rho_1$};
\draw[->,>=stealth, line width=1.0] (-1.6,0.9) -- (-2.0,1.0) node[above] {$\calN'_1$};
\draw[line width=1.0] (-2.0,1.0) -- (-2.4,1.1) node[left] {$\rho'_1$};
\draw[->,>=stealth, line width=1.0] (1.0, -0.1) node[below]{$\ket{\Psi_2}$}-- (1.5,0.5) node[left] {$\calN_2$\;};
\draw[line width=1.0] (1.5,0.5) -- (2,1.1) node[right] {$\rho_2$};
\end{tikzpicture}
\caption{An SPT mixed state in class $[\omega]$, if admits a local purification, can be always realized by applying a non-degenerate symmetric local channel $\calN$ to an SPT pure state $\ket{\Psi}$ in the same class. On the other hand, the symmetric non-degenerate local channel preserves the SPT phase in mixed states.}
\label{fig:one-way_connection}
\end{figure}

Based on Theorem~\ref{thm:purification}, we have a compatible definition of mixed-state SPT phases (protected by strong symmetry $G$) using ``one-way" connection.
\begin{definition}
For strong injective MPDOs with a local purification, we define a relation $\rho \sim \rho'$ if there exists a non-degenerate strongly symmetric local channel $\calN$ such that $\rho = \calN[\rho']$ or $\rho' = \calN[\rho]$. Then, $\rho$ and $\rho'$ belong to the same SPT phase protected by a strong symmetry if they are in the same equivalence class generated by $``\sim"$\footnote{In general, $\rho,\rho'$ are linked through a zig-zag diagram.}.
\end{definition}

\section{String order parameter}\label{sec:stringorder}

Having established the classification of SPT phases in MPDO, we now discuss the diagnostics of distinct phases.
In quantum ground states, the string order parameter is a faithful characterization of the SPT phases if the protecting symmetry is Abelian~\cite{pollmann2012detection}, or the phase is characterized by the mixed anomaly between $G_1$ and $G_2$ belonging to the symmetry group $G_1\times G_2$~\cite{chen2014symmetry}.

The string order parameters for decohered SPT states have been discussed in various cases.
Reference~\cite{de2022symmetry} shows that an SPT pure state protected by an Abelian symmetry subject to a strong symmetric dephasing channel always has a non-decaying string order parameter as long as the channel is short of fully dephasing.
For the SPT phases protected by the mixed anomaly between the strong $G$ and the weak $K$ symmetry, Ref.~\cite{ma2023average} shows that the state exhibits a non-decaying string order given by the partial symmetry transformation of $G$ decorated by the charge of $K$.

This section discusses the string order parameters for the SPT phases in strong injective MPDO, which includes the decohered SPT states as a subset.
The theorem proved in this section leads to two main results.
First, the string order parameter is a faithful characterization of the SPT phase in strong injective MPDO protected by strong Abelian symmetry.
We note that for a restricted set of MPDO that has a local purification, this result stems directly from Theorem~\ref{thm:purification}.
Second, one can also probe the SPT phases in the MPDO characterized by $\calH^1(G,\calH^1(K,\U(1)))$ using the string order parameter that consists of the strong symmetry transformation on a partial region decorated by the charge of weak symmetry.

To begin, for an SPT mixed state protected by symmetry $G$ and $G'$, it is straightforward to show that we have the following lemma if the element $g \in G$ and $g' \in G'$ commute (see \appref{app:direct-sum}).
\begin{lemma}\label{lemma:commutator}
The slant product $i_{g}\omega(g') := \omega(g', g)/\omega(g, g')$ for $g \in G$ and $g' \in G'$ forms a one-dimensional representation of $G$ and $G'$, i.e. 
\begin{equation}
\begin{aligned}
    i_{g_1}\omega(g')i_{g_2}\omega(g') &= i_{g_1g_2}\omega(g'), \\
    i_{g}\omega(g'_1)i_{g}\omega(g'_2) &= i_{g}\omega(g'_1g'_2),
\end{aligned}
\end{equation}
if $gg'=g'g$ for $\forall g\in G$, $g' \in G'$.
\end{lemma}

A non-trivial slant product guarantees the existence of string order according to the following theorem.
\begin{theorem}\label{thm:string_order_parameter}
Consider a strong injective MPDO $\rho$ in an SPT mixed state protected by symmetry $G$ and $G'$, where $G$ is a strong symmetry, and the group elements $g \in G$ and $g'\in G'$ commute, i.e. $gg' = g'g$.
If the slant product $i_{g}\omega(g')$ forms a non-trivial one-dimensional representation of $G'$, there always exists a non-decaying string order parameter 
\begin{equation}
    \calO_{\alpha,g}(x_1,x_2) = R_{\alpha^*}(x_1) \left(\prod_{x = x_1+1}^{x_2-1} U_{g}(x) \right)R_\alpha(x_2),\nonumber
\end{equation}
where $R_\alpha$ is a charge operator with $e^{-\ri\alpha(\cdot)} = i_{g}\omega(\cdot)$ and transforms as $U_{g'} R_\alpha U_{g'}^\dagger = e^{\ri\alpha(g')} R_\alpha$, $\forall g' \in G'$.
On the other hand, the string order parameter with $e^{-\ri\alpha(\cdot)} \neq i_{g}\omega(\cdot)$ is vanishing.
\end{theorem}

\begin{proof}
We take a tensor network formulation of the string order parameter $\tr\rho\calO_{\alpha, g}$. 
\begin{align}
    \tr\rho\calO_{\alpha, g} = \begin{tikzpicture}[scale=0.7,line cap=round,line join=round,x=1cm,y=1cm,baseline={(0, -\MathAxis pt)}]
\definecolor{myred}{RGB}{240,83,90};
\definecolor{myblue}{RGB}{73,103,189};
\definecolor{myturquoise}{RGB}{83,195,189};
\draw[black!80, line width=1.0] (-0.5,0) -- (1.5,0);
\draw[black!80, line width=1.0] (2.5,0) -- (4.5,0);
\draw[black!80, line width=1.0] (5.5,0) -- (7.5,0);
\foreach \x in {0,1,3,4,6,7}{
\draw[black!80, line width=1.0] (\x,-0.5) -- (\x, 1.0);
\filldraw[fill=white, draw=black!80, line width=1.0] (\x-0.25,-0.25) rectangle (\x+0.25,0.25);
\filldraw[white] (\x+0.35,-0.05) rectangle ++(0.1,0.1);
\draw[black, line width=1.0] plot [smooth, tension=0.5] coordinates { (\x,-0.5) (\x+0.2,-0.65) (\x+0.4,-0.5) (\x+0.4,0) (\x+0.4,0.5) (\x+0.4,1.0) (\x+0.2, 1.15) (\x, 1.0)};
}
\node[black] at (2,0) {$\cdots$};
\node[black] at (5,0) {$\cdots$};
\filldraw[myred] (1,0.5) circle (0.1);
\filldraw[white] (0.3,0.3) rectangle ++(0.2,0.6);
\node[myred] at (0.5,0.6) {$R_{\alpha^*}$};
\filldraw[myred] (6,0.5) circle (0.1);
\node[myred] at (5.5,0.6) {$R_\alpha$};
\filldraw[myblue] ([xshift=-3pt,yshift=-3pt]3,0.5) rectangle ++(6pt,6pt);
\node[myblue] at (2.5,0.6) {$U_g$};
\filldraw[myblue] ([xshift=-3pt,yshift=-3pt]4,0.5) rectangle ++(6pt,6pt);
\end{tikzpicture}\nonumber
\end{align}
By grouping a constant number of sites, one can convert the local transfer matrix for the strong injective MPDO to its fixed point $\bbT = |\phi_r) (\phi_\ell|$, where $|\phi_r)$ and $(\phi_\ell|$ are the right and the left eigenvector associated with its unique largest eigenvalue, respectively.
\begin{equation}
\bbT = 
\begin{tikzpicture}[line cap=round,line join=round,x=1cm,y=1cm,baseline={(0, -\MathAxis pt)}]
\draw[black!80, line width=1.0] (-0.5,0) -- (0.5,0);
\draw[black!80,line width=1.0] (0.1,0.0) -- (0.1,0.5);
\draw[black!80,line width=1.0] (-0.1,0.0) -- (-0.1,0.5);
\filldraw[fill=white, draw=black!80, line width=1.0] (-0.25,-0.25) rectangle (+0.25,0.25);
\def\Radius{0.1}
\draw[black, line width=1.0] (0.1,0.5) arc(0:180:\Radius);
\node[black] at (0.0,0.0) {$\bbA$};
\end{tikzpicture}
=
\begin{tikzpicture}[line cap=round,line join=round,x=1cm,y=1cm,baseline={(0, -\MathAxis pt)}]
\draw[black!80, line width=1.0] (-0.5,0) -- (0,0);
\draw[black!80, line width=1.0] (0.8,0.0) -- (1.3,0.0);
\draw[fill=white, draw=black!80, line width=1.0] (-0.25,0.25) -- (-0.25,-0.25) -- (0.1,-0.25) arc(-90:90:0.25) -- cycle;
\draw[fill=white, draw=black!80, line width=1.0] (1.05,-0.25) -- (1.05,0.25) -- (0.7,0.25) arc(90:270:0.25) -- cycle;
\node[black] at (0,0) {$\phi_r$};
\node[black] at (0.8,0) {$\phi_\ell$};
\end{tikzpicture}
\end{equation}
The symmetry transformation acting on the physical leg fractionalizes into two projective representations, i.e.
\begin{equation}
\sum_{i,j} U_g^{ji} \bbA^{ij} = 
\begin{tikzpicture}[line cap=round,line join=round,x=1cm,y=1cm,baseline={(0, -\MathAxis pt)}]
\definecolor{myblue}{RGB}{73,103,189};
\draw[black!80, line width=1.0] (-0.5,0) -- (0.5,0);
\draw[black!80,line width=1.0] (0.1,0.0) -- (0.1,0.8);
\draw[black!80,line width=1.0] (-0.1,0.0) -- (-0.1,0.8);
\filldraw[fill=white, draw=black!80, line width=1.0] (-0.25,-0.25) rectangle (+0.25,0.25);
\def\Radius{0.1}
\draw[black, line width=1.0] (0.1,0.8) arc(0:180:\Radius);
\node[black] at (0.0,0.0) {$\bbA$};
\node[myblue] at (-0.5,0.6) {$U_g$};
\filldraw[myblue] ([xshift=-3pt,yshift=-3pt]-0.1,0.5) rectangle ++(6pt,6pt);
\end{tikzpicture}
=
\begin{tikzpicture}[line cap=round,line join=round,x=1cm,y=1cm,baseline={(0, -\MathAxis pt)}]
\definecolor{myturquoise}{RGB}{83,195,189};
\draw[black!80, line width=1.0] (-0.8,0) -- (0,0);
\draw[black!80, line width=1.0] (0.8,0.0) -- (1.6,0.0);
\draw[fill=white, draw=black!80, line width=1.0] (-0.25,0.25) -- (-0.25,-0.25) -- (0.1,-0.25) arc(-90:90:0.25) -- cycle;
\draw[fill=white, draw=black!80, line width=1.0] (1.05,-0.25) -- (1.05,0.25) -- (0.7,0.25) arc(90:270:0.25) -- cycle;
\filldraw[myturquoise] ([xshift=-3pt,yshift=-3pt]-0.5,0) rectangle ++(6pt,6pt);
\node[myturquoise] at (-0.5,0.5) {$V_g^{-1}$};
\filldraw[myturquoise] ([xshift=-3pt,yshift=-3pt]1.3,0) rectangle ++(6pt,6pt);
\node[myturquoise] at (1.3,0.5) {$V_g$};
\node[black] at (0,0) {$\phi_r$};
\node[black] at (0.8,0) {$\phi_\ell$};
\end{tikzpicture}
\end{equation}
Thus, the string order parameter is given by the product of two tensor network diagrams
\begin{align}
    \tr\rho\calO =&\, 
\begin{tikzpicture}[line cap=round,line join=round,x=1cm,y=1cm,baseline={(0, -\MathAxis pt)}]
\definecolor{myred}{RGB}{240,83,90};
\definecolor{myblue}{RGB}{73,103,189};
\definecolor{myturquoise}{RGB}{83,195,189};
\draw[black!80, line width=1.0] (-1.0,0) -- (1.0,0);
\draw[black!80,line width=1.0] (0.1,0.0) -- (0.1,0.8);
\draw[black!80,line width=1.0] (-0.1,0.0) -- (-0.1,0.8);
\filldraw[fill=white, draw=black!80, line width=1.0] (-0.25,-0.25) rectangle (+0.25,0.25);
\draw[fill=white, draw=black!80, line width=1.0] (-0.75,-0.25) -- (-0.75,0.25) -- (-1.1,0.25) arc(90:270:0.25) -- cycle;
\node[black] at (-1.0,0) {$\phi_\ell$};
\def\Radius{0.1}
\draw[black, line width=1.0] (0.1,0.8) arc(0:180:\Radius);
\node[black] at (0.0,0.0) {$\bbA$};
\node[myred] at (-0.5,0.5) {$R_{\alpha^*}$};
\filldraw[myred] (-0.1,0.5) circle (3pt);
\draw[fill=white, draw=black!80, line width=1.0] (0.75,0.25) -- (0.75,-0.25) -- (1.1,-0.25) arc(-90:90:0.25) -- cycle;
\filldraw[myturquoise] ([xshift=-3pt,yshift=-3pt]0.5,0) rectangle ++(6pt,6pt);
\node[myturquoise] at (0.5,0.5) {$V_g^{-1}$};
\node[black] at (1.0,0) {$\phi_r$};
\end{tikzpicture}\;
\begin{tikzpicture}[line cap=round,line join=round,x=1cm,y=1cm,baseline={(0, -\MathAxis pt)}]
\definecolor{myred}{RGB}{240,83,90};
\definecolor{myblue}{RGB}{73,103,189};
\definecolor{myturquoise}{RGB}{83,195,189};
\draw[black!80, line width=1.0] (-1.0,0) -- (1.0,0);
\draw[black!80,line width=1.0] (0.1,0.0) -- (0.1,0.8);
\draw[black!80,line width=1.0] (-0.1,0.0) -- (-0.1,0.8);
\filldraw[fill=white, draw=black!80, line width=1.0] (-0.25,-0.25) rectangle (+0.25,0.25);
\draw[fill=white, draw=black!80, line width=1.0] (-0.75,-0.25) -- (-0.75,0.25) -- (-1.1,0.25) arc(90:270:0.25) -- cycle;
\node[black] at (-1.0,0) {$\phi_\ell$};
\def\Radius{0.1}
\draw[black, line width=1.0] (0.1,0.8) arc(0:180:\Radius);
\node[black] at (0.0,0.0) {$\bbA$};
\node[myred] at (0.4,0.5) {$R_{\alpha}$};
\filldraw[myred] (-0.1,0.5) circle (3pt);
\draw[fill=white, draw=black!80, line width=1.0] (0.75,0.25) -- (0.75,-0.25) -- (1.1,-0.25) arc(-90:90:0.25) -- cycle;
\filldraw[myturquoise] ([xshift=-3pt,yshift=-3pt]-0.5,0) rectangle ++(6pt,6pt);
\node[myturquoise] at (-0.5,0.5) {$V_g$};
\node[black] at (1.0,0) {$\phi_r$};
\end{tikzpicture} \nonumber \\
:=& \tr(R_{\alpha^*}M_L) \tr(R_\alpha M_R)
\end{align}

The tensor $M_R$ lives in an adjoint representation of the symmetry group $G'$. 
Specifically, we have
\begin{align}\label{adjointRepresentation}
    U_{g'}^\dagger M_R U_{g'} &= \begin{tikzpicture}[line cap=round,line join=round,x=1cm,y=1cm,baseline={(0, -\MathAxis pt)}]
\definecolor{myred}{RGB}{240,83,90};
\definecolor{myblue}{RGB}{73,103,189};
\definecolor{myturquoise}{RGB}{83,195,189};
\draw[black!80, line width=1.0] (-1.0,0) -- (1.0,0);
\draw[black!80,line width=1.0] (0.1,0.0) -- (0.1,1.3);
\draw[black!80,line width=1.0] (-0.1,0.0) -- (-0.1,1.3);
\def\Radius{0.1}
\draw[black, line width=1.0] (0.1,1.3) arc(0:180:\Radius);
\filldraw[fill=white, draw=black!80, line width=1.0] (-0.25,-0.25) rectangle (+0.25,0.25);
\node[black] at (0.0,0.0) {$\bbA$};
\draw[fill=white, draw=black!80, line width=1.0] (-0.75,-0.25) -- (-0.75,0.25) -- (-1.1,0.25) arc(90:270:0.25) -- cycle;
\node[black] at (-1.0,0) {$\phi_\ell$};
\draw[fill=white, draw=black!80, line width=1.0] (0.75,0.25) -- (0.75,-0.25) -- (1.1,-0.25) arc(-90:90:0.25) -- cycle;
\filldraw[myturquoise] ([xshift=-3pt,yshift=-3pt]-0.5,0) rectangle ++(6pt,6pt);
\node[black] at (1.0,0) {$\phi_r$};
\filldraw[white] (-0.1,0.8) circle (3pt);
\node[myblue] at (-0.5,0.5) {$U_{g'}^\dagger$};
\filldraw[myblue] ([xshift=-3pt,yshift=-3pt]-0.1,0.5) rectangle ++(6pt,6pt);
\node[myblue] at (-0.5,1.1) {$U_{g'}$};
\filldraw[myblue] ([xshift=-3pt,yshift=-3pt]-0.1,1.1) rectangle ++(6pt,6pt);
\end{tikzpicture}
= \begin{tikzpicture}[line cap=round,line join=round,x=1cm,y=1cm,baseline={(0, -\MathAxis pt)}]
\definecolor{myred}{RGB}{240,83,90};
\definecolor{myblue}{RGB}{73,103,189};
\definecolor{myturquoise}{RGB}{83,195,189};
\draw[black!80, line width=1.0] (-1.6,0) -- (1.0,0);
\draw[black!80,line width=1.0] (0.1,0.0) -- (0.1,0.8);
\draw[black!80,line width=1.0] (-0.1,0.0) -- (-0.1,0.8);
\def\Radius{0.1}
\draw[black, line width=1.0] (0.1,0.8) arc(0:180:\Radius);
\filldraw[fill=white, draw=black!80, line width=1.0] (-0.25,-0.25) rectangle (+0.25,0.25);
\node[black] at (0.0,0.0) {$\bbA$};
\draw[fill=white, draw=black!80, line width=1.0] (-1.05,-0.25) -- (-1.05,0.25) -- (-1.4,0.25) arc(90:270:0.25) -- cycle;
\node[black] at (-1.3,0) {$\phi_\ell$};
\draw[fill=white, draw=black!80, line width=1.0] (0.75,0.25) -- (0.75,-0.25) -- (1.1,-0.25) arc(-90:90:0.25) -- cycle;
\filldraw[myturquoise] ([xshift=-3pt,yshift=-3pt]-0.8,0) rectangle ++(6pt,6pt);
\node[black] at (1.0,0) {$\phi_r$};
\filldraw[white] (-0.1,0.5) circle (3pt);
\node[myblue!80] at (-0.6,0.45) {$V^{\vphantom{-1}}_{g'\bar{g}'}$};
\filldraw[myblue!80] ([xshift=-3pt,yshift=-3pt]-0.5,0) rectangle ++(6pt,6pt);
\node[myblue!80] at (0.48,0.45) {$V^{-1}_{g'\bar{g}'}$};
\filldraw[myblue!80] ([xshift=-3pt,yshift=-3pt]0.5,0) rectangle ++(6pt,6pt);
\node[myturquoise] at (-0.8,-0.4) {$V_g$};
\end{tikzpicture} \nonumber \\
&= i^*_{g}\omega(g')\begin{tikzpicture}[line cap=round,line join=round,x=1cm,y=1cm,baseline={(0, -\MathAxis pt)}]
\definecolor{myred}{RGB}{240,83,90};
\definecolor{myblue}{RGB}{73,103,189};
\definecolor{myturquoise}{RGB}{83,195,189};
\draw[black!80, line width=1.0] (-1.6,0) -- (1.0,0);
\draw[black!80,line width=1.0] (0.1,0.0) -- (0.1,0.8);
\draw[black!80,line width=1.0] (-0.1,0.0) -- (-0.1,0.8);
\def\Radius{0.1}
\draw[black, line width=1.0] (0.1,0.8) arc(0:180:\Radius);
\filldraw[fill=white, draw=black!80, line width=1.0] (-0.25,-0.25) rectangle (+0.25,0.25);
\node[black] at (0.0,0.0) {$\bbA$};
\draw[fill=white, draw=black!80, line width=1.0] (-1.05,-0.25) -- (-1.05,0.25) -- (-1.4,0.25) arc(90:270:0.25) -- cycle;
\node[black] at (-1.3,0) {$\phi_\ell$};
\draw[fill=white, draw=black!80, line width=1.0] (0.75,0.25) -- (0.75,-0.25) -- (1.1,-0.25) arc(-90:90:0.25) -- cycle;
\filldraw[myturquoise] ([xshift=-3pt,yshift=-3pt]-0.5,0) rectangle ++(6pt,6pt);
\node[black] at (1.0,0) {$\phi_r$};
\filldraw[white] (-0.1,0.5) circle (3pt);
\node[myblue!80] at (-0.8,0.45) {$V^{\vphantom{-1}}_{g'\bar{g}'}$};
\filldraw[myblue!80] ([xshift=-3pt,yshift=-3pt]-0.8,0) rectangle ++(6pt,6pt);
\node[myblue!80] at (0.48,0.45) {$V^{-1}_{g'\bar{g}'}$};
\filldraw[myblue!80] ([xshift=-3pt,yshift=-3pt]0.5,0) rectangle ++(6pt,6pt);
\node[myturquoise] at (-0.5,-0.4) {$V_g$};
\end{tikzpicture}\nonumber \\
&= i^*_{g}\omega(g') M_R,
\end{align}
where $g'\bar{g}'$ labels the symmetry transformation acting on both ket and bra Hilbert space.
On the second line, we use $V_gV_{g'\bar{g}'} = i^*_{g}\omega(g'\bar{g}') V_{g'\bar{g}'} V_g$, and $i_{g}\omega(g'\bar{g}') = i_{g}\omega(g')i_{g}\omega(\bar{g}') = i_{g}\omega(g')$ because the projective representation of the left and the right symmetry commute, i.e. $i_{g}\omega(\bar{g}') = 1$, according to Sec.~\ref{sec:positivity_constraint}. 
Furthermore, since $\bbT = V_{g'\bar{g}'}^{-1} \bbT V_{g'\bar{g}'}$, we have $V^{-1}_{g'\bar{g}'}|\phi_r)(\phi_\ell|V_{g'\bar{g}'} = |\phi_r)(\phi_\ell|$, leading to the final result.
According to the Lemma~\ref{lemma:commutator}, $i_{g}\omega$ is one-dimensional representation.
Thus, due to the selection rule, if $e^{-\ri\alpha(\cdot)} \neq i_{g}\omega(\cdot)$, then $ \tr R_\alpha M_R = 0$ (similarly $\tr R_{\alpha^*}M_L =0$).

There always exists an operator $R_\alpha$ in the representation $\alpha$ such that $\tr(R_\alpha M_R) \neq 0$.
Since the local tensor $\bbA$ is injective, there always exists an operator $R'$ such that
\begin{align}
    M_R = \begin{tikzpicture}[line cap=round,line join=round,x=1cm,y=1cm,baseline={(0, -\MathAxis pt)}]
\definecolor{myred}{RGB}{240,83,90};
\definecolor{myblue}{RGB}{73,103,189};
\definecolor{myturquoise}{RGB}{83,195,189};
\draw[black!80, line width=1.0] (-1.0,0) -- (1.0,0);
\draw[black!80,line width=1.0] (0.1,0.0) -- (0.1,1.1);
\draw[black!80,line width=1.0] (-0.1,0.0) -- (-0.1,1.1);
\def\Radius{0.1}
\draw[black, line width=1.0] (0.1,1.1) arc(0:180:\Radius);
\filldraw[fill=white, draw=black!80, line width=1.0] (-0.25,-0.25) rectangle (+0.25,0.25);
\node[black] at (0.0,0.0) {$\bbA$};
\draw[fill=white, draw=black!80, line width=1.0] (-0.75,-0.25) -- (-0.75,0.25) -- (-1.1,0.25) arc(90:270:0.25) -- cycle;
\node[black] at (-1.0,0) {$\phi_\ell$};
\draw[fill=white, draw=black!80, line width=1.0] (0.75,0.25) -- (0.75,-0.25) -- (1.1,-0.25) arc(-90:90:0.25) -- cycle;
\filldraw[myred] (-0.1,0.5) circle (3pt);
\node[myred] at (-0.5,0.5) {$R'$};
\node[black] at (1.0,0) {$\phi_r$};
\filldraw[white] (-0.1,0.8) circle (3pt);
\end{tikzpicture},
\end{align}
where $R'$ is invertible and is in the irrep $i^*_{g}\omega(\cdot)$.
Thus, by choosing $R_\alpha = (R')^{-1}$, $\tr(R_\alpha M_R) = 1$, we prove the existence of a non-vanishing string order parameter.
\end{proof}

Theorem~\ref{thm:string_order_parameter} leads to two main results. First, in the case that $G = G'$ is a strong Abelian symmetry, we have the following corollary.
\begin{corollary}
Consider a strong injective MPDO $\rho$ in an SPT phase protected by strong Abelian symmetry $G$. The string order parameter $\calO_{\alpha, g}$ is non-decaying in $\rho$ if $e^{-\ri\alpha(\cdot)} = i_{g}\omega(\cdot)$, while $\calO_{\alpha, g}$ vanishes otherwise.
\end{corollary}

We remark that Ref.~\cite{de2022symmetry} defines SPT phases protected by Abelian symmetry in locally decohered SPT states according to the non-decaying string order parameter.
They showed that the string order parameter is preserved under a strongly symmetric local channel as long as the channel is short of fully dephasing.
This is consistent with our result because a partially dephasing channel is non-degenerate and therefore preserves the SPT phase; the locally decohered Abelian SPT state belongs to the same phase of strong injective MPDO and features the non-decaying string order.

Theorem~\ref{thm:string_order_parameter} also applies when $G' = K$ is the weak symmetry, indicating the existence of string order parameters for the SPT phases that feature the mixed anomaly between the strong and the weak symmetry.
In this case, the SPT phases are fully specified by the slant product $i_{g}\omega(k)$ for $g\in G$ and $k \in K$.
We have the following corollary for general group $G$ and $K$, which are not necessarily Abelian.
\begin{corollary}
Consider a strong injective MPDO $\rho$ in an SPT phase protected by symmetry $G\times K$, where $G$ is strong symmetry and $K$ is weak symmetry. The string order parameter $\calO_{\alpha, g}$ for $g \in G$ is non-decaying in $\rho$ if $e^{-\ri\alpha(\cdot)} = i_{g}\omega(\cdot)$, while $\calO_{\alpha, g}$ vanishes otherwise.
\end{corollary}

\section{Two-dimensional generalization}\label{sec:generalization}

In this section, we study SPT phases in two-dimensional mixed states.
We introduce two-dimensional tensor network density operators with strong semi-injectivity, generalizing the concept of semi-injective projected entangled pair states (PEPS).
We provide a classification of the SPT phases in the corresponding doubled state.

\subsection{2d SPT and MPO injectivity}

Our first job is to find a suitable generalization of strong injective MPDO into two dimensions, which would allow us to properly define mixed-state SPTs. 
One crucial point is that an injective PEPS (as a generalization of injective MPS) fails to encapsulate various SPT or topological ground states of gapped local Hamiltonians in higher dimensions~\cite{Chen_2011, PEPSlocalHam2019}. 
In particular, injective PEPS can only represent weak SPT~\cite{directSplit}, which is the stacking of 1d SPT protected by the translation symmetry. 
This stems from the fundamental theorem of PEPS, which states that any two injective tensors $A$ and $B$ generating the same state should be related by the following gauge transformation~\cite{perez2009canonical}:
\begin{align}
    \begin{tikzpicture}[line cap=round,line join=round,x=1cm,y=1cm,baseline={(0, -\MathAxis pt)}]
\definecolor{myred}{RGB}{148, 181, 235};
\definecolor{myturquoise}{RGB}{83,195,189};
    \draw[color=black, line width=1.5] (0,0) -- (0.5,-0.5);    
    \draw[black, line width=1.0] (-0.8,0) -- (0.8,0);
    \draw[black, line width=1.0] (0,-0.8) -- (0,0.8);    
    \draw[black, fill=white, line width=1.0] (-0.3,-0.3) rectangle (0.3,0.3);
    \node at (0,0) {$A$};
\end{tikzpicture} \,=\, 
\begin{tikzpicture}[line cap=round,line join=round,x=1cm,y=1cm,baseline={(0, -\MathAxis pt)}]
\definecolor{myred}{RGB}{148, 181, 235};
\definecolor{myturquoise}{RGB}{83,195,189};
    \draw[color=black, line width=1.5] (0,0) -- (0.5,-0.5);    
    \draw[black, line width=1.0] (-1.2,0) -- (1.2,0);
    \draw[black, line width=1.0] (0,-1.2) -- (0,1.2);    
    \draw[black, fill=white, line width=1.0] (-0.3,-0.3) rectangle (0.3,0.3);
    \node at (0,0) {$B$};
    \draw[black, fill=myred, line width=1.0] (-0.2,-1) rectangle (0.2,-0.6);
    \node at (0.45,-1) {$Y$};
    \draw[black, fill=myred, line width=1.0] (-0.2,0.6) rectangle (0.2,1);
    \node at (0.65,1) {$Y^{-1}$};    
    \draw[black, fill=myred, line width=1.0] (-1,-0.2) rectangle (-0.6,0.2);
    \node at (-0.8,0.45) {$X^{-1}$};
    \draw[black, fill=myred, line width=1.0] (0.6,-0.2) rectangle (1,0.2);
    \node at (1,0.45) {$X$};
\end{tikzpicture}.
\end{align}
Accordingly, if a given injective PEPS is invariant under on-site symmetry $G$, we obtain that
\begin{align}
    U_g A = A (X_g \otimes X_g^{-1} \otimes Y_g \otimes Y_g^{-1}).
\end{align}
When $X_g$ or $Y_g$ forms a non-trivial projective representation of $G$, the 2d PEPS is in a weak SPT phase characterized by a stacking of 1d SPTs.

To describe general 2d SPT pure states, one needs to consider a more generalized class of PEPS, where the on-site symmetry action $U_g$ is not simply translating into a tensor product of gauge transformation; rather, $U_g$ is given as a matrix product operator (MPO) action $V_g$ along the virtual legs as follows:
\begin{align} \label{eq:MPO_sym}
    \begin{tikzpicture}[line cap=round,line join=round,x=1cm,y=1cm,baseline={(0, -\MathAxis pt)}]
\definecolor{myred}{RGB}{148, 181, 235};
\definecolor{myturquoise}{RGB}{83,195,189};
    \draw[color=black, line width=1.5] (0,0) -- (0.5,-0.5);    
    \node at (0.7,-0.7) {$U_g$};
    \draw[black, line width=1.0] (-0.8,0) -- (0.8,0);
    \draw[black, line width=1.0] (0,-0.8) -- (0,0.8);    
    \draw[black, fill=white, line width=1.0] (-0.3,-0.3) rectangle (0.3,0.3);
    \node at (0,0) {$A$};
\end{tikzpicture} \,=\, \begin{tikzpicture}[line cap=round,line join=round,x=1cm,y=1cm,baseline={(0, -\MathAxis pt)}]
\definecolor{myred}{RGB}{148, 181, 235};
\definecolor{myturquoise}{RGB}{83,195,189};
    \draw[color=black, line width=1.5] (0,0) -- (0.5,-0.5);    
    \draw[black, line width=1.0] (-1.2,0) -- (1.2,0);
    \draw[black, line width=1.0] (0,-1.2) -- (0,1.2);    
    \draw[black, fill=white, line width=1.0] (-0.3,-0.3) rectangle (0.3,0.3);
    \node at (0,0) {$A$};
    \draw[myred, line width=2.0]  (0,-0.8) .. controls (0.8,-0.8) .. (0.8,0);
    \draw[myred, line width=2.0]  (0.8,0) .. controls (0.8,0.8) .. (0,0.8);
    \draw[myred, line width=2.0]  (0,0.8) .. controls (-0.8,0.8) .. (-0.8,0);
    \draw[myred, line width=2.0]  (-0.8,0) .. controls (-0.8,-0.8) .. (0,-0.8);    
    \draw[black, fill=myred, line width=1.0] (-0.2,-1) rectangle (0.2,-0.6);
    \draw[black, fill=myred, line width=1.0] (-0.2,0.6) rectangle (0.2,1);
    \draw[black, fill=myred, line width=1.0] (-1,-0.2) rectangle (-0.6,0.2);
    \draw[black, fill=myred, line width=1.0] (0.6,-0.2) rectangle (1,0.2);
    \node[text=black] at (1,0.8) {$V_g$};
\end{tikzpicture},
\end{align}
where blue legs visualize that $V_g$ is not a simple tensor product.

Is it possible to refine the concept of injectivity in PEPS to accommodate symmetries as indicated by \eqnref{eq:MPO_sym}, without overly broadening the scope to include states with topological order or gaplessness? 
Williamson et al.~\cite{Williamson_2016} proposed that single-blocked MPO-injective PEPS provides a faithful description of 2d SPT states. MPO-injective PEPS $A:\mathbb{C}^{D^z}\,{\rightarrow}\,\mathbb{C}^d$ admits the pseudoinverse as $A^+$ such the projection $A^+ A:(\mathbb{C}^D)^{\otimes z} \,{\rightarrow}\, (\mathbb{C}^D)^{\otimes z}$ is expressed as the matrix product operator (MPO) that is built from copies of a single local tensor $P$ with $P_{\mu\nu}:\mathbb{C}^D\,{\rightarrow}\,\mathbb{C}^D$ where $\mu,\nu$ are MPO bond indices~\cite{_ahino_lu_2021}. 
Specifically, when $P$ is an injective MPO, or \emph{single-block}, the MPO symmetries $V_g$s align with this single-blocked structure~\cite{Williamson_2016}.
Moreover, such a single-block MPO-injective PEPS is associated with a local, gapped parent Hamiltonian with a unique ground state\footnote{Multi-block MPO-injective PEPS can have a topological order with the ground state degeneracy larger than one.}.

Nevertheless, the principle of MPO-injectivity is primarily established on axioms, with \eqnref{eq:MPO_sym} accepted rather than derived as a fundamental theorem. 
Addressing this gap, Molnar et al. (2018) rigorously proved Eq.~\eqref{eq:MPO_sym} for a \emph{semi-injective} PEPS; semi-injective PEPS  includes a broad spectrum of states, notably 2d SPT states, and also guarantees the existence of a local gapped Hamiltonian with a unique ground state for such PEPS.
It is posited that semi-injective PEPS might be considered equivalent to single-block MPO-injective PEPS under mild conditions.
Therefore, for a more concrete discussion, we adopt the notion of semi-injectivity in what follows.

\begin{definition}
(Semi-injective PEPS) A semi-injective PEPS $\ket{\Psi[A]}$ on a square lattice is defined by two objects, a four-body state $\psi$ and an invertible operator $O$: at each site, there are four virtual degrees of freedom (dof). At each plaquette, four dofs from four sites form an entangled state $\psi$ (gray color). Then, $O$ (dashed line) maps four unentangled virtual dofs at each site into a physical wavefunction:
\begin{align} \label{eq:semiPEPS}
    |\Psi[A(\psi,O)] \rangle \,=\,
    \begin{tikzpicture}[line cap=round,line join=round,x=1cm,y=1cm,baseline={(0, -\MathAxis pt)}]
\definecolor{myred}{RGB}{135, 10, 10};
\definecolor{myblue}{RGB}{204, 194, 194};
\definecolor{myturquoise}{RGB}{83,195,189};
\draw[black, line width=1, fill=myblue] (-0.7, -0.7) rectangle (0.7,0.7);
\fill[fill=myblue] (0.8,-0.7) rectangle (1.4,0.7);
\fill[fill=myblue] (-1.4,-0.7) rectangle (-0.8,0.7);
\fill[fill=myblue] (-0.7,-1.4) rectangle (0.7,-0.8);
\fill[fill=myblue] (-0.7,0.8) rectangle (0.7,1.4);
\fill[fill=myblue] (-1.4,-1.4) rectangle (-0.8,-0.8);
\fill[fill=myblue] (0.8,-1.4) rectangle (1.4,-0.8);
\fill[fill=myblue] (-1.4,0.8) rectangle (-0.8,1.4);
\fill[fill=myblue] (0.8,0.8) rectangle (1.4,1.4); 
\foreach \i in {-0.5,0.5}
{
    \foreach \j in {-0.5,0.5}
    {
        \draw[black, fill=white, line width=1.0] (-0.2 + 1.5*\i,-0.2+1.5*\j) circle (0.1);
        \draw[black, fill=white, line width=1.0] (0.2 + 1.5*\i,-0.2+1.5*\j) circle (0.1);
        \draw[black, fill=white, line width=1.0] (-0.2 + 1.5*\i,0.2+1.5*\j) circle (0.1);
        \draw[black, fill=white, line width=1.0] (0.2 + 1.5*\i,0.2+1.5*\j) circle (0.1);
    }
}
\draw[black, dotted, line width=1] (-0.75,-0.75) circle (0.45);
\draw[black, dotted, line width=1] (0.75,-0.75) circle (0.45);
\draw[black, dotted, line width=1] (-0.75,0.75) circle (0.45);
\draw[black, dotted, line width=1] (0.75,0.75) circle (0.45);
\node[text=black] at (-0.15,-1.05) {$\bm{O}$};
\node[text=black] at (0.,0.) {$\bm{\psi}$};
\end{tikzpicture}  
\end{align}
\end{definition} 

The fundamental theorem for semi-injective PEPS~\cite{Molnar_2018} states that the action of an on-site symmetry $g \in G$ on the physical legs in the region $\calR$ transforms into the action on the virtual legs in the boundary $\partial \calR$ as
\begin{align} \label{eq:symmMPO}
    U_g(\calR) A(\calR) = A(\calR) V_g(\partial \calR)
\end{align}
where $V_g(\partial \calR)$ is an MPO acting on the virtual legs in $\partial \calR$, as visualized in \figref{fig:PEPS}(a)~\footnote{More precisely, in \cite{Molnar_2018} the $V_g(\partial \calR)$ was defined on the  Hilbert space of the minimal rank decomposition of four-qubit state $\psi$ in semi-injective PEPS. If the projection from the PEPS bond Hilbert space to this Hilbert space is $P$, one can always find a pseudoinverse $P^+$ (not unique) to bring $V_g(\partial \calR)$ into the PEPS bond Hilbert space.}. 
For a horizontal strip $\calR$, $V_g(\partial \calR) = V_g(\partial \calR_\textrm{top}) \cdot V^{-1}_g(\partial \calR_\textrm{bot})$ consists of two parts acting on top and bottom boundaries. Note that the operations on the top and bottom boundaries cancel each other such that the symmetry action on a larger strip $\calR$ still results in the virtual action on its two boundaries. 
One can obtain a similar behavior for a vertical strip.
We would use $V_g$ to denote a single local tensor for an MPO $V_g(\rd \calR)$.

\begin{figure}[t]
\centering
\includegraphics[width=1\columnwidth]{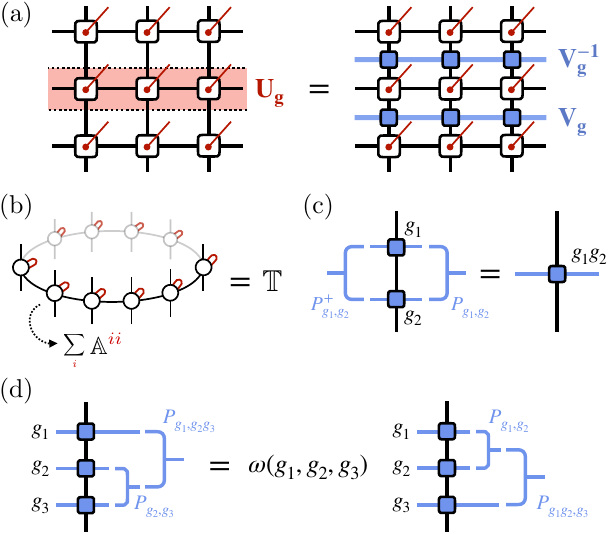}
\caption{ 
(a) In a semi-injective PEPS, the action of an on-site symmetry $U_g$ on a strip $R$  transforms into the action of an MPO on the boundaries $V_g$ (bottom) and $V_g^{-1}$ (top). 
(b) For a given 2d local tensor $\bbA$, one contract two physical legs (red lines). The transfer matrix $\bbT$ on an infinite cylinder can be obtained by contracting it along the circumference. This transfer matrix captures correlation along the infinite direction.  (c) Given two MPO symmetries $V_{g_1}$ and $V_{g_2}$, one can reduce them into $V_{g_1, g_2}$ by applying reduction tensors $P_{g_1,g_2}$ and $P_{g_1,g_2}^+$ from its left and right. (d) Two different ways to merge three symmetry MPOs from one end give rise to a relative phase factor  $\omega(g_1,g_2,g_3)$.}
\label{fig:PEPS}
\end{figure}

With these definitions in hand, we can extend the concept from pure states to two-dimensional tensor-network density operators (TNDO) as follows:
\begin{definition}\label{def:strong_semi_injective}
(Strong semi-injective TNDO) Consider a TNDO  $\rho[\bbA]$ with local tensor $\bbA$ with two physical and four virtual indices. It is strong semi-injective if 
\begin{enumerate}
    \item The doubled state $\kket{\rho[\bbA] }$ for the TNDO is a semi-injective PEPS, where $\bbA = \bbA[\mathbb{\psi},\mathbb{O}]$.
    \item The transfer matrix $\bbT$ of the TNDO on the infinite cylinder is defined by contracting local tensor $\sum_i \bbA^{ii}$ (two superscripts are physical indices for ket and bra) as in \figref{fig:PEPS}(b). $\mathbb{T}$ has a unique largest eigenvalue at any circumference length $L$, and the corresponding eigenvector is given as an injective MPS with a finite bond dimension.
\end{enumerate}
These two conditions are the direct generalizations of the conditions for MPDO. 
The first condition suggests the density matrix represents a short-range entangled state within a doubled Hilbert space. 
The second condition ensures that the correlations are short-range within the physical Hilbert space. If the largest eigenvalue of $\bbT$ is not unique, there exists either a local or loop operator with a long-range correlation along the cylinder's axis, which deviates from the intended short-range correlated criteria.
\end{definition}

A natural consequence of this definition is that the non-degenerate onsite quantum channel preserves strong semi-injectivity.
This is because a non-degenerate onsite channel acts on physical Hilbert space at each site as an injective map, thus its action can be absorbed into an invertible map $\mathbb{O}$. Moreover, its action cannot change the transfer matrix $\bbT$.

Equipped with \eqnref{eq:symmMPO}, one can study how MPO symmetries are fused or decomposed. As MPO symmetries $\{V_g | g \in G\}$ should form a representation of $G$, given two MPO symmetries $V_{g_2}$ and $V_{g_1}$, there exists a three-leg \emph{reduction} tensor $P_{g_1,g_2}$ (as a surjective map from two virtual bonds to one virtual bond) and its pseudo-inverse $P^+_{g_1,g_2}$ such that it merges $V_{g_1}$ and $V_{g_2}$ into $V_{g_2 g_1}$.
Taking further steps, one can imagine the composition of three MPO symmetries. At each end, there are two different ways to compose $(g_1, g_2, g_3)$ as illustrated in \figref{fig:PEPS}(c), and this leads to the phase factor $\omega$ (c.f. \eqnref{eq:2-cocycle-proj}):
\begin{align} \label{eq:3rdCohom}
P_{g_1, g_2 g_3} {}^{g_1}\!P_{g_2, g_3}  &= \omega(g_1,g_2,g_3) P_{g_1 g_2, g_3} P_{g_1, g_2},
\end{align}
where superscript on the left of $P_{g_2,g_3}$ corresponds to a complex conjugation if $g_1$ is anti-unitary~\cite{Else_2014, Chen_2011}. 
Note that the merging operation for $P^+$ on the left of the local MPO tensor produces the relative phase factor $\omega^{-1}$, which cancels the phase factor from merging $P$. 
The 3-cocycle $\omega$ is defined up to a gauge transformation as one can always redefine $P_{g,h} \mapsto \varphi(g,h) P_{g,h}$ and $P^+_{g,h} \mapsto \varphi^*(g,h) P^+_{g,h}$, i.e.
\begin{align}
    \omega(g_1,g_2,g_3) \mapsto \omega(g_1,g_2,g_3) \frac{\varphi(g_1g_2, g_3)\varphi(g_1, g_2)}{\varphi(g_1, g_2 g_3){}^{g_1}\!\varphi(g_2, g_3)}.
\end{align}
Furthermore, by considering two different reduction procedures for four MPO symmetries, we can arrive at the following consistency equation [c.f. Eq.~\eqref{eq:2-cocycle}]:
\begin{align}
    \frac{\omega(g_1, g_2, g_3) \omega(g_1, g_2 g_3, g_4) {}^{g_1}\!\omega(g_2,g_3,g_4)}{\omega(g_1 g_2, g_3, g_4) \omega(g_1, g_2, g_3 g_4)} = 1. \label{eq:3-cocycle}
\end{align}
As this is exactly the 3-cocycle condition, $\omega$ is given as a particular equivalence class of the third cohomology group $\calH^3(G,\U(1))$~\cite{Chen_2011}. 

Note that local tensors for vertical and horizontal MPO symmetries can form a projective representation: $V_g V_{h} = \lambda(g,h) V_{gh}$. This projective representation can be different for vertical and horizontal MPOs, each of which corresponds to a weak SPT protected by the translation $T_x$ and $T_y$, respectively. However, for an onsite symmetry group with a finite order, $\lambda^n = 1$ for some $n$. Accordingly, if we regroup $n$ sites along $x$ or $y$ direction, $\lambda \rightarrow \lambda^n = 1$ and this weak SPT structure can be trivialized. On the other hand, the equivalence class of $\omega$ is stable under regrouping multiple sites into a single site, thus $[\omega]$ provides a scale-invariant labeling of a semi-injective PEPS~\cite{Molnar_2018}.

\subsection{Classification}
Next, we classify the SPT phases of strong MPO-injective TNDO protected by a direct product of the strong symmetry $G$ and the weak symmetry $K$.
Again, we formulate the density operator as the doubled state and investigate the SPT phases of the doubled state subject to the physical constraints of the density matrix similar to Sec.~\ref{sec:classification}.
We focus on the \emph{strong} SPTs as in Ref.~\cite{directSplit} characterized by the third group cohomology of internal symmetries. 
The \emph{weak} SPT states in two dimension~\cite{directSplit}, which are a stack of 1d SPT states protected by the internal and translational symmetries, are constrained by the constraints discussed in Sec.~\ref{sec:classification}.

The doubled state $\kket{\rho}$ of the strong injective TNDO exhibits a symmetry $\bbG = \bbG_p \rtimes \bbZ_2^\bbH$, which consists of the physical symmetry $\bbG_p = G_l \times G_r \times K$ and the hermiticity symmetry $\mathbb{Z}_2^\mathbb{H}$, same as its 1d counterpart.
In two dimensions, the possible SPT phases are classified by the third cohomology group $\calH^3(\bbG, \U_T(1))$.

To begin, the anti-unitary $\bbZ_2^\bbH$ cannot protect a nontrivial 2d SPT phase by itself~\cite{chen2013symmetry}.
Moreover, as shown in Appendix~\ref{app:no_H_mixed_anomaly_2d}, one can pick a gauge of the reduction operator such that the 3-cocycle involving Hermiticity symmetry is trivial, i.e. $\omega(\bbH, \cdot, \cdot) = \omega(\cdot, \bbH, \cdot) = \omega(\cdot, \cdot, \bbH) = 1$.
This indicates the absence of mixed anomaly between $\bbZ_2^\bbH$ and $\bbG_p$.

The SPT phases are therefore only protected by the physical symmetry and are classified by the third cohomology group $\calH^3(\bbG_p, \U(1))$.
According to the K\"unneth decomposition,
\begin{align}
&\calH^3(\bbG_p, \U(1))  \nonumber\\
=&\calH^3(G_l, \U(1)) \oplus \calH^3(G_r, \U(1)) \oplus \calH^3(K, \U(1)) \nonumber \\
& \oplus \calH^1(G_l, \calH^2(K,\U(1))) \oplus \calH^1(G_r, \calH^2(K,\U(1)))\nonumber \\
& \oplus \calH^1(K, \calH^2(G_l,\U(1))) \oplus \calH^1(K, \calH^2(G_r,\U(1))) \nonumber \\
& \oplus \calH^1(G_l, \calH^2(G_r,\U(1))) \oplus \calH^1(G_r, \calH^2(G_l,\U(1))) \nonumber \\
&\oplus \calH^1(K,\calH^1(G_l,\calH^1(G_r,\U(1)))).\label{eq:kunneth2}
\end{align}
As shown in Appendix~\ref{app:direct-sum}, each component in the decomposition has a faithful representation given by the slant product.
In what follows, we discuss the constraints on the possible phases due to the Hermiticity, semi-positivity, and strong injectivity of the density operator.

\subsubsection{Absence of SPT phase protected by weak symmetry}
To show the absence of nontrivial SPT phases protected by the weak symmetry alone, we consider the transfer matrix $\bbT$ contracted along the circumference in the cylindrical geometry as in \figref{fig:PEPS}(b):
\begin{align}
\bbT = \begin{tikzpicture}[line cap=round,line join=round,x=1cm,y=1cm,baseline={(0, -\MathAxis pt)}]
\draw[black!80, line width=1.0] (-0.5,0) -- (4.5,0);
\foreach \x in {0,1,...,4}{
\draw[black!80, line width=1.0] (\x,-0.5) -- (\x, 0.5);
\draw[black!80, line width=1.0] (\x,0.1) -- (\x+0.4, 0.5);
\draw[black!80, line width=1.0] (\x+0.1,0) -- (\x+0.5, 0.4);
\draw[black, line width=1.0] (\x+0.5,0.4) arc(0:90:0.1);
\filldraw[fill=white, draw=black!80, line width=1.0] (\x-0.25,-0.25) rectangle (\x+0.25,0.25);
}
\node[black] at (2,0) {$\bbA$};
\end{tikzpicture}.\label{eq:2d_transfer_matrix}
\end{align}
The strong injectivity of the TNDO requires that the eigenvector associated with the largest eigenvalue of $\bbT$ is in the form of injective MPS
\begin{align}
|\phi[C]) = \begin{tikzpicture}[line cap=round,line join=round,x=1cm,y=1cm,baseline={(0, -\MathAxis pt)}]
\draw[black!80, line width=1.0] (-0.5,0) -- (4.5,0);
\foreach \x in {0,1,...,4}{
\draw[black!80, line width=1.0] (\x,0) -- (\x, 0.5);
\filldraw[fill=white, draw=black!80, line width=1.0] (\x-0.25,-0.25) rectangle (\x+0.25,0.25);
}
\node[black] at (2,0) {$C$};
\end{tikzpicture}.
\end{align}

The global symmetry transformation $U_{\sfk^{\rw}}$ with $k \in K$ on the 2d state does not change the transfer matrix $\bbT$; the transformations on the ket and the bra Hilbert space cancel after tracing over physical legs. 
We thus have
\begin{align}
    \bbT = V_{\sfk^{\rw}}^{-1} \bbT V_{\sfk^{\rw}}.
\end{align}
This indicates that $|\phi[C])$ is symmetric under the projective MPO symmetry, i.e. $V_{{\sfk^{\rw}}}$ acting on $|\phi[C])$ produces a global phase, $V_{\sfk^{\rw}} |\phi[C]) = e^{\ri\theta_{\sfk^{\rw}}}|\phi[C])$.
Hence, according to Ref.~\cite{Molnar_2018,Chen_2011}, there exists a reduction $Q_{{\sfk^{\rw}}}$ such that
\begin{align}
\begin{tikzpicture}[line cap=round,line join=round,x=1cm,y=1cm,baseline={(0, -\MathAxis pt)}]
\definecolor{myred}{RGB}{240,83,90};
\definecolor{myblue}{RGB}{73,103,189};
\definecolor{myturquoise}{RGB}{83,195,189};
\draw[black!80, line width=1.0] (-1,0) -- (1,0);
\draw[black!80, line width=1.0] (-0.7,1) -- (0.7,1);
\def\x{0}
\draw[black!80, line width=1.0] (\x,0) -- (\x, 1.5);
\filldraw[fill=white, draw=black!80, line width=1.0] (\x-0.3,-0.3) rectangle (\x+0.3,0.3);
\filldraw[fill=white, draw=black!80, line width=1.0] (\x-0.3,0.7) rectangle (\x+0.3,1.3);
\filldraw[fill=myblue!80, draw=myblue, line width=1.0] (-0.6,-0.1) rectangle (-0.8,1.1);
\filldraw[fill=myblue!80, draw=myblue, line width=1.0] (0.6,-0.1) rectangle (0.8,1.1);
\node[black] at (0,1.0) {$V_{\sfk^{\rw}}$};
\node[black] at (0,0) {$C$};
\node[myblue] at (1.2,1.0) {$Q_{\sfk^{\rw}}$};
\end{tikzpicture} 
= 
\begin{tikzpicture}[line cap=round,line join=round,x=1cm,y=1cm,baseline={(0, -\MathAxis pt)}]
\draw[black!80, line width=1.0] (-0.5,0) -- (0.5,0);
\def\x{0}
\draw[black!80, line width=1.0] (\x,0) -- (\x, 0.5);
\filldraw[fill=white, draw=black!80, line width=1.0] (\x-0.3,-0.3) rectangle (\x+0.3,0.3);
\node[black] at (0,0) {$C$};
\end{tikzpicture}.
\end{align}
The above property indicates that $|\phi[C])$ is symmetric under the weak symmetry transformation.
Thus, $|\phi[C])$ is non-anomalous under the weak symmetry, suggesting that the weak symmetry cannot protect nontrivial SPT phases by itself.
Specifically, consider local tensors $V_{\sfk^\rw_1}, V_{\sfk^\rw_2}, V_{\sfk^\rw_3}$ associated with the MPO for weak symmetry. 
The reduction among the three MPOs and the eigenstate $|\phi[C])$ satisfies a pentagon equation~\cite{Chen_2011}, 
\begin{align}
\begin{tikzpicture}
[line cap=round,line join=round,x=1cm,y=1cm,baseline={(0, -\MathAxis pt)}]
\definecolor{myred}{RGB}{240,83,90};
\definecolor{myblue}{RGB}{73,103,189};
\definecolor{myturquoise}{RGB}{83,195,189};
\draw[black!80, line width=1.0] (-0.4,0) -- (1.75,0);
\draw[black!80, line width=1.0] (-0.4,0.75) -- (0.7,0.75);
\draw[black!80, line width=1.0] (-0.4,1.5) -- (1.1,1.5);
\draw[black!80, line width=1.0] (-0.4,2.25) -- (1.5,2.25);
\filldraw[fill=myblue!80, draw=myblue, line width=1.0] (+0.6,-0.1) rectangle (+0.8,0.85);
\node[myblue] at (+0.7,1.1) {$Q_{\sfk^{\rw}_3}$};
\filldraw[fill=myblue!80, draw=myblue, line width=1.0] (+1.0,-0.1) rectangle (+1.2,1.6);
\node[myblue] at (+1.1,1.85) {$Q_{\sfk^{\rw}_2}$};
\filldraw[fill=myblue!80, draw=myblue, line width=1.0] (+1.4,-0.1) rectangle (+1.6,2.35);
\node[myblue] at (+1.5,2.6) {$Q_{\sfk^{\rw}_1}$};
\def\xtwo{3};
\draw[black!80, line width=1.0] (\xtwo-0.4,0) -- (\xtwo+1.35,0);
\draw[black!80, line width=1.0] (\xtwo-0.4,0.75) -- (\xtwo+0.7,0.75);
\draw[black!80, line width=1.0] (\xtwo-0.4,1.5) -- (\xtwo+1.1,1.5);
\draw[black!80, line width=1.0] (\xtwo-0.4,2.25) -- (\xtwo+0.7,2.25);
\filldraw[fill=myblue!80, draw=myblue, line width=1.0] (\xtwo+0.6,-0.1) rectangle (\xtwo+0.8,0.85);
\node[myblue] at (\xtwo+0.7,1.1) {$Q_{\sfk^{\rw}_3}$};
\filldraw[fill=myblue!80, draw=myblue, line width=1.0] (\xtwo+0.6,1.4) rectangle (\xtwo+0.8,2.35);
\node[myblue] at (\xtwo+0.7,2.6) {$P_{\sfk^{\rw}_1,\sfk^{\rw}_2}$};
\filldraw[fill=myblue!80, draw=myblue, line width=1.0] (\xtwo+1,-0.1) rectangle (\xtwo+1.2,1.6);
\node[myblue] at (\xtwo+1.3,1.85) {$Q_{\sfk^{\rw}_1\sfk^{\rw}_2}$};
\def\xthree{6-0.4};
\draw[black!80, line width=1.0] (\xthree-0.4,0) -- (\xthree+1.75,0);
\draw[black!80, line width=1.0] (\xthree-0.4,0.75) -- (\xthree+1.5,0.75);
\draw[black!80, line width=1.0] (\xthree-0.4,1.5) -- (\xthree+1.1,1.5);
\draw[black!80, line width=1.0] (\xthree-0.4,2.25) -- (\xthree+0.7,2.25);
\filldraw[fill=myblue!80, draw=myblue, line width=1.0] (\xthree+0.6,1.4) rectangle (\xthree+0.8,2.35);
\node[myblue] at (\xthree+0.7,2.6) {$P_{\sfk^{\rw}_1,\sfk^{\rw}_2}$};
\filldraw[fill=myblue!80, draw=myblue, line width=1.0] (\xthree+1,0.65) rectangle (\xthree+1.2,1.6);
\node[myblue] at (\xthree+1.5,1.85) {$P_{\sfk^{\rw}_1\sfk^{\rw}_2,\sfk^{\rw}_3}$};
\filldraw[fill=myblue!80, draw=myblue, line width=1.0] (\xthree+1.4,-0.1) rectangle (\xthree+1.6,0.85);
\node[myblue] at (\xthree+1.5,-0.45) {$Q_{\sfk^{\rw}_1\sfk^{\rw}_2\sfk^{\rw}_3}$};
\def\xfour{1.5};
\def\yy{-3.5}
\draw[black!80, line width=1.0] (\xfour-0.4,\yy+0) -- (\xfour+1.75,\yy+0);
\draw[black!80, line width=1.0] (\xfour-0.4,\yy+0.75) -- (\xfour+1.1,\yy+0.75);
\draw[black!80, line width=1.0] (\xfour-0.4,\yy+1.5) -- (\xfour+0.7,\yy+1.5);
\draw[black!80, line width=1.0] (\xfour-0.4,\yy+2.25) -- (\xfour+1.5,\yy+2.25);
\filldraw[fill=myblue!80, draw=myblue, line width=1.0] (\xfour+0.6,\yy+0.65) rectangle (\xfour+0.8,\yy+1.6);
\node[myblue] at (\xfour+0.7,\yy+1.85) {$P_{\sfk^{\rw}_2,\sfk^{\rw}_3}$};
\filldraw[fill=myblue!80, draw=myblue, line width=1.0] (\xfour+1,\yy-0.1) rectangle (\xfour+1.2,\yy+0.85);
\node[myblue] at (\xfour+1.1,\yy-0.35) {$Q_{\sfk^{\rw}_2\sfk^{\rw}_3}$};
\filldraw[fill=myblue!80, draw=myblue, line width=1.0] (\xfour+1.4,\yy-0.1) rectangle (\xfour+1.6,\yy+2.35);
\node[myblue] at (\xfour+1.5,\yy+2.6) {$Q_{\sfk^{\rw}_1}$};
\def\xfive{4.5};
\draw[black!80, line width=1.0] (\xfive-0.4,\yy+0) -- (\xfive+1.75,\yy+0);
\draw[black!80, line width=1.0] (\xfive-0.4,\yy+0.75) -- (\xfive+1.5,\yy+0.75);
\draw[black!80, line width=1.0] (\xfive-0.4,\yy+1.5) -- (\xfive+0.7,\yy+1.5);
\draw[black!80, line width=1.0] (\xfive-0.4,\yy+2.25) -- (\xfive+1.1,\yy+2.25);
\filldraw[fill=myblue!80, draw=myblue, line width=1.0] (\xfive+0.6,\yy+0.65) rectangle (\xfive+0.8,\yy+1.6);
\node[myblue] at (\xfive+0.7,\yy+0.4) {$P_{\sfk^{\rw}_2,\sfk^{\rw}_3}$};
\filldraw[fill=myblue!80, draw=myblue, line width=1.0] (\xfive+1,\yy+0.65) rectangle (\xfive+1.2,\yy+2.35);
\node[myblue] at (\xfive+1.1,\yy+2.6) {$P_{\sfk^{\rw}_1,\sfk^{\rw}_2\sfk^{\rw}_3}$};
\filldraw[fill=myblue!80, draw=myblue, line width=1.0] (\xfive+1.4,\yy-0.1) rectangle (\xfive+1.6,\yy+0.85);
\node[myblue] at (\xfive+1.9,\yy+1.1) {$Q_{\sfk^{\rw}_1\sfk^{\rw}_2\sfk^{\rw}_3}$};
\foreach \x in {0,3,6-0.4}{
\draw[black!80, line width=1.0] (\x,0) -- (\x, 2.75);
\filldraw[fill=white, draw=black!80, line width=1.0] (\x-0.25,-0.25) rectangle (\x+0.25,0.25);
\filldraw[fill=white, draw=black!80, line width=1.0] (\x-0.25,0.5) rectangle (\x+0.25,1);
\filldraw[fill=white, draw=black!80, line width=1.0] (\x-0.25,1.25) rectangle (\x+0.25,1.75);
\filldraw[fill=white, draw=black!80, line width=1.0] (\x-0.25,2) rectangle (\x+0.25,2.5);
\node[black] at (\x,2.25) {$V_{\sfk^{\rw}_1}$};
\node[black] at (\x,1.5) {$V_{\sfk^{\rw}_2}$};
\node[black] at (\x,0.75) {$V_{\sfk^{\rw}_3}$};
\node[black] at (\x,0) {$C$};
}
\def\y{-3.5}
\foreach \x in {1.5,4.5}{
\draw[black!80, line width=1.0] (\x,\y+0) -- (\x, \y+2.75);
\filldraw[fill=white, draw=black!80, line width=1.0] (\x-0.25,\y-0.25) rectangle (\x+0.25,\y+0.25);
\filldraw[fill=white, draw=black!80, line width=1.0] (\x-0.25,\y+0.5) rectangle (\x+0.25,\y+1);
\filldraw[fill=white, draw=black!80, line width=1.0] (\x-0.25,\y+1.25) rectangle (\x+0.25,\y+1.75);
\filldraw[fill=white, draw=black!80, line width=1.0] (\x-0.25,\y+2) rectangle (\x+0.25,\y+2.5);
\node[black] at (\x,\y+2.25) {$V_{\sfk^{\rw}_1}$};
\node[black] at (\x,\y+1.5) {$V_{\sfk^{\rw}_2}$};
\node[black] at (\x,\y+0.75) {$V_{\sfk^{\rw}_3}$};
\node[black] at (\x,\y+0) {$C$};
\draw[->,>=stealth, line width=1.0] (.4, -0.5) -- (0.7,-1) node[left]{$\varphi(\sfk^{\rw}_2,\sfk^{\rw}_3)$} -- (1,-1.5);
\draw[->,>=stealth, line width=1.0] (6.6, -0.9) -- (6,-1.5);
\node[black] at (6.8,-1.7){$\omega(\sfk^{\rw}_1,\sfk^{\rw}_2,\sfk^{\rw}_3)$};
\draw[->,>=stealth,black, line width=1.0] (3.3,-4) arc(180:360:0.3);
\node[black] at (3.6,-4.6){$\varphi(\sfk^{\rw}_1,\sfk^{\rw}_2\sfk^{\rw}_3)$};
\draw[->,>=stealth,black, line width=1.0] (1.8,2.7) arc(180:0:0.3);
\node[black] at (2.1,3.3){$\varphi(\sfk^{\rw}_1,\sfk^{\rw}_2)$};
\draw[->,>=stealth,black, line width=1.0] (4.4,2.7) arc(180:0:0.3);
\node[black] at (4.7,3.3){$\varphi(\sfk^{\rw}_1\sfk^{\rw}_2,\sfk^{\rw}_3)$};
}
\end{tikzpicture}\nonumber
\end{align}
For simplicity, we only draw the reduction on the right of the local tensors.
The relation indicates
\begin{align}
    \omega(\sfk^{\rw}_1,\sfk^{\rw}_2,\sfk^{\rw}_3) = \frac{\varphi(\sfk^{\rw}_1,\sfk^{\rw}_2\sfk^{\rw}_3)\varphi(\sfk^{\rw}_2,\sfk^{\rw}_3)}{\varphi(\sfk^{\rw}_1\sfk^{\rw}_2,\sfk^{\rw}_3)\varphi(\sfk^{\rw}_1,\sfk^{\rw}_2)}.\label{eq:2d_weak_symmetry_coboundary}
\end{align}
Thus, the 3-cocycle $\omega(\sfk^{\rw}_1,\sfk^{\rw}_2,\sfk^{\rw}_3)$ is coboundary, i.e. no SPT phases protected by the weak symmetry.

\subsubsection{Constraints from Hermiticity} 
The Hermiticity of the density matrix constrains the possible SPT phases protected by $\bbG_p$ by requiring the 3-cocycle to satisfy
\begin{align}
\omega(\sfg_1,\sfg_2,\sfg_3) = \omega^*(\bar{\sfg}_1, \bar{\sfg}_2, \bar{\sfg}_3),\label{eq:3-cocycle-hermiticity_constraint}
\end{align}
where  $\bar{\sfg} := \bbH \sfg \bbH$. This can be shown using the 3-cocycle condition in Eq.~\eqref{eq:3-cocycle}.

This relation puts constraints on the slant product
\begin{align}
[\omega(\sfg^l_1, \sfg^l_2, \sfg^l_3)] &= -[\omega(\sfg^r_1, \sfg^r_2, \sfg^r_3)], \\
[i_{\sfk^\rw}\omega(\sfg_1^l,\sfg_2^l)] &= -[i_{\sfk^\rw}\omega(\sfg_1^r,\sfg_2^r)] \\
[i_{\sfk_1^\rw} i_{\sfk_2^\rw}\omega(\sfg^l)] &= -[i_{\sfk_1^\rw} i_{\sfk_2^\rw}\omega(\sfg^r)].
\end{align}

\subsubsection{Constraints from semi-positivity} 
Here, we show the 2d SPT phases involving the mixed anomaly between the left and the right symmetry cannot exist due to the semi-positivity of the density matrix.
In the K\"unneth decomposition, these phases correspond to $\calH^1(G_l,\calH^2(G_r, \U(1)))$, $\calH^1(G_r,\calH^2(G_l, \U(1)))$, and $\calH^1(K, \calH^1(G_l, \calH^1(G_r, \U(1))))$.

To begin, the SPT phases in $\calH^1(G_l,\calH^2(G_r, \U(1)))$ has a faithful representation in terms of the slant product (as shown in Appendix~\ref{app:direct-sum})
\begin{align}
    i_{\sfg^l}\omega(\sfh^r_1, \sfh^r_2) := \frac{\omega(\sfg^l, \sfh^r_1, \sfh^r_2) \omega(\sfh^r_1, \sfh^r_2, \sfg^l)}{\omega(\sfh^r_1, \sfg^l, \sfh^r_2)}.
\end{align}
The slant product is a 2-cocycle, and its cohomology class $[i_{\sfg^l}\omega] \in \calH^2(G_r, \U(1))$.
At the same time, it is a one-dimensional representation of $G^l$, i.e. $i_{\sfg_1^l}\omega \; i_{\sfg_2^l}\omega = i_{\sfg_1^l\sfg_2^l}\omega$.
Hence, $[i_{\sfg^l}\omega] \in \calH^1(G_l,\calH^2(G_r, \U(1)))$.

The SPT state characterized by a non-trivial slant product $i_{\sfg^l}\omega$ exhibits a membrane order parameter 
\begin{align}
    \calM_{[i_{\sfg^l}\omega], \sfg^l}(\calR) := U_{\sfg^l}(\calR)\otimes W_{[i_{\sfg^l}\omega]}(\partial \calR),
\end{align} 
where $U_{\sfg^l}(\calR)$ is the partial symmetry transformation on region $\calR$, and $W_{[i_{\sfg^l}\omega]}(\partial \calR)$ is an entangler of 1d SPT in class $[i_{\sfg^l}\omega]$ on its boundary.
The membrane operator acquires a perimeter-law decaying expectation value in the doubled state with an increasing size of region $\calR$.

To show this, we consider a partial symmetry transformation $U_{\sfg^l}(\calR)$ in the region $\calR$ (upper half-plane), which fractionalizes into an MPO $V_{\sfg^l}(\partial \calR)$ on the virtual bond. 
As we will see, the MPO $V_{\sfg^l}(\partial \calR)$ acts like a one-dimensional SPT entangler.
Consider the symmetry transformations $U_{\sfh^r_2}$ followed by $U_{\sfh^r_1}$. 
As $V_{\sfg^l}$ is already inserted along $\rd \calR$, these symmetry transformations translate into a sequence of MPO symmetries $V_{\sfh^r_1} V_{\sfh^r_2} V_{\sfg^l} V_{\sfh^r_2}^{-1} V_{\sfh^r_1}^{-1}$, which can be reduced in two different ways as the following:
\begin{align}
\begin{tikzpicture}[scale=0.9,line cap=round,line join=round,x=1cm,y=1cm,baseline={(0, -\MathAxis pt)}]
\definecolor{myred}{RGB}{240,83,90};
\definecolor{myblue}{RGB}{73,103,189};
\definecolor{myturquoise}{RGB}{83,195,189};
\def\x{0};
\def\a{0.35};
\def\ar{0.1};
\foreach\y in {-6*\a,-3*\a, 0, 3*\a, 6*\a}{
\draw[black!80, line width=1.0] (-1.5*\a, \y) -- (2*\a, \y);
}
\draw[black!80, line width=1.0] (\x,-8*\a) -- (\x, 8*\a);
\filldraw[fill=white, draw=black!80, line width=1.0] (\x-\a,-2*\a) rectangle (\x+\a,-4*\a);
\filldraw[fill=white, draw=black!80, line width=1.0] (\x-\a,-5*\a) rectangle (\x+\a,-7*\a);
\filldraw[fill=white, draw=black!80, line width=1.0] (\x-\a,-\a) rectangle (\x+\a,\a);
\filldraw[fill=white, draw=black!80, line width=1.0] (\x-\a,2*\a) rectangle (\x+\a,4*\a);
\filldraw[fill=white, draw=black!80, line width=1.0] (\x-\a,5*\a) rectangle (\x+\a,7*\a);
\node[black] at (\x,6*\a) {$V_{\sfh^r_1}$};
\node[black] at (\x,3*\a) {$V_{\sfh^r_2}$};
\node[black] at (\x,-6*\a) {$V_{\sfh^r_1}^{-1}$};
\node[black] at (\x,-3*\a) {$V_{\sfh^r_2}^{-1}$};
\node[black] at (\x,0) {$V_{\sfg^l}$};
\draw[black!80, line width=1.0] (2*\a, 1.5*\a) -- (3*\a, 1.5*\a);
\draw[black!80, line width=1.0] (2*\a,-3*\a) -- (3*\a,-3*\a);
\draw[black!80, line width=1.0] (3*\a,-0.75*\a) -- (4*\a,-0.75*\a);
\draw[black!80, line width=1.0] (2*\a,6*\a) -- (4*\a,6*\a);
\draw[black!80, line width=1.0] (4*\a,2.625*\a) -- (5*\a,2.625*\a);
\draw[black!80, line width=1.0] (2*\a,-6*\a) -- (5*\a,-6*\a);
\draw[black!80, line width=1.0] (5*\a,0) -- (6*\a,0);
\filldraw[fill=myblue!80, draw=myblue, line width=1.0] (2*\a-\ar,-\ar) rectangle (2*\a+\ar,3*\a+\ar);
\filldraw[fill=myblue!80, draw=myblue, line width=1.0] (3*\a-\ar,1.5*\a+\ar) rectangle (3*\a+\ar,-3*\a-\ar);
\filldraw[fill=myblue!80, draw=myblue, line width=1.0] (4*\a-\ar,-0.75*\a-\ar) rectangle (4*\a+\ar,6*\a+\ar);
\filldraw[fill=myblue!80, draw=myblue, line width=1.0] (5*\a-\ar,2.625*\a+\ar) rectangle (5*\a+\ar,-6*\a-\ar);
\end{tikzpicture} 
= 
\begin{tikzpicture}[scale=0.9,line cap=round,line join=round,x=1cm,y=1cm,baseline={(0, -\MathAxis pt)}]
\definecolor{myred}{RGB}{240,83,90};
\definecolor{myblue}{RGB}{73,103,189};
\definecolor{myturquoise}{RGB}{83,195,189};
\def\x{0};
\def\a{0.35};
\def\ar{0.1};
\foreach\y in {-6*\a,-3*\a, 0, 3*\a, 6*\a}{
\draw[black!80, line width=1.0] (-1.5*\a, \y) -- (2*\a, \y);
}
\draw[black!80, line width=1.0] (\x,-8*\a) -- (\x, 8*\a);
\filldraw[fill=white, draw=black!80, line width=1.0] (\x-\a,-2*\a) rectangle (\x+\a,-4*\a);
\filldraw[fill=white, draw=black!80, line width=1.0] (\x-\a,-5*\a) rectangle (\x+\a,-7*\a);
\filldraw[fill=white, draw=black!80, line width=1.0] (\x-\a,-\a) rectangle (\x+\a,\a);
\filldraw[fill=white, draw=black!80, line width=1.0] (\x-\a,2*\a) rectangle (\x+\a,4*\a);
\filldraw[fill=white, draw=black!80, line width=1.0] (\x-\a,5*\a) rectangle (\x+\a,7*\a);
\node[black] at (\x,6*\a) {$V_{\sfh^r_1}$};
\node[black] at (\x,3*\a) {$V_{\sfh^r_2}$};
\node[black] at (\x,-6*\a) {$V_{\sfh^r_1}^{-1}$};
\node[black] at (\x,-3*\a) {$V_{\sfh^r_2}^{-1}$};
\node[black] at (\x,0) {$V_{\sfg^l}$};
\draw[black!80, line width=1.0] (2*\a, 4.5*\a) -- (3*\a, 4.5*\a);
\draw[black!80, line width=1.0] (2*\a,0) -- (3*\a,0);
\draw[black!80, line width=1.0] (4*\a,0) -- (5*\a,0);
\draw[black!80, line width=1.0] (3*\a,2.25*\a) -- (4*\a,2.25*\a);
\draw[black!80, line width=1.0] (2*\a,-4.5*\a) -- (4*\a,-4.5*\a);
\filldraw[fill=myblue!80, draw=myblue, line width=1.0] (2*\a-\ar,3*\a-\ar) rectangle (2*\a+\ar,6*\a+\ar);
\filldraw[fill=myblue!80, draw=myblue, line width=1.0] (2*\a-\ar,-3*\a+\ar) rectangle (2*\a+\ar,-6*\a-\ar);
\filldraw[fill=myblue!80, draw=myblue, line width=1.0] (3*\a-\ar,-\ar) rectangle (3*\a+\ar,4.5*\a+\ar);
\filldraw[fill=myblue!80, draw=myblue, line width=1.0] (4*\a-\ar,2.25*\a+\ar) rectangle (4*\a+\ar,-4.5*\a-\ar);
\end{tikzpicture}
\;\Lambda(\sfh^r_1, \sfh^r_2, \sfg^l),
\end{align}
where
\begin{align}
\Lambda(\sfh^r_1, \sfh^r_2, \sfg^l) := \frac{\omega(\sfh^r_1,\sfh^r_2\sfg^l, (\sfh^r_2)^{-1})\omega(\sfh^r_1,\sfh^r_2,\sfg^l)}{\omega(\sfh^r_1\sfh^r_2\sfg^l,(\sfh^r_1)^{-1},(\sfh^r_2)^{-1})}.
\end{align}
One can show that the phase difference with and without the MPO $V_{\sfg^l}$ inserted is given by the slant product (up to a coboundary term), i.e.
\begin{align}
\frac{\Lambda(\sfh^r_1, \sfh^r_2, \sfg^l)}{\Lambda(\sfh^r_1, \sfh^r_2, \mathsf{e})} = i_{\sfg^l}\omega(\sfh_1^r, \sfh_2^r).
\end{align}
Hence, the MPO $V_{\sfg^l}$ transforms as a 1d SPT state
labeled by $[i_{\sfg^l}\omega]$.

For the doubled state of TNDO represented by a semi-injective PEPS, one can always pull the MPO $V_{\sfg^l}$ to the Hilbert space of nearby physical qubits\footnote{This is because the MPO $V_{\sfg^l}$ acts on the Hilbert space of the minimal rank decomposition of four-qubit state $\psi$ in semi-injective PEPS. The operator in this space is mapped injectively to the physical Hilbert space~\cite{Molnar_2018}. Under the MPO-injectivity condition, the injective portion of $V_g$ can be still pulled onto the physical space~\cite{_ahino_lu_2021, Williamson_2016}.}, i.e. there always exists a 1d MPO $\tilde{W}^{-1}$ such that
\begin{align}
\begin{tikzpicture}[line cap=round,line join=round,x=1cm,y=1cm,baseline={(0, -\MathAxis pt)}]
\definecolor{myred}{RGB}{240,83,90};
\definecolor{myblue}{RGB}{73,103,189};
\definecolor{myturquoise}{RGB}{83,195,189};
\draw[black!80, line width=1.0] (-0.5,0) -- (2.5,0);
\foreach \x in {0,1,2}{
\draw[black!80, line width=1.0] (\x,-0.5) -- (\x, 0.7);
\draw[black!80, line width=1.0] (\x,0.1) -- (\x+0.4, 0.5);
\draw[black!80, line width=1.0] (\x+0.1,0) -- (\x+0.5, 0.4);
\filldraw[fill=myturquoise,draw=myturquoise] (\x-0.1,0.6) rectangle (\x+0.1,0.4);
\filldraw[fill=white, draw=black!80, line width=1.0] (\x-0.25,-0.25) rectangle (\x+0.25,0.25);
}
\draw[myturquoise,dashed,line width=1.0] (-0.5,0.5) -- (2.5, 0.5);
\node[myturquoise] at (1.4,0.8) {$V_{\sfg^l}$};
\node[black] at (1,0) {$\bbA$};
\end{tikzpicture}
=
\begin{tikzpicture}[line cap=round,line join=round,x=1cm,y=1cm,baseline={(0, -\MathAxis pt)}]
\definecolor{myred}{RGB}{240,83,90};
\definecolor{myblue}{RGB}{73,103,189};
\definecolor{myturquoise}{RGB}{83,195,189};
\draw[black!80, line width=1.0] (-0.5,0) -- (2.5,0);
\foreach \x in {0,1,2}{
\draw[black!80, line width=1.0] (\x,-0.5) -- (\x, 0.5);
\draw[black!80, line width=1.0] (\x,0.1) -- (\x+0.6, 0.7);
\draw[black!80, line width=1.0] (\x+0.1,0) -- (\x+0.7, 0.6);
\draw[myred, fill=myred] (\x+0.45, 0.45) circle (0.13);
\filldraw[fill=white, draw=black!80, line width=1.0] (\x-0.25,-0.25) rectangle (\x+0.25,0.25);
}
\draw[myred, dashed, line width=1.0] (-0.05, 0.45) -- (2.95,0.45);
\node[myred] at (3.2,0.8) {$\tilde{W}^{-1}$};
\end{tikzpicture}.
\end{align}
Hence, the symmetry transformation $U_{\sfg^l}$ decorated with $\tilde{W}$ acquires a unit  expectation value, $\lAngle \tilde{W}(\rd \calR) U_{\sfg^l}(\calR) \rAngle\,{=}\,1$.
The MPO $\tilde{W}$ is invertible as $V_{\sfg^l}$ is invertible; its inverse creates an SPT in class $-[i_{\sfg^l}\omega]$. Furthermore, as $\sfg^n\,{=}\,1$ for some finite $n$,  $\tilde{W}^n\,{=}\,1$ on the physical wavefunction, and the eigenvalues of $\tilde{W}$ are of unit magnitude. 

To further proceed, we want to find an MPO that satisfies the three conditions: (1) It is an SPT entangler purely supported on the right Hilbert space that creates the 1d SPT state labeled as $[i_{\sfg^l}\omega]$, (2) its MPO bond dimension is finite, and (3) its overlap per site with $\tilde{W}$ is finite, i.e., the overlap with $\tilde{W}(\calC)$ decays as a perimeter law. The general recipe to construct such an MPO is the following. First, we find a unitary operator $U$ acting on the Hilbert space of a left physical site $\cH^l_i$ such that $W\,{:=}\,\tr_{\cH^l_i}(\tilde{W} U^\dagger) \neq 0$ as illustrated in \figref{fig:cylinder}(a). Then, the MPO $U \otimes W$ in \figref{fig:cylinder}(b) must have the same MPO bond dimension as $\tilde{W}$ and creates the same SPT state as the symmetry action by $G_r$ on this MPO translates into the same gauge transformation as $\tilde{W}$. Furthermore, one can always find $U$ where $W \neq 0$, ensuring the overlap between $U \otimes W$ and $\tilde{W}$ is finite per site. 
We thus show the existence of a membrane order parameter
\begin{align}
\calM_{[i_{\sfg^l}\omega],\sfg^l}(\calR) = [U(\rd \calR) \cdot U_{\sfg^l}(\calR) ] \otimes W_{[i_{\sfg^l}\omega]}(\rd \calR),
\end{align}
which decays with a perimeter law.

\begin{figure}[t]
\centering
\includegraphics[width=1\columnwidth]{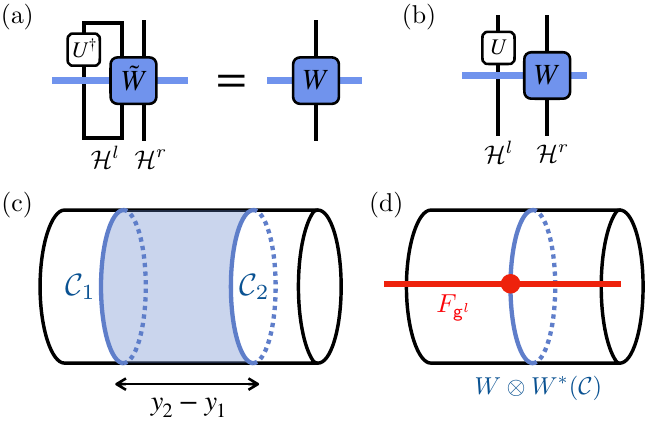}
\caption{ 
(a) Membrane order parameter with a perimeter law decay on shaded (blue) region implies that the correlation between two loops $\calC_1$ and $\calC_2$ is constant with the separation $y_2-y_1$. (b) The $\sfg^l \in G_l$ symmetry flux insertion through the cylinder $F_{\sfg^l}$ corresponds to the introduction of a $\sfg^l$-defect along the red line, which has a nontrivial commutation relation with $W \otimes W^*(\calC)$.}
\label{fig:cylinder}
\end{figure}

To further proceed, consider a cylindrical geometry with circumference size $L$. The existence of a membrane order parameter on the shaded region in \figref{fig:cylinder}(c) decaying with perimeter-law in the doubled state indicates the long-range correlation following the Cauchy-Schwarz inequality, which is similar to the derivation of Eq.~\eqref{eq:inequality}:
\begin{align}
    &\langle\!\langle W\otimes W^*(\calC_{1})\;W\otimes W^*(\calC_{2}) \rangle\!\rangle \xrightarrow{|y_2-y_1| \to \infty} \text{const.}
\end{align}
where $\calC_i$ is the loop around the cylinder at the location $y_i$.
With this, there are two possible scenarios for the behavior of $W\otimes W^*(\calC)$: 

\vspace{5pt}
\noindent $(i)$ Non-vanishing connected correlation: if the expectation value $\lAngle W\otimes W^*(\calC_{y}) \rAngle$ does not match its correlation in the infinite separation limit, there is a long-ranged connected correlation of the string operator $W\otimes W^*(\calC_{y})$, which contradicts the uniqueness of the largest eigenvalue of the transfer matrix as in Def.~\ref{def:strong_semi_injective}. 

\vspace{5pt}
\noindent $(ii)$ Vanishing connected correlation: in this case, the expectation value $\lAngle W\otimes W^*(\calC_{y}) \rAngle$  matches the correlation in the infinite separation limit. Accordingly, there is a perimeter-law decaying loop operator $W\otimes W^* (\calC)$ that has a mixed anomaly with $G_r$ and $G_l$. This is because $W\otimes W^* (\calC)$ creates 1d SPT states labeled by $[i_{\sfg^l} \omega]$ on the left Hilbert space and $-[i_{\sfg^l} \omega]$ on the right Hilbert space, and such entangler with finite MPO bond dimension should break the symmetry $G_r$ and $G_l$ locally\footnote{For example, a finite-depth circuit creating 1d $\mathbb{Z}_2 \times \mathbb{Z}_2$ cluster SPT state can be constructed as the product of controlled-Z gates, which respects the symmetry on the periodic boundary condition as a whole, but breaks the symmetry on the open boundary condition at each end.}. 

As a crucial step, we introduce the flux insertion operator $F_{\sfg^l}$ along the cylinder, which gives rise to a $\sfg^l$ (or $\sfg^r$) symmetry defect along the vertical line on the cylinder surface~\cite{Barkeshli_2019}, see \figref{fig:cylinder}(b). Such a line of $\sfg^l$-defect translates into the gauge transformation associated with $\sfg^l$ on the virtual legs crossed by the line. Consider a sequence of flux insertion operators $\{ F_{\sfg_i} \}$ such that $\prod_i \sfg_i = 1$ but its projective representation in $[i_{\sfg^l} \omega]$ generates $\prod_i V_{\sfg_i} = e^{i \eta} \neq 1$. As the doubled state $|\rho \rAngle$ is symmetric under $G_l$, if the doubled state is the unique symmetric ground state of a gapped parent Hamiltonian, the sequence of symmetry flux insertions that trivializes should map the state into itself, i.e., $\calF |\rho \rAngle = e^{i \phi} | \rho \rAngle$ where $\calF = \prod_i F_{\sfg^l_i}$. $\calF$ further satisfies the following relation
\begin{align} \label{eq:anomaly}
    \calF \cdot [W \otimes W^*(\calC)] = e^{i \eta}   [W \otimes W^*(\calC)] \cdot \calF
\end{align}
since the line defect that crosses with $W \otimes W^*(\calC)$ at a single point must give rise to a phase factor associated with the projective representation.  
However, plugging in $\calF |\rho \rAngle = e^{i \phi} | \rho \rAngle$ to \eqnref{eq:anomaly}, the constant expectation value of $\lAngle W \otimes W^*(\calC) \rAngle$ implies that $e^{i \eta} = 1$, a contradiction to the presence of the anomaly. 
We thus argue the absence of SPT phases in $\calH^1(G_l, \calH^2(G_r, \U(1)))$, and similarly for the phases in $\calH^1(G_r, \calH^2(G_l, \U(1)))$.

Next, we consider a non-trivial SPT phase in $\calH^1(K,\calH^1(G_l,\calH^1(G_r, \U(1))))$. 
Such an SPT phase is characterized by a decorated membrane operator in the doubled state
\begin{align}
    \calM_{[i_{\sfk^{\rw}}\omega],\sfk^{\rw}} := U_{\sfk^{\rw}}(\calR)\otimes W_{[i_{\sfk^{\rw}}\omega]}(\partial \calR),
\end{align}
Here, the weak symmetry transformation $U_{\sfk^\rw}$ acts from both left and right on the density matrix and preserves the positivity of the density matrix, so is $W_{[i_{\sfk^\rw}\omega]}$.
However, when acting on a symmetric product state, $W_{[i_{\sfk^\rw}\omega]}$ creates an SPT state with mixed anomaly between the left and the right symmetry. 
As shown in Sec.~\ref{sec:positivity_constraint}, the resulting SPT state violates the positivity of the density matrix, indicating that $W_{[i_{\sfk^\rw}\omega]}$ cannot be a positive map.
Thus, there cannot be such SPT phases.

\subsubsection{Summary}

To summarize, the Hermiticity and semi-positivity of the density matrix constrain the possible SPT phases in the doubled state of semi-injective TNDO, The resulting classification is given by $\calH^3(G_l, \U(1))\oplus\calH^2(G_l, \calH^1(K,\U(1)))\oplus\calH^1(G_l,\calH^2(K,\U(1)))$, which agrees with that obtained for average SPT in Ref.~\cite{ma2023average}.

To close, we comment on the constraints of possible doubled-state 2d SPT phases protected by strong 1-form and 0-form symmetry.
In this case, it remains unclear regarding the proper tensor network density operator that represents such states.
However, one can still constrain the mixed anomaly between the left and the right symmetry based on the semi-positivity of the density matrix.
Since the 1-form symmetry cannot protect SPT phases alone~\cite{ZhuLanWen}, the only possible mixed anomaly involving 1-form symmetry is between the left (right) 0-form symmetry and the right (left) 1-form symmetry. 
Phases with such anomaly would feature the string order given by partial 1-form symmetry decorated by 0-form charge, which according to the inequality in Eq.~\eqref{eq:inequality}, would imply long-range charge correlation, contradicting the assumption of short-range correlation.

We leave the generalization of our results to a broader class of 2d mixed states as well as the rigorous study of possible SPT mixed states protected by higher-form symmetries to future study.

\section{Discussion}\label{sec:discussion}
In this work, we introduce a broad class of one-dimensional mixed states described by strong injective MPDO, generalizing the concept of injective MPS. 
We define and classify the SPT phases in the strong injective MPDO $\rho$ protected by both strong and weak symmetry according to the cohomology class of projective representation in the corresponding doubled state.
In the strongly symmetric MPDO that admits a local purification, we show that the SPT phases have a compatible definition based on the equivalence class of mixed states under non-degenerate local channels.
Finally, we extend our framework to study SPT phases in two-dimensional mixed states, where we identify a strong semi-injective TNDO as a suitable generalization, and classify the possible SPT phases.

Notably, our definition of the mixed-state SPT phases is different from the definition based on the ``two-way" connection~\cite{coser2019classification,ma2023average,sang2023mixed}: two mixed states $\rho_{1,2}$ are in the same SPT phase if there exist symmetric local quantum channels $\calN_{1,2}$ such that $\rho_1 = \calN_2[\rho_2]$ and $\rho_2 = \calN_1[\rho_1]$.
Instead, we define the equivalence classes of symmetric mixed states based on a simpler ``one-way" connection by a \emph{non-degenerate} symmetric channel.
The non-degenerate property of the local channel guarantees that the output state $\calN[\rho]$ retains the signatures of the anomaly in the input state $\rho$, distinguishing it from SPT states in other phases.
We justify this simplified equivalence relation by proving that for the locally purifiable strong injective MPDOs, the equivalence class under the ``one-way" connection reproduces the cohomology classification.
Whether the equivalence class under the two-way connection produces a definition of SPT phases compatible with the definitions in this work is left for future research.

The mixed-state SPT phases in this work are preserved under non-degenerate local quantum channels and characterized by observables in a single-copy density matrix, e.g., string or membrane order parameters.
These phases are different from the mixed-state phases characterized by their information theoretical properties, such as the decohered SPT states originally protected by higher-form symmetry~\cite{raussendorf2005long, Roberts2020, lee2022symmetry} and the decohered topologically ordered states~\cite{bao2023mixed,fan2023diagnostics,lee2023quantum}, which can change under a non-degenerate local channel circuit and may still require ``two-way" connection in the definition~\cite{sang2023mixed}.
It remains open whether one can generalize the tensor network formulation to describe non-linear functions of the density matrix and characterize the information theoretical phases in mixed states.

Our results also open several directions for future study.
First, one can apply our framework to study the mixed-state SPT phases protected by a non-direct product of the weak and the strong symmetry, which may exhibit a richer phase structure~\cite{ma2023topological}.
Moreover, it is worth generalizing our formulation to study the fermionic SPT phases in mixed states, which are of a different classification.

\begin{acknowledgments}
We thank Ehud Altman, Ruihua Fan, Tarun Grover, Shengqi Sang, Ruiqi Bai, Xie Chen, Yuchen Guo, Shang Xu, Jiangnan Xiong, Yu-An Chen, Nathanan Tantivasadakarn, and Dominic Williamson for helpful discussions.
We thank especially Ehud Altman for the encouragement in the pursuit of this project. 
H.X. also thanks Tingyu Zhu for practicing piano and violin together.
We are especially grateful to Ruochen Ma for pointing out a potential mistake and Meng Cheng for help identifying it in the Appendix~\ref{app:no_H_mixed_anomaly_2d} of the early version of the manuscript.
J.Y.L is supported by a Simons investigator fund and a faculty startup grant at the University of Illinois, Urbana-Champaign.
Y.B. is supported in part by grant NSF PHY-2309135 and the Gordon and Betty Moore Foundation Grant No. GBMF7392 to the Kavli Institute for Theoretical Physics (KITP).

\emph{Note:} Upon completion of the manuscript, we became aware of several independent works that are broadly related. Ref.~\cite{ma2024symmetry} discusses the constraints on the possible SPT phases in the doubled state. Ref.~\cite{wang2024anomaly} argues that (2+1)d SPT mixed states with different anomalies cannot be connected using finite-depth symmetric channel circuits. Ref.~\cite{ZhenBi_to_appear} studies mixed-state SPT phases in locally purifiable density operators. We thank the authors of Ref.~\cite{ZhenBi_to_appear} for informing us of their work in advance.
\end{acknowledgments}

\bibliography{refs}

\appendix

\section{Local non-degenerate channel preserves strong injectivity}\label{app:injectivity_local_channel}

In this appendix, we prove that a local channel $\calN$ (described by a channel circuit of finite depth as shown in Fig.~\ref{fig:local_channel}) preserves the strong injectivity of MPDO $\rho$, i.e. $\rho$ is strong injective if and only if $\calN[\rho]$ is strong injective.

Before we delve into the mathematical proof.
In pure ground states, we have a similar statement that local unitary circuits preserve the injectivity of MPS.
The statement holds because the exponentially decaying connected correlation functions cannot change dramatically under local unitaries due to the Lieb-Robinson bound.
Alternatively, one can also prove it based on the parent Hamiltonian~\cite{perez2006matrix}.
Here, we prove that a non-degenerate local channel circuit preserves condition 2 in Def.~\ref{def:stronginjective} based on the correlation functions in the state.
Showing that the non-degenerate local channel also preserves condition 1 is much more involved.
When we treat the MPDO as an MPS in the doubled space, quantum channels act as non-unitary operations.
This results in a doubled state with a non-Hermitian parent Hamiltonian and can change the correlation functions dramatically.
However, the characteristics of injectivity only involve the linear structure of the Hilbert space but not the inner product. 
We here modify the proof in Ref.\cite{perez2006matrix} to avoid considering the parent Hamiltonian and prove that the injectivity is preserved when local unitaries are replaced by local injective maps.

First, we consider the condition 2 in Def.~\ref{def:stronginjective}. 
We show that the transfer matrix $\bbT$ of $\calN[\rho]$ has a unique largest eigenvalue if and only if $\bbT$ of $\rho$ has a unique largest eigenvalue.
To prove this statement, the physical connected correlation function $\calN[\rho]$ takes the form
\begin{align}
&\tr(\calN[\rho] O_x O_{x'}) - \tr(\calN[\rho] O_x) \tr(\calN[\rho] O_{x'}) \nonumber \\
=&\tr(\rho \calN^\dagger[O_x O_{x'}]) - \tr(\rho \calN^\dagger[O_x]) \tr(\rho \calN^\dagger[O_{x'}]) 
\end{align}
where $\calN^\dagger[\cdot] := \sum_i K_i^\dagger (\cdot) K_i$ is the adjoint channel.
For a local channel circuit, $\calN^\dagger[O]$ has local support.
When the separation $|x-x'|$ is sufficiently greater than the circuit depth, $\calN^\dagger[O_xO_{x'}] = \calN^\dagger[O_x] \calN^\dagger[O_{x'}]$.
Since $\calN^\dagger$ is non-degenerate, all the connected correlations in $\rho$ are exponentially decaying if and only if the same is true for correlations in $\calN[\rho]$.

Next, to address condition 1 in Def.~\ref{def:stronginjective}, we prove that the injectivity of MPS is preserved under a local circuit consisting of not only unitary maps but also invertible linear maps. 
This implies that the doubled state $\kket{\rho}$ remains an injective MPS after applying a local channel circuit.

To begin, we introduce a concept called \emph{locally reconstructible state}.\footnote{Mathematically, constructions below fit the concepts of presheaf and Galois connection. Here, we only present minimum definitions for simplicity.}
Consider a given quantum state $\ket{\Psi} \in H_{\Lambda}$,  where $H_{\Lambda}$ is the Hilbert space of a one-dimensional spin chain $\Lambda$ of size $L$.
We define a subspace $\calL_U(\Psi) \subset H_\Lambda$ as the image of reduced density matrix on subset $U$, i.e. $\calL_U(\Psi) := \Im \tr_{\Lambda-U}\ketbra{\Psi}$.
The subspace $\calL_U(\Psi)$ represents the space of possible quantum states on the subsystem $U$ when the entire system is in quantum state $\ket{\Psi}$.

\begin{definition}
Consider a quantum state $\ket{\Psi} \in H_\Lambda$. The state $\ket{\Psi}$ is locally reconstructible with a reconstructible length (not greater than) $d$ if $\calL_U(\Phi) \subset \calL_U(\Psi)$ for any subregion $U$ of size (not greater than) $d$ has a unique nonzero solution $\ket{\Phi} = \ket{\Psi}$.
\end{definition}

We note that among all solutions to $\calL_U(\Phi) \subset \calL_U(\Psi)$, there is the largest element $\calL_U(\Psi)\otimes H_{\Lambda -U}$, which includes all the other solutions as its subspace.
Thus, the condition for a locally reconstructible state can be equivalently formulated as $\ket{\Psi} = \overline{\calL_d{(\Psi)}}$, and we call $\overline{\calL_d{(\Psi)}}$ the closure of $\calL_U(\Psi)$ on a family of sets\footnote{We use same notation for a state $\Psi$ and its corresponding one-dimensional subspace.},
\begin{align}
    \overline{\calL_d{(\Psi)}} := \bigcap_U \calL_U(\Psi) \otimes H_{\Lambda -U},
\end{align}
where $U$ runs over all the interval of length no greater than $d$.

Physically, a locally reconstructible state is fully determined by the reduced density matrices on subsets of size smaller than $d$.
According to this definition, the state with a reconstructible length $d = 1$ is an unentangled product state, such as $\ket{00\cdots 0}$,
, while a cat state $(\ket{00\cdots 0} + \ket{11\cdots 1})/\sqrt{2}$ is not locally reconstructible.

\begin{theorem}
Let $\calE: H_{\Lambda} \mapsto H_{\Lambda}$ be a local invertible map (arranged in a brickwork circuit). $\ket{\Psi}$ is locally reconstructible if and only if $\calE\ket{\Psi}$ is locally reconstructible.
\end{theorem}
\begin{proof}
Since $\calE$ is invertible, it is sufficient to prove that $\ket{\Psi}$ being locally reconstructible implies that $\calE\ket{\Psi}$ is locally reconstructible.
The key observation here is for any operator $O$ supported in a finite region $d$, $\calE O \calE^{-1}$ supports on an interval of length no greater than $d+2t$, where $t$ is the circuit depth of the invertible map $\calE$.

Let $d$ be the reconstructible length of $\ket{\Psi}$. 
For every interval $U$ of length $d$, we define a projection operator $P_U$ with image $\calL_U(\Psi)$. 
Then, we have
\begin{align}
\calL_U(\Psi)\otimes H_{\Lambda-U}=\Im (P_U\otimes 1_{X-U}).
\end{align} 
Appling $\calE$ to the both sides of the equation, we have 
\begin{align}
\calE(\calL_U(\Psi)\otimes H_{\Lambda-U})=\Im(\calE(P_U\otimes 1_{\Lambda-U})\calE^{-1}).
\end{align} 
On the right hand side, the invertible map acts on the projection operator yields $\calE(P_U\otimes 1_{\Lambda-U})\calE^{-1}=Q_U\otimes 1_{\Lambda-V_U}$, where $Q_U$ is supported on an interval $V_U$ with length $d+2t$. 
We denote $\Im Q_U =Z_U$, then 
\begin{align}
    \calE(\calL_U(\Psi)\otimes H_{\Lambda-U})=Z_U\otimes H_{\Lambda-V_U}.
\end{align}
Because $\ket{\Psi}\subset \calL_U(\Psi)\otimes H_{\Lambda-U}$, we have
\begin{align}
    \calL_{V_U}(\calE \ket{\Psi})\subset \calL_{V_U}\left(\calE(\calL_U(\Psi)\otimes H_{\Lambda-U})\right)=Z_U.
\end{align}
Taking the intersection of all $U$ yields
\begin{align}
    &\bigcap_U \calL_{V_U}(\calE \ket{\Psi})\otimes H_{\Lambda-V_U}\subset\bigcap_U Z_U\otimes H_{\Lambda-V_U}\nonumber\\
    =&\bigcap_U \calE\left(\calL_U(\ket{\Psi})\otimes H_{\Lambda-U}\right)=\calE\left(\bigcap_U\calL_U(\Psi)\otimes H_{\Lambda-U}\right)\nonumber\\=& \calE\ket{\Psi}.
\end{align}
Considering that $\calE\ket{\Psi} \subset \calL_{V_U}(\calE\ket{\Psi})\otimes H_{\Lambda-V_U}$, we have
\begin{align}
    \calE\ket{\Psi} = \bigcap_U \calL_{V_U}(\calE\ket{\Psi})\otimes H_{\Lambda-V_U},
\end{align}
which indicates that $\calE\ket{\Psi}$ is locally reconstructible, and the reconstructible length is not greater than $d+2t$.
\end{proof}

\begin{lemma}
Let $\ket{\Psi[A]}$ be a translationally invariant MPS. The reconstructible length of $\ket{\Psi}$ is not greater than $2$.
\end{lemma}
\begin{proof}
The MPS is in the form
\begin{align}
    \Psi^{i_1\cdots i_L}=\tr(A^{i_1}A^{i_2}\cdots A^{i_L}).
\end{align}
Assume $\Phi$ is a state satisfying $\calL_U(\Phi)\subset \calL_U(\Psi)$ for all interval of length $2$. 
First, we take $U=[1,2]$, because $A^{i_1}A^{i_2}$ is an injective map from the virtual to the physical space, $\calL_U(\Phi)\subset \calL_U(\Psi)$ indicates the existence of $X$ such that $\Phi=\operatorname{tr}A^{i_1}A^{i_2}X^{i_3\cdots i_L}$. 
Similarly, for $U = [2,3]$, there exist $Y$ such that $\Phi=\tr A^{i_2}A^{i_3}Y^{i_4\cdots i_Li_1}$. 
Since $A^{i_2}$ is injective, we have $X^{i_3\cdots i_L}A^{i_1}=A^{i_3}Y^{i_4\cdots i_Li_1}$. 
Because $A$ is injective, there exists $A^{-1}$ such that $\sum_{i_1} A^{i_1}(A^{-1})^{i_1}=\mathds{1}$ on the virtual bond. 
So $X^{i_3\cdots i_L}= \sum_{i_1} A^{i_3}Y^{i_1i_4\cdots i_L}(A^{-1})^{i_1}$. 
Define $Z^{i_4\cdots i_L} :=\sum_{i_1}Y^{i_1i_4\cdots i_L}(A^{i_1})^{i_1}$, we then have $\Phi^{i_1\cdots i_L}=\tr A^{i_1}A^{i_2}A^{i_3}Z^{i_4\cdots i_L}$. 
We can repeat this process and obtain the final result $\Phi^{i_1\cdots i_L}=\lambda\tr(A^{i_1}A^{i_2}\cdots A^{i_L})$, where $\lambda$ is a phase factor. 
Thus, $\ket{\Psi}$ has a local reconstructible length not greater than $2$.
\end{proof}

\begin{theorem}
Let $\ket{\Psi}$ be a translationally invariant MPS. The dimension of $\overline{\calL_d(\Psi)}$ for large enough $d$ is given by the number of normal bases in the canonical form of $\ket{\Psi}$.
\end{theorem}
\begin{proof}
Let $k$ be the number of normal bases in the canonical form of MPS $\Psi$. We have the following decomposition $\Psi=\sum_{\mu=1}^k\Psi_\mu$, where every $\Psi_\mu$ can be written as a normal MPS. 
According to the theory of MPS, the normal bases are linearly independent in the physical space after grouping a finite number of sites~\cite{cirac2021matrix}. 
This is to say for an interval $U$ of sufficiently large length $d$, we have $\calL_U(\Psi)=\oplus_{\mu=1}^k\calL_U(\Psi_\mu)$. 
So $\overline{\calL_d{(\Psi)}}=\oplus_{\mu=1}^k\bbC\Psi_\mu$, which has dimension $k$.
\end{proof}

Finally, we prove the key result that the brickwork circuit of invertible maps preserves the injectivity (normality) of MPS.
\begin{theorem}
Let $\calE: H_\Lambda \mapsto H_\Lambda$ be a translationally invariant invertible local map. Then, the number of normal bases of $\ket{\Psi}$ is invariant under $\calE$. In particular, $\calE$ maps a normal MPS to a normal MPS.
\end{theorem}
\begin{proof} 
According to the previous theorem, the number of normal bases is the dimension of $\overline{\calL_d(\Psi)}$. 
Because $\calE$ preserves the local reconstructibility of the state, $\calE(\overline{\calL_d(\Psi)})$ is closed, namely $\calE(\overline{\calL_d(\Psi)})=\overline{\calL_d(\calE(\overline{\calL_d(\Psi)}))}$. Since $\overline{\calL_d(\Psi)} = \oplus_{\mu=1}^k \bbC \Psi_\mu$, we have $\calE(\overline{\calL_d(\Psi)})=\overline{\calL_d(\calE\ket{\Psi})}$. Thus, we obtain 
\begin{align}
\dim\overline{\calL_d(\calE\ket{\Psi})}=\dim\calE(\overline{\calL_d(\Psi)})=\dim\overline{\calL_d(\Psi)},
\end{align} 
which indicates the number of normal bases in $\calE\ket{\Psi}$ and $\ket{\Psi}$ are the same.
\end{proof}

\section{Basics of group cohomology}\label{app:group_cohomology}
Here, we briefly review the basic concepts in group cohomology. 
Given a group $G$, let $\omega_n(g_1, \cdots, g_n)$ be a function of $n$ group elements of $G$ and take the value in the $G$-module denoted by $M$. 
Such a function is called an $n$-cochain, and we use $\calC^n(G,M)$ to denote the space of $n$-cochains. 
We define a series of map $d_n: \calC^n(G, M) \mapsto \calC^{n+1}(G, M)$ as
\begin{align}
    d\omega_n&(g_{1:n+1}) \nonumber \\
    =\, &{}^{g_1}\!\omega_n(g_{2:n+1}) \prod_{i = 1}^n \omega_n^{(-1)^i}(g_{1:i-1},g_ig_{i+1},g_{i+2:n+1}) \nonumber \\
    &\omega_n^{(-1)^{n+1}}(g_{1:n}),\label{eq:derivative_cohomology}
\end{align}
where we use the shorthand notation $g_{1:i}$ to denote a set of group elements $(g_1, g_2,\cdots, g_i)$ in the argument of $\omega_n$.
The map satisfies $d \cdot d = 1$. 
We introduce the following definitions
\begin{equation}
\begin{aligned}
    \calZ^n(G,M) &:= \{\omega_n | d_n \omega_n = 1, \omega_n \in \calC^n(G,M)\}, \\
    \calB^n(G,M) &:= \{\omega_n | \omega_n\!=\! d_{n-1}\omega_{n-1}, \omega_{n-1} \in \calC^{n-1}(G,M)\},
\end{aligned}
\end{equation}
where $\calZ^n(G,M)$ and $\calB^n(G,M)$ are the group of $n$-cocycles and $n$-coboundaries, respectively. 
They are abelian groups and satisfies $\calB^n(G,M) \subset \calZ^n(G,M)$.
The $n$-th cohomology group is defined as $\calH^n(G,M):= \calZ^n(G,M)/\calB^n(G,M)$. 
For an $n$-cocycle $\omega_n \in \calZ^n$, we denote its cohomology class as $[\omega_n]$.

In this paper, we focus on the cases of $n = 2,3$ and consider $M$ being a $G$-module $\U_T(1)$.
The group element $g \in G$ acts on $a \in U_T(1)$ as ${}^g\!a = a$ and $a^*$ for $g$ being unitary and anti-unitary, respectively.
A $2$-cochain $\omega$ is a $2$-cocycle if the following equation [Eq.~\eqref{eq:2-cocycle}] holds for $g_{1,2,3} \in G$,
\begin{align}
    \omega(g_1, g_2)\omega(g_1 g_2, g_3) = \omega(g_1, g_2 g_3) {}^{g_1}\!\omega(g_2,g_3).
\end{align}
The $2$-cochain $\omega$ is a $2$-coboundary if there exists $\varphi_1 \in \calC^1(G, U_T(1))$ such that
\begin{align}
    \omega(g_1, g_2) = \frac{\varphi_1(g_1) {}^{g_1}\!\varphi_1(g_2)}{\varphi_1(g_1g_2)}.
\end{align}

For $n = 3$, the 3-cochain $\omega$ is a 3-cocycle if the following equation [Eq.~\eqref{eq:3-cocycle}] holds
\begin{align}
\frac{\omega(g_1, g_2, g_3) \omega(g_1, g_2 g_3, g_4) {}^{g_1}\!\omega(g_2,g_3,g_4)}{\omega(g_1 g_2, g_3, g_4) \omega(g_1, g_2, g_3 g_4)} = 1.
\end{align}
The 3-cochain $\omega$ is a 3-coboundary if there exists $\varphi_2\in \calC^2(G,U_T(1))$ such that
\begin{align}
    \omega(g_1,g_2,g_3) = \frac{\varphi_2(g_1,g_2g_3){}^{g_1}\!\varphi_2(g_2,g_3)}{\varphi_2(g_1g_2, g_3)\varphi_2(g_1,g_2)},
\end{align}
for $g_{1,2,3,4} \in G$.

\begin{widetext}

\section{Decomposition of cocycles and the K\"unneth formula}\label{app:direct-sum}

The $n$-th cohomology of a unitary symmetry group given by a direct product $G \times H$ takes a direct sum decomposition according to the K\"unneth formula
\begin{align}
    \calH^n(G\times H, \U(1)) = \bigoplus_{k = 0}^n \cH^{n-k}(G, \cH^k(H, \U(1))).
\end{align}
Accordingly, one may decompose the representative $n$-cocycle for the element of $\calH^n(G\times H, U(1))$ as the product $\prod_k \nu_k$, where each $\nu_k$ is a representative cocycle for the element of $\cH^{n-k}(G, \cH^k(H, \U(1)))$.  To facilitate the discussion we define the slant product.
The slant product is a mapping from $n$-cochain to $(n\,{-}\,1)$-cochain defined as 
\begin{align}
    i_g \omega (g_1,...,g_{n-1}) := \prod_{j=0}^{n-1} \omega(g_1,...,g_j, g, g_{j+1},...,g_{n-1})^{(-1)^{n-1+j}}
\end{align}
where $\omega$ is a $n$-cochain. It is easy to check that this is also a mapping from $n$-cocycle to $(n\,{-}\,1)$-cocycle since $d(i_g \omega) = i_g (d \omega)$, and thus if $d \omega = 0$, then $d(i_g \omega) = 0$. 
Accordingly, this is a group homomorphism from $\cH^n$ to $\cH^{n-1}$.

In what follows, we explicitly demonstrate this decomposition for $n\,{=}\,2$ and $3$ using the slant product.

\subsection{n=2}

Using 2-cocycle condition $\omega(a,b) \omega(ab,c) = \omega(b,c) \omega(a,bc)$, we can show that
\begin{align}
    \omega(g_1 h_1, g_2 h_2) &= \frac{\omega(g_1, g_2 h_1 h_2) \omega(h_1, g_2 h_2)}{\omega(g_1,h_1)} =  \frac{  \omega(h_1, g_2 h_2)}{\omega(g_1,h_1)} \frac{\omega(g_1, g_2) \omega(g_1 g_2, h_1 h_2)}{\omega(g_2, h_1 h_2)} \nonumber \\
    &=  \omega(g_1, g_2) \omega(h_1,h_2)   \frac{  \omega(h_1, g_2 h_2)}{\omega(g_1,h_1)} \frac{ \omega(g_1 g_2, h_1 h_2)}{\omega(g_2 h_1, h_2) \omega(g_2, h_1)} \nonumber \\
    &=  \omega(g_1, g_2) \omega(h_1,h_2)   \frac{  \omega(h_1, g_2 h_2)}{\omega(g_1,h_1)} \frac{ \omega(g_1 g_2, h_1 h_2)}{\omega(h_1 g_2, h_2) \omega(g_2, h_1)} \nonumber \\
    &=   \omega(g_1, g_2) \frac{ \omega(h_1, g_2) }{\omega(g_2, h_1)}  \omega(h_1,h_2)  \frac{  \omega(g_1 g_2, h_1 h_2)}{\omega(g_1,h_1) \omega(g_2, h_2) }.
\end{align}
By taking a gauge transformation $V_x \mapsto \varphi(x) V_x$ with $\varphi(gh) := \omega(g,h)^{-1}$, we get
\begin{align}
    \omega \mapsto \omega \cdot \left(\frac{\omega(g_1 g_2, h_1 h_2)}{\omega(g_1, h_1) \omega(g_2, h_2)}\right)^{-1}
\end{align}
and the last bracket can be removed. Here, we used the direct product structure of $G \times H$. Alternatively, this can be derived by considering two equivalent ways to combine the projective representation of $g_1, h_1, g_2, h_2$:
\begin{align}
    V_{g_1} V_{h_1} V_{g_2} V_{h_2} &= \omega(g_1, h_1) \omega(g_2, h_2) V_{g_1 h_1} V_{g_2 h_2} = \omega(g_1, h_1) \omega(g_2, h_2) \omega(g_1 h_1, g_2 h_2) V_{g_1 g_2  h_1 h_2} \nonumber \\
    &= \frac{ \omega(h_1, g_2) }{\omega(g_2, h_1)} V_{g_1} V_{g_2}  V_{h_1}  V_{h_2} = \frac{ \omega(h_1, g_2) }{\omega(g_2, h_1)} \omega(g_1, g_2) \omega(h_1, h_2) \omega(g_1g_2, h_1 h_2) V_{g_1 g_2  h_1 h_2}
\end{align}
Accordingly, $\omega \in \cH^2(G \times H, \U(1))$ is decomposed into three terms, each of which belongs to a certain term in the Kunneth decomposition as the following:
\begin{align} \label{eq:kunH2}
     \nu_0(g_1, g_2) &:= \omega(g_1, g_2),\qquad 
     \nu_1(h_1 | g_2) := \frac{ \omega(h_1, g_2) }{\omega(g_2, h_1)} = i_{g_2} \omega(h_1), \qquad
     \nu_2(h_1, h_2) := \omega(h_1, h_2)
\end{align}
Furthermore, note that the first group cohomology is nothing but the possible homomorphism from group $G$ to $\U(1)$, i.e., abelianization of $G$ since $f(g_1) f(g_2) = f(g_1 g_2) = f(g_2) f(g_1)$ as a $\U(1)$ number. We can show that $\nu_1(h|g) = i_g(h)$  for a fixed $h$ is a proper group homomorphism from $G$ to $\U(1)$, i.e., $i_{g_1}(h) \cdot i_{g_2}(h) = i_{g_1 g_2}(h)$:
\begin{align}
\frac{\omega(g_1,h)}{\omega(h,g_1)}\frac{\omega(g_2,h)}{\omega(h,g_2)} \frac{\omega(h,g_1g_2)}{\omega(g_1g_2,h)} = \frac{\omega(g_2,h)\omega(hg_1,g_2)}{\omega(h,g_2)\omega(g_1,g_2h)} = 1,
\end{align}
and $\nu_1(h|g)$ is a 1-cochain for $H$ on the module $\cH^1(G,\U(1))$, i.e., $[\nu_1(h|g)] \in \cH^1(H,\cH^1(G,\U(1)))$. Therefore, each equivalence class of $[\nu_k]$ is a faithful representation of $\calH^{2-k}(G, \calH^{k}(H, \U(1)))$ for $k = 0,1,2$.

Note that this decomposition can be performed successively. If a given group is $G \times H \times K$, one can first decompose based on \eqnref{eq:kunH2} between $G$ and $H\times K$, and then further apply this decomposition for $\nu_1(h_1 k_1 | g_2)$ and $\nu_2(h_1 k_1, h_2 k_2)$ as follows:
\begin{align}
    \nu_1(h_1 k_1 | g_2) &= \frac{\omega(h_1 k_1, g_2)}{\omega(g_2,h_1 k_1)} = \frac{\omega(h_1, g_2)}{\omega(g_2,h_1)} \frac{\omega(k_1, g_2)}{\omega(g_2,k_1)} = \nu_1(h_1 | g_2) \cdot \nu_1(k_1 | g_2) \nonumber \\
    \nu_2(h_1 k_1, h_2 k_2) &= \omega(h_1 k_1, h_2 k_2) = \nu_0(h_1, h_2) \cdot \nu_1(h_1 | k_2) \cdot \nu_2(k_1, k_2)
\end{align}
where $\nu_1(x|y)= i_y \omega (x)$. In the first line, we used the 2-cocycle consistency equation. Since
\begin{align}
    \cH^2(G\times H \times K,\U(1)) &= \cH^2(G,\U(1)) \oplus \cH^2(K,\U(1)) \oplus \cH^2(H,\U(1)) \oplus \cH^1(G,\cH^1(K,\U(1))) \nonumber \\
    & \quad \oplus \cH^1(G,\cH^1(H,\U(1))) \oplus \cH^1(H,\cH^1(K,\U(1))),
\end{align}
six $\nu_k$s gives the proper decomposition of $\omega$ into the representative elements of the six Kunneth components. 

\subsection{n=3}

Using 3-cocycle condition $\omega(a,b,c) \omega(a,bc,d) \omega(b,c,d) = \omega(ab,c,d) \omega(a,b,cd)$, we can show that
\begin{align} \label{eq:3cocycle_decomposition}
    \omega(g_1 h_1, g_2 h_2, g_3 h_3) &= \prod_{k=0}^3 \nu_k(h_1, ..., h_k | g_{k+1}, ..., g_{3})  \nonumber \\
    \nu_0(g_1, g_2, g_3) &:= \omega(g_1, g_2,g_3)   \in \cH^3(G,\U(1)) \nonumber \\
     \nu_1(h_1 | g_2, g_3) &:= \frac{ \omega(h_1, g_2, g_3) \omega(g_2,g_3,h_1) }{\omega(g_2, h_1,g_3)}  = i_{h_1} \omega(g_2, g_3)   \in \cH^2(G,\cH^1(H,\U(1))) \nonumber \\
     \nu_2(h_1, h_2 | g_3 ) &:= \frac{ \omega(h_1, h_2, g_3) \omega(g_3,h_1,h_2) }{\omega(h_1, g_3, h_2)}  = i_{g_3} \omega(h_1, h_2)    \in \cH^1(G,\cH^2(H,\U(1))) \simeq \cH^2(H,\cH^1(G,\U(1)))  \nonumber \\
     \nu_3(h_1, h_2, h_3) &:= \omega(h_1, h_2, h_3)    \in \cH^0(G,\cH^3(H,\U(1))) \simeq \cH^3(H,\U(1)).
\end{align} 
up to 3-coboundary terms. To show this, let us define several coboundary terms as follows
\begin{align}
    \varphi_1(g_1 h_1, g_2 h_2) &:= \omega(g_1, h_1, g_2 h_2),  \quad 
    d \varphi_1(g_1 h_1, g_2 h_2, g_3 h_3)
    = \frac{ \omega(g_1 g_2, h_1 h_2, g_3 h_3) \omega(g_1, h_1, g_2 h_2)}{\omega(g_1, h_1, g_2 g_3 h_2 h_3) \omega(g_2, h_2, g_3 h_3)}
    \nonumber \\
    \varphi_2(g_1 h_1, g_2 h_2) &:= \omega(g_1, g_2, h_1 h_2)^{-1},  \quad 
    d \varphi_2(g_1 h_1, g_2 h_2, g_3 h_3)
    = \frac{ \omega(g_1 g_2, g_3, h_1 h_2 h_3) \omega(g_1, g_2, h_1 h_2)}{\omega(g_1, g_2 g_3, h_1 h_2 h_3) \omega(g_2, g_3, h_2 h_3)} \nonumber \\
    \varphi_3(g_1 h_1, g_2 h_2) &:= \omega(h_1, g_2, h_2)^{-1},  \quad 
    d \varphi_3(g_1 h_1, g_2 h_2, g_3 h_3)
    = \frac{ \omega(h_2, g_3, h_3) \omega(h_1, g_2 g_3, h_2 h_3)}{
        \omega(h_1, g_2, h_3) \omega(h_1 h_2, g_3, h_3)
    } \nonumber \\
    \varphi_4(g_1 h_1, g_2 h_2) &:= \omega(g_2, h_1, h_2),  \quad 
    d \varphi_4(g_1 h_1, g_2 h_2, g_3 h_3)
    = \frac{ \omega(g_2, h_1, h_2) \omega(g_3, h_1 h_2, h_3) }{ \omega(g_3, h_2, h_3) \omega(g_2 g_3, h_1, h_2 h_3) }.
\end{align}
Then, we can use the 3-cocycle condition successively to explicitly show that
\begin{align}
    \omega(g_1h_1, g_2 h_2, g_3 h_3) &= \qty[ \prod_{k=0}^3 \nu_k(h_1,...,h_k | g_{k+1},..., g_3) ] \cdot d \varphi_1 d \varphi_2 d \varphi_3 d \varphi_4
\end{align}
Furthermore, note that for fixed $(g_2, g_3)$:
\begin{align}
    \frac{i_{h_1} \omega(g_2, g_3) \cdot i_{h_2} \omega(g_2, g_3)}{i_{h_1 h_2} \omega(g_2, g_3)} = 1  
\end{align}
and $\nu_1(h_1 | g_2, g_3)\,{=}\,i_{h_1} \omega(g_2, g_3)$ is a 2-cochain for $G$ (since $d i_g \omega = i_g d \omega = 0$) on the module $\cH^1(H,\U(1))$, i.e., $[\nu_1(h_1 | g_2, g_3)] \in \cH^2(G,\cH^1(H,\U(1)))$. Similarly, $\nu_2(h_1, h_2| g_3)\,{=}\,i_{g_3} \omega(h_1, h_2)$ is a 2-cochain for $H$ on the module $\cH^1(G,\U(1))$. Therefore, each equivalence class of $[\nu_k]$ is a faithful representation of $\calH^{3-k}(G, \calH^{k}(H, \U(1)))$ for $k = 0,1,2,3$. Again, this decomposition can be applied iteratively for a direct product of more than two groups.

\end{widetext}

\section{Absence of mixed anomaly between $\bbZ_2^\bbH$ and $\bbG_p$}\label{app:no_H_mixed_anomaly}

In this section, we show the absence of mixed anomaly between the Hermiticity symmetry $\bbZ_2^\bbH$ and the physical symmetry $\bbG_p$ in strong injective MPDO in 1d and strong semi-injective TNDO in 2d.

\subsection{One dimension}\label{app:no_H_mixed_anomaly_1d}
Section~\ref{sec:no_hermiticity_SPT} shows the absence of non-trivial SPT phase protected by $\bbZ_2^\bbH$, i.e. $\omega(\bbH, \bbH) = 1$.
Here, we show that, by redefining the projective representation, one can trivialize the 2-cocycle involving the physical symmetry $\sfg \in \bbG_p$ and the Hermitian conjugate $\bbH$, i.e. $\omega(\sfg, \bbH) = \omega(\bbH, \sfg) = 1$.
This indicates the absence of mixed anomaly between $\bbZ_2^\bbH$ and $\bbG_p$.

To begin, the projective representation satisfies
\begin{align}
    V_\sfg V_\bbH = \omega(\sfg, \bbH) V_{\sfg\bbH} = \frac{\omega(\sfg, \bbH)}{\omega(\bbH, \bar{\sfg})} V_\bbH V^*_{\bar{\sfg}},
\end{align}
where $\bar{\sfg} = \bbH \sfg \bbH$.
We consider a gauge transformation (i.e. redefinition) of the projective representation,
\begin{align}
    V_{\sfg} \mapsto  \varphi(\sfg) V_\sfg, \quad V_{\sfg\bbH} \mapsto \varphi(\sfg\bbH) V_{\sfg\bbH}.
\end{align}
The gauge transformation adds 2-coboundary to the 2-cocycle
\begin{align}
    \omega(\sfg, \bbH) \mapsto \omega(\sfg, \bbH) \frac{\varphi(\sfg\bbH)}{\varphi(\sfg)\varphi(\bbH)}.
\end{align}
By picking a particular gauge
\begin{equation}
\begin{gathered}
    \varphi(\sfg) = \sqrt{\frac{\omega(\sfg, \bbH)}{\omega(\bbH, \bar{\sfg})}}, \quad \varphi(\bar{\sfg}) = \sqrt{\frac{\omega(\bar{\sfg}, \bbH)}{\omega(\bbH, \sfg)}}, \\
    \varphi(\sfg\bbH) = \frac{1}{\sqrt{\omega(\sfg, \bbH)\omega(\bbH, \bar{\sfg})}}, \quad \varphi(\bbH) = 1,
\end{gathered}\label{eq:coboundary_hermiticity}
\end{equation}
one obtain a trivial 2-cocycle $\omega(\sfg, \bbH) = \omega(\bbH, \sfg) = 1$.
Here, we use the relation
\begin{equation}
\begin{aligned}
V_\bbH^{-1} V_\sfg V_\bbH (V^*_{\bar{\sfg}})^{-1} &= V_{\bbH}^*V_{\sfg}(V_\bbH^*)^{-1}(V_{\bar{\sfg}}^*)^{-1} \\
\Leftrightarrow\quad \frac{\omega(\sfg, \bbH)}{\omega(\bbH, \bar{\sfg})} &= \frac{\omega(\bar{\sfg}, \bbH)}{\omega(\bbH, \sfg)},
\end{aligned}
\end{equation}
knowing that $V_\bbH^*V_\bbH = \mathds{1}$.
This result indicates the absence of mixed anomaly.

In the analysis of SPT phases protected by the physical symmetry $\bbG_p$, we work with the gauge choice with $\omega(\sfg,\bbH) = \omega(\bbH, \sfg) = 1$.
We note that such choices are not unique.
In this partially fixed gauge, we have (Eq.~\eqref{eq:2-cocycle-hermiticity_constraint})
\begin{align}
    \omega(\sfg_1, \sfg_2) = \omega^*(\bar{\sfg}_1, \bar{\sfg}_2),
\end{align}
which one can show using the 2-cocycle condition in Eq.~\eqref{eq:2-cocycle}.

\subsection{Two dimension}\label{app:no_H_mixed_anomaly_2d}
The reduction operators in 2d TNDO are defined up to a gauge transformation, 
\begin{equation}
P_{g_1,g_2} \mapsto \varphi(g_1,g_2) P_{g_1,g_2}, \text{ for } g_1,g_2\in \bbG_p\rtimes \bbZ_2^\bbH.
\end{equation}
where $\varphi(g_1,g_2,g_3)$ is a $\U(1)$ phase.
Accordingly, the 3-cocycle transforms as
\begin{align}
\omega(g_1,g_2,g_3) &\mapsto \omega(g_1,g_2,g_3) d\varphi(g_1,g_2,g_3) \nonumber \\
&= \omega(g_1,g_2,g_3) \frac{\varphi(g_1,g_2)\varphi(g_1g_2,g_3)}{\varphi(g_1,g_2g_3){}^{g_1}\!\varphi(g_2,g_3)},
\end{align}
where $d\varphi$ is the derivative of the 2-cochain $\varphi$ defined in Eq.~\eqref{eq:derivative_cohomology}.
Here, we show the existence of a gauge choice such that the 3-cocycle is trivial if one element is the Hermitian conjugation.
We further derive the constraint on 3-cocycle in this gauge due to the Hermiticity symmetry.

First, the anti-unitary $\bbZ_2^\bbH$ cannot protect non-trivial SPT phases in 2D~\cite{chen2013symmetry}. 
Specifically, $\omega(\bbH, \bbH, \bbH)$ transforms as
\begin{align}
\omega(\bbH, \bbH, \bbH) &\mapsto \omega(\bbH, \bbH, \bbH) d\varphi(\bbH, \bbH, \bbH) \nonumber \\
&= \omega(\bbH, \bbH, \bbH) \varphi(\bbH, \bbH)^2.
\end{align}
By picking a gauge $\varphi_1(\bbH,\bbH) = 1/\sqrt{\omega(\bbH,\bbH,\bbH)}$, one obtains $\omega(\bbH, \bbH, \bbH) = 1$, indicating the absence of SPT phases protected by the anti-unitary $\bbZ_2^\bbH$.

Next, we work with the gauge choice $\omega(\bbH,\bbH,\bbH)=1$ and discuss the possibility to set $\omega(\sfg,\bbH,\bbH)=\omega(\bbH,\sfg,\bbH)=\omega(\bbH,\bbH,\sfg)=1$ for $\sfg \in \bbG_p$ by further fixing the gauge.
In the case $\sfg \neq \bar{\sfg}$, we can choose the gauge
\begin{equation}\begin{aligned}
\varphi_2(\sfg,\bbH) &= \left(\frac{\omega(\bar{\sfg},\bbH,\bbH)\omega(\bbH,\bbH,\bar{\sfg})}{\omega(\bbH,\sfg,\bbH)}\right)^{1/4}, \\
\varphi_2(\bbH,\sfg) &= \left(\frac{\omega(\bar{\sfg},\bbH,\bbH)\omega(\bbH,\bbH,\bar{\sfg})}{\omega(\bbH,\sfg,\bbH)}\right)^{1/4}, \\
\varphi_2(\sfg\bbH,\bbH) &= \frac{1}{\omega(\sfg,\bbH,\bbH)}\left(\frac{\omega(\bbH,\sfg,\bbH)}{\omega(\bar{\sfg},\bbH,\bbH)\omega(\bbH,\bbH,\bar{\sfg})}\right)^{1/4}, \\
\varphi_2(\bbH,\bbH\sfg) &= \omega(\bbH,\bbH,\sfg)\left(\frac{\omega(\bar{\sfg},\bbH,\bbH)\omega(\bbH,\bbH,\bar{\sfg})}{\omega(\bbH,\sfg,\bbH)}\right)^{1/4},
\end{aligned}\end{equation}
which sets $\omega(\sfg,\bbH,\bbH)=\omega(\bbH,\sfg,\bbH)=\omega(\bbH,\bbH,\sfg)=1$ after the transformation $\omega \mapsto \omega \, d\varphi_2$.
Here, we use the relation
\begin{align}
\frac{\omega(\bar{\sfg},\bbH,\bbH)\omega(\bbH,\bbH,\bar{\sfg})}{\omega(\bbH,\sfg,\bbH)} = \frac{\omega(\bbH,\bar{\sfg},\bbH)}{\omega(\sfg,\bbH,\bbH)\omega(\bbH,\bbH,\sfg)},
\end{align}
which is derived directly from the $3$-cocycle condition in Eq.~\eqref{eq:3-cocycle} and $\omega(\bbH, \bbH, \bbH) = 1$.

In the case $\sfg = \bar{\sfg}$, the cocycle can remain non-trivial, representing the mixed anomaly between $\bbD \cong G \times K \in \bbG_p$ and the Hermiticity symmetry, characterized by $\calH^1(\bbD, \calH^2(\bbZ_2^\bbH, U_T(1)))$.
Here, $\bbD$ involves the weak symmetry $K$ and the diagonal symmetry in $G_l\times G_r$, in which each element transforms the ket and bra copy in the same way.
Such an SPT state is the condensate of domain walls of symmetry $\bbD$ decorated by the 1d SPT states protected by $\bbZ_2^\bbH$ symmetry.

Although these phases are allowed based on group cohomology, they cannot exist in strong semi-injective TNDO.
The anomaly in these states is in the nontrivial cocycles associated with the reduction operators for $V_\sfg$ with $\sfg \in \bbD \times \bbZ_2^\bbH$ as in Eq.~\eqref{eq:3rdCohom}.
Crucially, the transfer matrix $\bbT$ (Eq.~\eqref{eq:2d_transfer_matrix}) in strong semi-injective TNDO is invariant under the symmetry $\bbD$ and transforms to its complex conjugate $\bbT^*$ under $\bbZ_2^\bbH$.
Hence, the eigenstate $|\phi[C])$ associated with the unique largest eigenvalue of $\bbT$ is invariant under $V_\sfg$ for $\sfg \in \bbD$ and transforms to $|\phi[C^*])$ under $V_\bbH$.
Similar to the discussion leading to Eq.~\eqref{eq:2d_weak_symmetry_coboundary}, this indicates that the 3-cocycle is coboundary, i.e.
\begin{align}
\omega(\sfg_1, \sfg_2, \sfg_3) = \frac{\varphi(\sfg_1,\sfg_2\sfg_3){}^{\sfg_1}\!\varphi(\sfg_2,\sfg_3)}{\varphi(\sfg_1\sfg_2,\sfg_3)\varphi(\sfg_1,\sfg_2)}
\end{align}
for elements $\sfg_{1,2,3} \in \bbD \times \bbZ_2^\bbH$.
Thus, one can always pick a gauge such that  $\omega(\sfg,\bbH,\bbH)=\omega(\bbH,\sfg,\bbH)=\omega(\bbH,\bbH,\sfg)=1$ in strong semi-injective TNDOs.

Last, we show the existence of a gauge choice to trivialize the 3-cocycle involving the Hermitian conjugate.
We work with the partially fixed gauge such that $\omega(\bbH,\bbH,\bbH) = \omega(\sfg,\bbH,\bbH) = \omega(\bbH,\sfg,\bbH) = \omega(\bbH,\bbH,\sfg) = 1$, which requires the gauge transformation to satisfy
\begin{equation}
\begin{gathered}
\varphi_3(\bbH,\bbH)^2 = 1, \quad
\frac{\varphi_3(\sfg\bbH,\bbH)\varphi_3(\sfg,\bbH)}{\varphi_3(\bbH,\bbH)} = 1, \\
\frac{\varphi_3(\bbH\sfg,\bbH)\varphi_3(\bbH,\sfg)\varphi_3(\sfg,\bbH)}{\varphi_3(\bbH,\sfg\bbH)} = 1, \quad
\frac{\varphi_3(\bbH,\bbH)\varphi_3(\bbH,\sfg)}{\varphi_3(\bbH,\bbH\sfg)} = 1.
\end{gathered}
\end{equation}
One can obtain $\omega(\sfg_1,\sfg_2,\bbH) = \omega(\sfg_1,\bbH,\sfg_2) = \omega(\bbH,\sfg_1,\sfg_2) = 1$ for $\sfg_1,\sfg_2\in \bbG_p$ by choosing the following gauge
\begin{equation}
\begin{aligned}
\varphi_3(\bbH, \bbH) &=\varphi_3(\bbH,\sfg) = \varphi_3(\sfg,\bbH) = 1,\\
\varphi_3(\sfg_1,\sfg_2) &= \left(\frac{\omega(\bar{\sfg}_1,\bbH, \sfg_2)}{\omega(\bar{\sfg}_1,\bar{\sfg}_2,\bbH)\omega(\bbH,\sfg_1,\sfg_2)}\right)^{1/2}, \\
\varphi_3(\bbH\sfg_1, \sfg_2) &= \frac{1}{\omega(\bbH, \sfg_1,\sfg_2)}\left(\frac{\omega(\bar{\sfg}_1,\bar{\sfg}_2,\bbH)\omega(\bbH,\sfg_1,\sfg_2)}{\omega(\bar{\sfg}_1,\bbH, \sfg_2)}\right)^{1/2}, \\
\varphi_3(\bar{\sfg}_1,\bar{\sfg}_2\bbH) &= \omega(\bar{\sfg}_1,\bar{\sfg}_2,\bbH)\left(\frac{\omega(\bar{\sfg}_1,\bbH, \sfg_2)}{\omega(\bar{\sfg}_1,\bar{\sfg}_2,\bbH)\omega(\bbH,\sfg_1,\sfg_2)}\right)^{1/2}.
\end{aligned}
\end{equation}
Here, we use the relation
\begin{align}
\frac{\omega(\bar{\sfg}_1,\bbH, \sfg_2)}{\omega(\bar{\sfg}_1,\bar{\sfg}_2,\bbH)\omega(\bbH,\sfg_1,\sfg_2)}=\frac{\omega(\sfg_1,\bbH, \bar{\sfg}_2)}{\omega(\sfg_1,\sfg_2,\bbH)\omega(\bbH,\bar{\sfg}_1,\bar{\sfg}_2)},
\end{align}
which can be derived from the 3-cocycle condition in Eq.~\eqref{eq:3-cocycle}.
Thus, we show the existence of a gauge choice to trivialize the cocycle involving Hermiticity symmetry, indicating the absence of mixed anomaly between $\bbZ_2^\bbH$ and $\bbG_p$.

When classifying the SPT phases protected by the physical symmetry, we work with the gauge choice such that the 3-cocycle involving Hermitian conjugate is trivial. 
The Hermiticity symmetry further imposes a constraint on the 3-cocycle (Eq.~\eqref{eq:3-cocycle-hermiticity_constraint})
\begin{align}
    \omega(\sfg_1,\sfg_2,\sfg_3) = \omega^*(\bar{\sfg}_1, \bar{\sfg}_2, \bar{\sfg}_3),
\end{align}
which is again derived from the 3-cocycle condition in Eq.~\eqref{eq:3-cocycle}.

\end{document}